\documentclass[11pt,a4paper]{article}
\pdfoutput=1
\usepackage{indentfirst}
\usepackage{hyperref}
\usepackage{fullpage}
\usepackage{amssymb}
\usepackage{mathtools}
\usepackage{amsthm}
\usepackage{graphics}
\usepackage{eucal}
\usepackage{dsfont}
\usepackage{framed}
\usepackage[T1]{fontenc}
\usepackage{lmodern}
\usepackage[stretch=10,shrink=10]{microtype}
\usepackage[capitalise]{cleveref}
\usepackage{color,soul}
\usepackage{CJK}
\usepackage[thicklines]{cancel}
\usepackage{longtable}
\usepackage{mathrsfs}
\usepackage{commutative-diagrams}
\allowdisplaybreaks[1]

\def\A{\CMcal{A}}
\def\B{\CMcal{B}}

\def\G{\CMcal{G}}
\def\H{\CMcal{H}}

\def\L{\CMcal{L}}
\def\M{\CMcal{M}}

\def\P{\CMcal{P}}
\def\Q{\CMcal{Q}}
\def\R{\CMcal{R}}
\def\S{\CMcal{S}}
\def\T{\CMcal{T}}

\def\X{\CMcal{X}}

\def\sA{\ensuremath{\mathscr{A}}}
\def\sB{\ensuremath{\mathscr{B}}}
\def\sP{\ensuremath{\mathscr{P}}}
\def\sQ{\ensuremath{\mathscr{Q}}}
\def\sR{\ensuremath{\mathscr{R}}}
\def\sS{\ensuremath{\mathscr{S}}}
\def\sT{\ensuremath{\mathscr{T}}}

\def\atildea{\ensuremath{\widetilde{\A_a}}}
\def\btildeb{\ensuremath{\widetilde{\B_b}}}
\def\ptildep{\ensuremath{\widetilde{\P_p}}}
\def\qtildeq{\ensuremath{\widetilde{\Q_q}}}
\def\rtilder{\ensuremath{\widetilde{\R_r}}}

\theoremstyle{plain}
\newtheorem{theorem}{Theorem}[section]
\newtheorem{lemma}[theorem]{Lemma}

\theoremstyle{definition}
\newtheorem{definition}[theorem]{Definition}
\newtheorem{claim}[theorem]{Claim}

\newtheorem{remark}[theorem]{Remark}
\newtheorem{fact}[theorem]{Fact}

\newtheorem{setup}[theorem]{Setup}
\newtheorem{convention}[theorem]{Convention}

\DeclareMathAlphabet{\mathpzc}{OT1}{pzc}{m}{it}
\newcommand {\minusspace} {\: \! \!}

\newcommand {\Fn} [2] {\ensuremath{ #1 \minusspace \Br{ #2 } }}
\newcommand{\reals}{{\mathbb R}}
\newcommand {\set} [1] {\ensuremath{ \left\lbrace #1 \right\rbrace }}

\newcommand {\br} [1] {\ensuremath{ \left( #1 \right) }}
\newcommand {\Br} [1] {\ensuremath{ \left[ #1 \right] }}

\newcommand {\norm} [1] {\ensuremath{ \left\| #1 \right\| }}

\newcommand {\normsub} [2] {\ensuremath{ \norm{#1}_{#2} }}
\newcommand {\onenorm} [1] {\normsub{#1}{1}}
\newcommand {\twonorm} [1] {\normsub{#1}{2}}
\newcommand {\abs} [1] {\ensuremath{ \left| #1 \right| }}

\newcommand {\bra} [1] {\ensuremath{ \left\langle #1 \right| }}
\newcommand {\ket} [1] {\ensuremath{ \left| #1 \right\rangle }}
\newcommand {\ketbratwo} [2] {\ensuremath{ \left| #1 \middle\rangle \middle\langle #2 \right| }}
\newcommand {\ketbra} [1] {\ketbratwo{#1}{#1}}

\newcommand {\innerproduct} [2] {\ensuremath{\left \langle #1 , #2 \right \rangle}}
\newcommand{\nnorm}[1]{{\left\vert\kern-0.25ex\left\vert\kern-0.25ex\left\vert #1
		\right\vert\kern-0.25ex\right\vert\kern-0.25ex\right\vert}}

\DeclareMathOperator*{\bigE}{\mathbb{E}}
\newcommand {\expec} [2] {\Fn{\bigE_{\substack{#1}}}{#2}}

\newcommand {\Tr} {\ensuremath{ \mathrm{Tr} }}

\newcommand {\id} {\ensuremath{\mathds{1}}}

\renewcommand{\dim}[1]{\ensuremath{\mathsf{#1}}}

\newcommand{\SKL}{State Key Laboratory for Novel Software Technology}
\newcommand{\NJU}{ Nanjing University}
\newcommand{\HNL}{Hefei National  Laboratory, Hefei 230088, China}
\newcommand{\pinv}[1]{\ensuremath{#1^+}}
\newcommand{\pos}[1]{\ensuremath{#1^{\mathpzc{pos}}}}
\newcommand {\email} [1] {\href{mailto:#1}{\texttt{#1}}}

\newcommand {\mytitle} {TBD}
\newcommand{\e}{\mathrm{e}}
\renewcommand{\d}{\mathrm{d}}
\newcommand{\influence}{\mathrm{Inf}}
\newcommand{\poly}[1]{\mathrm{poly}\br{#1}}
\newcommand{\choi}[1]{J\br{#1}}
\newcommand{\srange}[1]{\ensuremath{s\in\Br{\dim{s}^2}_{\geq0}^{#1}}}
\newcommand{\trange}[1]{\ensuremath{t\in\Br{\dim{t}^2}_{\geq0}^{#1}}}
\newcommand{\arange}{\ensuremath{a\in\Br{\dim{a}^2}_{\geq0}}}
\newcommand{\brange}{\ensuremath{b\in\Br{\dim{b}^2}_{\geq0}}}
\newcommand{\prange}{\ensuremath{p\in\Br{\dim{p}^2}_{\geq0}}}
\newcommand{\qrange}{\ensuremath{q\in\Br{\dim{q}^2}_{\geq0}}}
\newcommand{\rrange}{\ensuremath{r\in\Br{\dim{r}^2}_{\geq0}}}
\newcommand{\gsnuma}{2\br{\dim{s}^2-1}\br{n-h}}
\newcommand{\gsnumb}{2\br{\dim{t}^2-1}\br{n-h}}
\newcommand{\posint}{\mathbb{Z}_{>0}}
\newcommand {\var} [2] {\Fn{\mathrm{Var}_{#1}}{#2}}
\newcommand{\ipfootnote}{\footnote{The denominator is because of the demoninator in the definition of the inner product $\frac{1}{\dim{s}}\Tr~P^\dagger Q.$}}
\newcommand{\remindfootnote}{\footnote{Remind that $\atildea=\A_a/\sqrt{\dim{a}}$, $\btildeb=\B_b/\sqrt{\dim{b}}$ and $\rtilder=\R_r/\sqrt{\dim{r}}$.}}
\newcommand{\addsetup}{Given \cref{setup}, }

\newcommand{\oneshared}{\ensuremath{\psi^{\sS\sT}}}
\newcommand{\shared}[1]{\ensuremath{\br{\psi^{\sS\sT}}^{\otimes #1}}}
\newcommand{\abrshared}{\ensuremath{\phi_{\textsf{in}}}}
\newcommand{\hspa}[1]{\ensuremath{\H_{\sS^{#1}\sP\sA}}}

\newcommand{\htqb}[1]{\ensuremath{\H_{\sT^{#1}\sQ\sB}}}

\newcommand{\alice}{\ensuremath{\Phi}_{\textsf{Alice}}}
\newcommand{\bob}{\ensuremath{\Phi}_{\textsf{Bob}}}

\hypersetup{
	pdfstartview={FitH},
	pdfdisplaydoctitle={true},
	breaklinks={true},
	bookmarksopen={true},
	bookmarksnumbered={false},
	pdftitle={\mytitle},
}
\usepackage{stmaryrd}
\usepackage{color}

\begin{document}

\title{Decidability of fully quantum nonlocal games with noisy maximally entangled states}
\author{Minglong Qin\thanks{\SKL, \NJU
  (\email{mf1833054@smail.nju.edu.cn} )}
\and Penghui Yao\thanks{\SKL, \NJU(\email{pyao@nju.edu.cn})}~\thanks{\HNL }}
\date{}

		\maketitle

		\thispagestyle{empty}
		\begin{abstract}
This paper considers the decidability of fully quantum nonlocal games with noisy maximally entangled states. Fully quantum nonlocal games are a generalization of nonlocal games, where both questions and answers are quantum and the referee performs a binary POVM measurement to decide whether they win the game after receiving the quantum answers from the players. The quantum value of a fully quantum nonlocal game is the supremum of the probability that they win the game, where the supremum is taken over all the possible entangled states shared between the players and all the valid quantum operations performed by the players. The seminal work $\mathrm{MIP}^*=\mathrm{RE}$~\cite{JNVWY'20,JNVWYuen'20} implies that it is undecidable to approximate the quantum value of a fully nonlocal game. This still holds even if the players are only allowed to share (arbitrarily many copies of) maximally entangled states. This paper investigates the case that the shared maximally entangled states are noisy. We prove that there is a computable upper bound on the copies of noisy maximally entangled states for the players to win a fully quantum nonlocal game with a probability arbitrarily close to the quantum value. This implies that it is decidable to approximate the quantum values of these games. Hence, the hardness of approximating the quantum value of a fully quantum nonlocal game is not robust against the noise in the shared states.

This paper is built on the framework for the decidability of non-interactive simulations of joint distributions~\cite{7782969,doi:10.1137/1.9781611975031.174,Ghazi:2018:DRP:3235586.3235614} and generalizes the analogous result for nonlocal games in~\cite{qin2021nonlocal}. We extend the theory of Fourier analysis to the space of super-operators and prove several key results including an invariance principle and a dimension reduction for super-operators. These results are interesting in their own right and are believed to have further applications.

\end{abstract}

\section{Introduction}

{\em Nonlocal games} are a core model in the theory of quantum computing, which has found wide applications in quantum complexity theory, quantum cryptography, and the foundation of quantum mechanics. A nonlocal game is executed by three parties, a referee and two non-communicating players, which are usually named Alice and Bob. Before the game starts, the players may share an arbitrary bipartite quantum state. The referee samples a pair of questions and sends each of them to the players, separately. Each player is supposed to reply with a classical answer to the referee. They win the game if the questions and the answers satisfy a given predicate. The distribution of the questions and the predicate is known to the players. The {\em quantum value} is the supremum of the probability that the players win the game. It is a central topic in quantum computing to understand the computational complexity of computing the quantum value of a nonlocal game. After decades of efforts~\cite{Cleve:2004:CLN:1009378.1009560,Kempe:2008:EGH:1470582.1470594,doi:10.1137/090772885,Ito:2012:MIP:2417500.2417883,Ji:2016:CVQ:2897518.2897634,8948641,FJVYuen:2019}, it has been finally settled by the seminal work $\mathrm{MIP}^*=\mathrm{RE}$~\cite{JNVWY'20,JNVWYuen'20}, where Ji, Natarajan, Vidick, Wright and Yuen proved that it is undecidable to approximately compute the quantum value of a nonlocal game with constant precision. This result implies that there is no computable upper bound on the preshared entanglement for the players to win the game with a probability close to the quantum value. Otherwise, the probability of success can be obtained by $\varepsilon$-netting all possible quantum strategies and brute-force searching for the optimal value. Ji et al. essentially proved that it is still uncomputable even if the players are only allowed to share (arbitrarily many) EPR states.

In~\cite{qin2021nonlocal}, the authors investigated the robustness of the hardness of the nonlocal games under noise. More specifically, they considered a variant of nonlocal games, where the preshared quantum states are corrupted. It is shown that the quantum value of a nonlocal game is computable if the players are allowed to share arbitrarily many copies of {\em noisy maximally entangled states} (MES). Hence, the hardness of the nonlocal games collapses in the presence of noise from the preshared entangled states.

In this paper, we consider {\em fully quantum nonlocal games}, which are a broader class of games where both questions and answers are quantum and the predicates are replaced by quantum measurements with binary outcomes: win and loss. More specifically, a fully quantum nonlocal game

\[\mathfrak{G}=\br{\sP,\sQ,\sR,\sA,\sB,\phi_{\textsf{in}}^{\sP\sQ\sR},\set{P_{\textsf{win}}=M^{\sA\sB\sR},P_{\textsf{loss}}=\id-M^{\sA\sB\sR}}}\] consists of a referee and two non-communicating players: Alice and Bob, where $\sP,\sQ,\sR,\sA,\sB$ are quantum systems, $\phi_{\textsf{in}}^{\sP\sQ\sR}$ is a tripartite quantum state in $\sP\otimes\sQ\otimes\sR$ and $\set{P_{\textsf{win}},P_{\textsf{loss}}}$ is a measurement acting on $\sA\otimes\sB\otimes\sR$. Alice, Bob, and the referee share the input state $\phi_{\textsf{in}}^{\sP\sQ\sR}$, where Alice, Bob, and the referee hold $\sP, \sQ, \sR$, respectively, at the beginning of the game. Alice and Bob are supposed to perform quantum operations mapping $\sP$ to $\sA$ and $\sQ$ to $\sB$, and then send the quantum states in $\sA$ and $\sB$ to the referee, respectively. After receiving the quantum messages from the players, the referee performs the POVM measurement $\set{P_{\textsf{win}},P_{\textsf{loss}}}$. Again, the players are allowed to share arbitrary quantum states before the game starts. Both players know the description of $\abrshared$ and the POVM. The quantum value of the game $G$ is the supremum of the probability that the players win the game. The supremum is over all possible preshared quantum states and the quantum operations that can be implemented by both parties. It is not hard to see if $\abrshared=\sum_{x,y}\mu\br{x,y}\ketbra{x}^{\sP}\otimes\ketbra{y}^{\sQ}\otimes\ketbra{xy}^{\sR}$ and both $P_{\textsf{win}}$ and $P_{\textsf{loss}}$ are projectors on computational basis, where $\mu$ is a bipartite distribution, then it boils down to a nonlocal game.

Fully quantum nonlocal games also capture the complexity class of  two-prover one-round quantum multi-prover interactive proof systems $\mathsf{QMIP}(2,1)$. The variants of nonlocal games, where either the questions or the answers are replaced by quantum messages have occurred in much literature~\cite{PhysRevLett.108.200401,cj13-11,10.1145/2799560,chung_et_al:LIPIcs:2015:5072,10.1145/2688073.2688094,PhysRevA.87.032306,PhysRevLett.110.060405,Johnston_2016}. In~\cite{PhysRevLett.108.200401}, Buscemi introduced the so-called semi-quantum nonlocal games, which are nonlocal games with quantum questions and classical answers, and proved that semi-quantum nonlocal games can be used to characterize LOSR (local operations and shared randomness) paradigm. Such games are further used to study the entanglement verification in the subsequent work~\cite{PhysRevA.87.032306,PhysRevLett.110.060405}. In a different context, Regev and Vidick in~\cite{10.1145/2799560} proposed quantum XOR games, where the questions are quantum and the answers are still classical. In~\cite{cj13-11}, Leung, Toner, and Watrous introduced a communication task: coherent state exchange and its analogue in the setting of nonlocal games, where both questions and answers are quantum. In~\cite{10.1145/2688073.2688094}, Fitzsimons and Vidick demonstrated an efficient reduction that transforms a local Hamiltonian into a 5-players nonlocal game allowing classical questions and quantum answers. They showed that approximating the value of this game to a polynomial inverse accuracy is $\mathsf{QMA}$-complete.  In~\cite{chung_et_al:LIPIcs:2015:5072}, Chung, Wu, and Yuen further proved a parallel repetition for nonlocal games where again questions are classical and answers are quantum.

As fully quantum nonlocal games are a generalization of nonlocal games, Ji et al.'s result~\cite{JNVWY'20,JNVWYuen'20} implies that it is also undecidable to approximately compute the quantum value of a fully quantum nonlocal game, even if they are only allowed to share MESs.

 In this paper, we continue the line of research in~\cite{qin2021nonlocal} to investigate whether the hardness of fully quantum nonlocal games can be maintained against the noise. More specifically, we consider the games where the players share arbitrarily many copies of noisy MES's $\psi^{\sS\sT}$. Each $\psi^{\sS\sT}$ is a bipartite state in quantum system $\sS\otimes\sT$, where Alice and Bob hold $\sS$ and $\sT$, respectively.  The value of a  game can be written as
\[\mathrm{val}_Q(\mathfrak{G},\psi)=\lim_{n\rightarrow\infty}\max_{\alice,\bob}\Tr\Br{P_{\textsf{win}}\br{\br{\alice\otimes\bob}\br{\phi_{\textsf{in}}^{\sP\sQ\sR}\otimes\br{\psi^{\sS\sT}}^{\otimes n}}}}.\]
where the maximum is taken over all quantum operations $\alice:\sP\otimes\sS^{\otimes n}\rightarrow\sA$ and $\bob:\sQ\otimes\sT^{\otimes n}\rightarrow\sB$. Noisy MESs were introduced in~\cite{qin2021nonlocal}, which will be defined later. They include depolarized EPR states $(1-\varepsilon)\ketbra{\Psi}+\varepsilon \id/2 \otimes \id/2$, where $\varepsilon>0$ and $\ket{\Psi}=\br{\ket{00}+\ket{11}}/\sqrt{2}$ is an EPR state. \cite{JNVWY'20,JNVWYuen'20} proved that it is undecidable to approximate $\mathrm{val}_Q(\mathfrak{G},\ket{\Psi})$ within constant precision.

\section*{Main results}

In this paper, we prove that it is computable to approximate $\mathrm{val}_Q(\mathfrak{G},\psi)$ within arbitrarily small precision if $\psi$ is a noisy MES.

\begin{theorem}[Main result, informal]
Given integer $m\geq 2$, $\delta\in(0,1)$ and a fully quantum nonlocal game $\mathfrak{G}$, where players are allowed to share arbitrarily many copies $m$-dimensional noisy MESs $\psi$, there exists an explicitly computable bound $D=D\br{\varepsilon,\delta,m,\mathfrak{G}}$ such that it suffices for the players to share $D$ copies of $\psi$ to achieve the winning probability at least $\mathrm{val}_Q(\mathfrak{G},\psi)-\delta$. Thus it is feasible to approximate the quantum value of the game $\br{\mathfrak{G},\psi}$ to arbitrarily precision.
\end{theorem}

As mentioned above, the class of noisy MESs includes $(1-\varepsilon)\ketbra{\Psi}+\varepsilon \id/2 \otimes \id/2$, where $\varepsilon>0$ and $\Psi$ is an EPR state. It is as hard as Halting problem to approximate $\mathrm{val}_Q(\mathfrak{G},\ket{\Psi})$ proved by~\cite{JNVWY'20,JNVWYuen'20}. Therefore, our result implies that the hardness of fully quantum nonlocal games is also not robust against the noise in the preshared states.

This result generalizes~\cite{qin2021nonlocal} where the authors proved that it is feasible to approximate the values when both questions and answers are classical. Both works are built on the series of works for the decidability of {\em non-interactive simulations of joint distributions} ~\cite{7782969,Ghazi:2018:DRP:3235586.3235614,doi:10.1137/1.9781611975031.174}. In the setting of non-interactive simulations of joint distributions, two non-communicating players Alice and Bob are provided a sequence of independent samples $\br{x_1,y_1},\br{x_2,y_2},\ldots$ from a joint distribution $\mu$, where Alice observes $x_1,x_2,\ldots$ and Bob observes $y_1,y_2,\ldots$. The question is to decide what joint distribution $\nu$ Alice and Bob can sample. The research on this problem has a long history and fruitful results (see, for example~\cite{7452414} and the references therein). The quantum analogue was first studied by Delgosha and Beigi~\cite{Delgosha2014}, which is referred to as {\em local state transformation}. The decidability of local state transformation is still widely open. In this work, we prove that the local state transformation is decidable when the source states are noisy MESs.

\subsection{Contributions}

The main contribution in this paper is developing a Fourier-analytic framework for the study of the space of super-operators. Here we list some conceptual or technical contributions, which are believed to be interesting in their own right and have further applications in quantum information science.

\begin{enumerate}
  \item Analysis in the space of super-operators.

  The space of super-operators is difficult to understand in general.  In this paper, we make a crucial observation that the quantum value of a fully quantum nonlocal game can be reformulated in terms of the {\em Choi representations} of the adjoint maps of the quantum operations. Instead of the space of super-operators, we apply Fourier analysis to the space spanned by those Choi representations. Then we prove an invariance principle for super-operators as well as a dimension reduction for quantum operations, which generalize the analogous results in~\cite{qin2021nonlocal}.

  Our understanding of Fourier analysis in the space of super-operators is still very limited, although Boolean analysis has been studied extensively in both mathematics and theoretical computer science for decades. The approach taken in this paper may pave the way for the theory of Fourier analysis in the space of super-operators.

  \item Invariance principle for super-operators.

  The classical invariance principle is a central limit theorem for polynomials~\cite{MosselOdonnell:2010}, which asserts that the distribution of a low-degree and flat polynomial with random inputs uniformly drawn from $\set{\pm 1}^n$ is close to the distribution which is obtained by replacing the inputs with i.i.d. standard normal distributions. Here a polynomial is flat means that no variable has high influence on the value of the polynomial. In~\cite{qin2021nonlocal}, the authors established an invariance principle for matrix spaces. This paper further proves an invariance principle for super-operators. This is essential to reduce the number of shared noisy MESs.

  \item Dimension reduction for quantum operations.

  An important step in our proof is a dimension reduction for quantum operations, which enables us to reduce the dimensions of both players' quantum operations. It, in turn, reduces the number of noisy MESs shared between the players. Dimension reductions for quantum operations are usually difficult and sometimes even impossible~\cite{10.1007/978-3-642-22006-7_8,7426395}. In this paper, we prove a dimension reduction via an invariance principle for super-operators and the dimension reduction for polynomials in Gaussian spaces~\cite{Ghazi:2018:DRP:3235586.3235614}.  we adopt the techniques in~\cite{Ghazi:2018:DRP:3235586.3235614} with a delicate analysis. It leads to an exponential upper bound in the main theorem. which also improves the doubly exponential upper bound in~\cite{qin2021nonlocal}.

\end{enumerate}

\subsection{Comparison with~\cite{qin2021nonlocal} }\label{subsec:comparison}

In~\cite{qin2021nonlocal}, the authors applied Fourier analysis to the Hilbert space where both players' measurements stay, and proved hypercontractive inequalities, quantum invariance principles and dimension reductions for matrices and random matrices. In a fully quantum nonlocal game, both players perform quantum operations. Hence, a natural approach is to further extend the framework in~\cite{7782969,qin2021nonlocal} to the space of super-operators.

The first difficulty occurs as the answers are quantum. In~\cite{qin2021nonlocal}, the authors applied the framework to each pair of  POVM elements (one from Alice and one from Bob). Further taking a union bound, the result concludes. Hence, it suffices to work on the space where the POVM elements stay, which is a tensor product of identical Hilbert spaces. This approach fails when considering fully quantum nonlocal games as the answers are quantum. Hence, we need to have a convenient representation of super-operators to work on. It is known that there are several equivalent representations of super-operators~\cite{Wat08}. In this paper, we choose the Choi representations of super-operators, which view a super-operator as an operator in the tensor product of the input space and the output space. Hence, the underlying Hilbert space is a tensor product of a number of identical Hilbert spaces and the output Hilbert space.  Thus, the analysis in~\cite{qin2021nonlocal} cannot be generalized here directly.

The second difficulty occurs as the questions are quantum.  In~\cite{qin2021nonlocal}, the authors essentially proved an upper bound on the number of noisy MESs for each pair of inputs. If the precision of the approximation is good enough, then we can obtain an upper bound for all inputs  again by a union bound because the questions are finite in a nonlocal game. This argument cannot be directly generalized to fully nonlocal games as the questions are the marginal state of the input state with Alice and Bob. Fortunately, this difficulty can be avoided as the input state is in a bounded-dimensional space and thus it suffices to prove the theorem for each basis element from a properly chosen basis in the space, and then take a union bound.

The last difficulty is that the rounding argument in~\cite{qin2021nonlocal} does not apply to fully quantum nonlocal games. In the end of the construction, the new super-operators are no longer valid quantum operations. In~\cite{qin2021nonlocal}, the construction gives a number of Hermitian operators in the end. The rounding argument proves that it is possible to round these Hermitian operators to valid POVMs with small deviation. For fully quantum nonlocal games we need a new rounding argument which is able to round super-operators to valid quantum operations with small deviation in the end of the construction.

%

\subsection{Proof overview}

The proof is built on the framework in~\cite{7782969,Ghazi:2018:DRP:3235586.3235614,doi:10.1137/1.9781611975031.174} for the decidability of non-interactive simulation of joint distributions. To explain the high-level idea of our proof, we start with the decidability of a particular task of local state transformation. Then we explain how to generalize it to nonlocal games.

\subsection*{Local state transformation}

We are interested in the decidability of the following local state transformation problem.
\begin{framed}
\noindent Given $\delta>0$, a bipartite state $\sigma$ and a noisy MES $\psi$, suppose Alice and Bob share arbitrarily many copies of $\psi$.

\begin{itemize}
\item\textbf{Yes}. Alice and Bob are able to jointly generate a bipartite state $\sigma'$ using only local operations such that $\sigma'$ is $\delta$-close to $\sigma$, i.e., $\onenorm{\sigma-\sigma'}\leq \delta$.

\item\textbf{No}. Any quantum state $\sigma'$ that Alice and Bob can jointly generate using only local operations is $2\delta$-far from $\sigma$, i.e., $\onenorm{\sigma-\sigma'}\geq 2\delta$.
\end{itemize}
\end{framed}

As there is no upper bound on the number of copies of $\psi$, the decidability of this question is unclear.  If it were proved that any quantum operation could be simulated by a quantum operation in a bounded dimension, then the problem would be decidable as we could search all possible quantum operations in a bounded-dimensional space via an $\varepsilon$-net and brute force. More specifically, suppose Alice and Bob share $n$ copies of noisy MESs $\psi$ and they perform quantum operations $\alice$ and $\bob$. For any precision parameter $\delta\in(0,1)$, we need to construct quantum operations $\widetilde{\alice}$ and $\widetilde{\bob}$ acting on $D$ copies of $\psi$, where $D$ is independent of $n$, such that

\begin{equation}\label{eqn:1st}
  \br{\alice\otimes\bob}\br{\psi^{\otimes n}}\approx\br{\widetilde{\alice}\otimes\widetilde{\bob}}\br{\psi^{\otimes D}}.
\end{equation}

To explain the high-level ideas, we assume that $\psi$ is a $2$-qubit quantum state for simplicity. Let $\set{\X_{a}}_{a\in\set{0,1,2,3}}$ be an orthonormal basis in the space of $2\times 2$ matrices.
We observe that the left hand side of Eq.~\eqref{eqn:1st} is determined by the following $4^{2n}$ values:
\[\set{\Tr\Br{\br{\X_a\otimes\X_{b}}\br{\br{\alice\otimes\bob}\br{\psi^{\otimes n}}}}}_{a,b\in\set{0,1,2,3}^n},\]
where $\X_{a}=\X_{a_1}\otimes\cdots\otimes\X_{a_n}$.
Notice that
\[\Tr\Br{\br{\X_{a}\otimes\X_{b}}\br{\br{\alice\otimes\bob}\br{\psi^{\otimes n}}}}=\Tr\Br{\br{\br{\alice}^*\br{\X_{a}}\otimes\br{\bob}^*\br{\X_{b}}}\br{\psi^{\otimes n}}},\]
where $\br{\alice}^*$ and $\br{\bob}^*$ are the adjoints of $\alice$ and $\bob$, respectively. Hence, Eq.~\eqref{eqn:1st} is equivalent to
\begin{equation}\label{eqn:2nd}
\Tr\Br{\br{\br{\alice}^*\br{\X_{a}}\otimes\br{\bob}^*\br{\X_{b}}}\psi^{\otimes n}}\approx\Tr\Br{\br{\br{\widetilde{\alice}}^*\br{\X_{a}}\otimes\br{\widetilde{\bob}}^*\br{\X_{b}}}\psi^{\otimes D}}.
\end{equation}
Eq.~\eqref{eqn:2nd} resembles the setting considered in~\cite{qin2021nonlocal}.  It is proved in~\cite{qin2021nonlocal} that for any POVM $\set{M_i\otimes N_j}_{i,j}$ acting on $\psi^{\otimes n}$, there exists POVM $\set{M_i'\otimes N_j'}_{i,j}$ acting on $\psi^{\otimes D}$ such that
\[\Tr\Br{\br{M_i\otimes N_j}\psi^{\otimes n}}\approx\Tr\Br{\br{M'_i\otimes N'_j}\psi^{\otimes D}},\]
for all $i,j$. However, $\br{\alice}^*\br{\X_{a}}$ and $\br{\bob}^*\br{\X_{b}}$ are not positive. It is even not clear how to characterize $\br{\alice}^*\br{\X_{a}}$ and $\br{\bob}^*\br{\X_{b}}$ for valid quantum operations $\alice$ and $\bob$. Thus we cannot directly apply the results in~\cite{qin2021nonlocal}. Instead of working on each of $\br{\alice}^*\br{\X_{a}}$ and $\br{\bob}^*\br{\X_{b}}$, we work on the Choi representations $J\br{\br{\alice}^*}$ and $J\br{\br{\bob}^*}$, which include the information of $\br{\alice}^*\br{\X_{a}}$ and $\br{\bob}^*\br{\X_{b}}$ for all $a, b$. One more advantage of Choi representations is that we have a neat characterization of the Choi representations of quantum operations (refer to~\cref{fac:adjointchoi}). Thus it is more convenient to bound the deviations of the intermediate super-operators from valid quantum operations throughout the construction. We consider the Fourier expansions of $J\br{\br{\alice}^*}$ and $J\br{\br{\bob}^*}$, and reduce the dimensions of the super-operators via the framework for the decidability of non-interactive simulations of joint distributions in~\cite{7782969,Ghazi:2018:DRP:3235586.3235614,doi:10.1137/1.9781611975031.174,qin2021nonlocal}. To this end, we prove an invariance principle for super-operators, and combine it with the dimension reduction for polynomials in Gaussian spaces~\cite{Ghazi:2018:DRP:3235586.3235614}. There are several prerequisites for the invariance principle. Firstly, the Choi representation should have low degree. Secondly, all but a constant number of systems are of low influence, that is, all but a constant number of subsystems do not influence the super-operators much.  The construction takes several steps to adjust the Fourier coefficients of $J\br{\br{\alice}^*}$ and $J\br{\br{\bob}^*}$ to meet those prerequisites. Meanwhile, the new super-operators  still need to be close to valid quantum operations so that the value of the game does not change much. Once these prerequisites are satisfied, the basis elements in those subsystems with low influence are replaced by properly chosen Gaussian variables, which only causes a small deviation by the invariance principle.

The whole construction is summarized in~\cref{fig:construction}. Each step is sketched as follows.

\newcommand{\tabincell}[1]{\begin{tabular}{c}#1\end{tabular}}
\newcommand{\subtabtwo}[5]{\cline{2-3}\tabincell{#1}&\multicolumn{1}{|c|}{#2}&\multicolumn{1}{c|}{#2}&\\\cline{2-3}\rule{0pt}{\properheight}&#3&#4&\tabincell{#5}\\}
\newcommand{\subtabone}[5]{\cline{2-3}\tabincell{#1}&\multicolumn{2}{|c|}{#2}&\\\cline{2-3}\rule{0pt}{\properheight}&#3&#4&\tabincell{#5}\\}
\setlength{\arrayrulewidth}{1pt}
\setlength{\tabcolsep}{7pt}
\newcommand{\properheight}{6mm}

\begin{figure}[!htbp]
\centering
\begin{tabular}{cccc}
&$\choi{\br{\Phi_A}^*}$&$\choi{\br{\Phi_B}^*}$&$n$ q. systems\\

\subtabtwo{\textbf{Smoothing}\\\small{\textbf{objective}: bounded deg}}{\cref{lem:smoothing}}{$M^{(1)}$}{$N^{(1)}$}{$n$ q. systems}

\subtabone{\textbf{Regularization}\\\small{bounded deg}\\\small{bounded high inf systems}}{\cref{lem:regularization}}{$M^{(1)}$}{$N^{(1)}$}{$n$ q. systems}

\subtabtwo{\textbf{Invariance principle}\\\small{bounded deg}\\\small{bounded q. systems}\\\small{unbounded Gaussian vars}}{\cref{lem:mainIP}}{$\mathbf{M}^{(2)}$}{$\mathbf{N}^{(2)}$}{$h$ q. systems\\$O(n-h)$ Gaussian vars}

\subtabone{\textbf{Dimension reduction}\\\small{bounded q. systems}\\\small{bounded Gaussian vars}}{\cref{lem:dimensionreduction}}{$\mathbf{M}^{(3)}$}{$\mathbf{N}^{(3)}$}{$h$ q. systems\\ $n_0$ Gaussian vars}

\subtabtwo{\textbf{Smooth}\\\small{bounded q. systems}\\\small{bounded Gaussian vars}\\\small{bounded deg}}{\cref{lem:smoothGaussian}}{$\mathbf{M}^{(4)}$}{$\mathbf{N}^{(4)}$}{$h$ q. systems\\ $n_0$ Gaussian vars}

\subtabtwo{\textbf{Multilinearization}\\\small{bounded q. systems}\\\small{bounded Gaussian vars}\\\small{bounded deg \& multilinear}}{\cref{lem:multiliniearization}}{$\mathbf{M}^{(5)}$}{$\mathbf{N}^{(5)}$}{$h$ q. systems\\ $n_0n_1$ Gaussian vars}

\subtabtwo{\textbf{Invariance principle}\\\small{bounded q. systems}}{\cref{lem:invarianceback}}{$M^{(6)}$}{$N^{(6)}$}{$h+n_0n_1$ q. systems}

\subtabtwo{\textbf{Rounding}\\\small{quantum operations}}{\cref{lem:mainrounding}}{$\choi{\br{\widetilde{\Phi}_A}^*}$}{$\choi{\br{\widetilde{\Phi}_B}^*}$}{$h+n_0n_1$ q. systems}
\end{tabular} \caption{Construction of the transformations}\label{fig:construction}
\end{figure}

\begin{enumerate}
\item \textbf{Smoothing}

Suppose that Alice and Bob perform quantum operations $\alice$ and $\bob$, respectively. Let $\br{\alice}^*,\br{\bob}^*$ be the adjoints of $\alice,\bob$ (defined in \cref{eqn:adjointmap}), respectively, and $\choi{\br{\alice}^*}, \choi{\br{\bob}^*}$ be the corresponding Choi representations (defined in \cref{eqn:choi}). Notice that $\choi{\br{\alice}^*}, \choi{\br{\bob}^*}$ lie in the tensor-product space of the input Hilbert space and the output Hilbert space. The output Hilbert space is bounded-dimensional. And the input Hilbert space is unbounded, of which we aims to reduce the dimension.

This step is aimed to obtain bounded-degree approximations of $\choi{\br{\alice}^*}$ and $ \choi{\br{\bob}^*}$. We apply a noise operator $\Delta_{\gamma}$ for some $\gamma\in(0,1)$ defined in \cref{def:bonamibeckner} to both $\choi{\br{\alice}^*}$ and $ \choi{\br{\bob}^*}$ on the input spaces. Note that both Choi representations are positive operators. After smoothing the operation and truncating the high-degree parts, we get bounded-degree approximations $M^{(1)}$ and $N^{(1)}$, of $\choi{\br{\alice}^*}$ and $ \choi{\br{\bob}^*}$, respectively. Though the bounded-degree approximations may no longer be positive, the deviation can be proved to be small.


\item \textbf{Regularity}
	

	This step is aimed to prove that the number of subsystems having high influence is bounded. The influence of a subsystem of a multipartite Hermitian operator is defined in \cref{def:influenceGaussian}. Informally speaking, the influence measures how much the subsystem can affect the operator. For a bounded operator, the total influence, which is the summation of the influences of all subsystems, is upper bounded by the degree of the operator. This is a generalization of a standard result in Boolean analysis. Note that we have bounded-degree approximations after the first step. The desired result follows by a Markov inequality.
	
	\item \textbf{Invariance principle}
	
	
	In this step, we use correlated Gaussian variables to substitute the basis elements in all the subsystems with low influence in $M^{(1)}$ and $N^{(1)}$, after which we get random operators $\mathbf{M}^{(2)}$ and $\mathbf{N}^{(2)}$, whose Fourier coefficients are low-degree multilinear polynomials in Gaussian variables. We also need to prove that, $\mathbf{M}^{(2)}$ and $\mathbf{N}^{(2)}$ are close to positive operators in expectation.

	\item \textbf{Dimension reduction}

	This step is aimed to reduce the number of Gaussian variables. After applying a dimension reduction to $\mathbf{M}^{(2)}$ and $\mathbf{N}^{(2)}$, we get random operators $\mathbf{M}^{(3)}$ and $\mathbf{N}^{(3)}$ containing a bounded number of Gaussian random variables. Unlike~\cite{qin2021nonlocal}, we get an upper bound independent of the number of quantum subsystems via a more delicate analysis. However, the Fourier coefficients of $\mathbf{M}^{(3)}$ and $\mathbf{N}^{(3)}$ are no longer low-degree polynomials after the dimension reduction.

	\item \textbf{Smooth random operators}
	
%

	The remaining steps are mainly concerned with removing the Gaussian variables. This step is aimed to get low-degree approximations of the Fourier coefficients of $\mathbf{M}^{(3)}$ and $\mathbf{N}^{(3)}$. We apply the Ornstein-Uhlenbeck operator (aka noise operators in Gaussian space, see \cref{def:ornstein}) to the Gaussian variables in $\mathbf{M}^{(3)}$ and $\mathbf{N}^{(3)}$ and truncate the high-degree parts to get $\mathbf{M}^{(4)}$ and $\mathbf{N}^{(4)}$. We should note that the Fourier coefficients of $\mathbf{M}^{(4)}$ and $\mathbf{N}^{(4)}$ are polynomials, but not multilinear.
	
	\item \textbf{Multilinearization}
	
%
%
	
	This step is aimed to get multilinear approximations of the Fourier coefficients of $\mathbf{M}^{(4)}$ and $\mathbf{N}^{(4)}$.  To this end, We apply the multilinearization lemma in~\cite{Ghazi:2018:DRP:3235586.3235614} to get random operators $\mathbf{M}^{(5)}$ and $\mathbf{N}^{(5)}$. Now we are ready to use the invariance principle again to convert random operators back to operators.
	
	\item \textbf{Invariance to operators}

		
In this step we substitute the Gaussian variables with properly chosen basis elements, to get operators $M^{(6)}$ and $N^{(6)}$, which have a bounded number of quantum subsystems. Again, we need to apply a quantum invariance principle to ensure that $M^{(6)}$ and $N^{(6)}$ are close to positive operators.

\item \textbf{Rounding}

We now have operators $M^{(6)}$ and $N^{(6)}$ that are close to positive operators. The last thing to do is to round them to the Choi representations of the adjoints of some quantum operations. After the rounding, the whole construction is done.
\end{enumerate}

\subsection*{From local state transformation to fully quantum nonlocal games}

In the setting of a fully quantum nonlocal game, Alice and Bob share noisy MESs as well as an input state $\phi_{\textsf{in}}$ which may be entangled with the referee.  Thus, it can be reformulated as the following problem.

\begin{framed}

\noindent Given $\delta>0$, a tripartite state $\phi_{\textsf{in}}$, a noisy MES $\psi$ and a tripartite state $\sigma$, suppose Alice, Bob and referee share $\phi_{\textsf{in}}$. Additionally, Alice and Bob also share arbitrarily many copies of $\psi$.

\begin{itemize}
\item \textbf{Yes}. Alice and Bob are able to jointly generate a tripartite state $\sigma'$ among Alice, Bob and the referee using only local operations such that $\sigma'$ is $\delta$-close to $\sigma$, i.e., $\onenorm{\sigma-\sigma'}\leq \delta$.

\item \textbf{No}. Any tripartite state $\sigma'$ among Alice, Bob and the referee that Alice and Bob can jointly generate using only local operations is $2\delta$-far from $\sigma$, i.e., $\onenorm{\sigma-\sigma'}\geq 2\delta$.
\end{itemize}
In both cases, the referee does not perform any quantum operation.
\end{framed}

Suppose the input state $\phi_{\textsf{in}}$ is in the register $\sP\otimes\sQ\otimes\sR$, the target state $\sigma$ is in the register $\sA\otimes\sB\otimes\sR$, and Alice and Bob share $n$ copies of noisy MES's, i.e., $\shared{n}$. They perform quantum operations $\alice$ and $\bob$, respectively, where $\alice:\sP\otimes\sS^{\otimes n}\rightarrow\sA$ and $\bob:\sQ\otimes\sT^{\otimes n}\rightarrow\sB$. For any precision parameter $\delta\in(0,1)$, we need to construct quantum operations $\widetilde{\alice}:\sP\otimes\sS^{\otimes D}$ acting on $D$ copies of $\psi$ together with $\sP$, and $\widetilde{\bob}:\sQ\otimes\sT^{\otimes D}$ acting on $D$  copies of $\psi$ together with $\sQ$, such that

\begin{equation}\label{eqn:1stcoherent}
  \br{\alice\otimes\bob}\br{\phi_{\textsf{in}}\otimes\psi^{\otimes n}}\approx\br{\widetilde{\alice}\otimes\widetilde{\bob}}\br{\phi_{\textsf{in}}\otimes\psi^{\otimes D}},
\end{equation}
where $D$ is independent of $n$. It is illustrated in \cref{fig:localstatetrans}.

\newcommand{\upend}{5.5}
\newcommand{\downend}{-4.5}
\newcommand{\leftend}{-7}
\newcommand{\rightend}{5.5}
\begin{figure}[h]
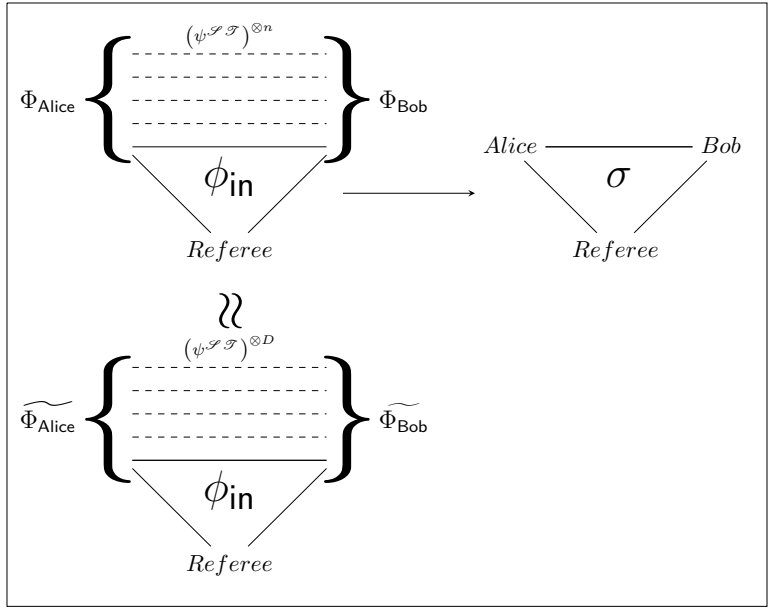

\begin{center}
\scalebox{.8}{
\begin{codi}

\obj{
|(PA)|	&						&|(PB)|	&[3em]			&						&\\[-2em]
|(A)|		&						&|(B)|	&|(A2)|Alice	&						&|(B2)|Bob	\\		
			&|(R)|Referee		&			&					&|(R2)|Referee	&\\[2em]
|(TA)|	&						&|(TB)|	&					&						&\\[-2em]
|(A3)|	&						&|(B3)|	&					&						&\\		
			&|(R3)|Referee	&			&					&						&\\
};

\mor :[swap]A "\raisebox{-1.5em}{\Huge$\phi_{\textsf{in}}$}":- B;
\mor A :- R;
\mor B :- R;
\mor :[shove=+2em] B -> A2;
\mor A2 :- B2;
\mor A2 :- R2;
\mor B2 :- R2;
\mor :[shove=-1.5em,dashed] PA {\Huge \shared{n}}:- PB;
\mor :[shove=+1.5em,dashed] PA - PB;
\mor :[shove=-.5em,dashed] PA - PB;
\mor :[shove=+.5em,dashed] PA - PB;
\mor :[shove=+1.5em,dashed] TA - TB;
\mor :[shove=-1.5em,dashed] TA \shared{D}:- TB;
\mor :[shove=-.5em,dashed] TA - TB;
\mor :[shove=+.5em,dashed] TA - TB;
\mor A3 :- B3;
\mor A3 :- R3;
\mor B3 :- R3;
\mor :[swap]A3 "\raisebox{-1.5em}{\Huge$\phi_{\textsf{in}}$}":- B3;
\mor :[swap]A2 "\raisebox{-1.5em}{\Huge$\sigma$}":- B2;
\node at (-3.3,0.4){\rotatebox{90}{\Huge$\approx$}};
\node at (-5.8,3.9){\raisebox{1em}{\large$\alice$}\fontsize{60pt}{\baselineskip}$\{$};
\node at (-1,3.9){\fontsize{60pt}{\baselineskip}$\}$\raisebox{1em}{\large$\bob$}};
\node at (-5.8,-1.4){\raisebox{1em}{\large$\widetilde{\alice}$}\fontsize{60pt}{\baselineskip}$\{$};
\node at (-1,-1.4){\raisebox{1em}{\fontsize{60pt}{\baselineskip}$\}$\raisebox{1em}{\large$\widetilde{\bob}$}}};
\draw (\leftend,\upend) -- (\leftend,\downend) -- (\rightend,\downend) -- (\rightend,\upend) -- (\leftend,\upend);
\end{codi}
}
\end{center}\caption{Local state transformation}\label{fig:localstatetrans}
\end{figure}

Let $\set{\R_r}_r$ be an orthogonal basis in $\sR$. Then the left-hand side of Eq.~\eqref{eqn:1stcoherent} is determined by the following set of values
\begin{multline*}
  \set{\Tr\Br{\br{\X_a\otimes\X_{b}\otimes\R_r}{\br{\alice\otimes\bob}\br{\phi_{\textsf{in}}\otimes\psi^{\otimes n}}}}}_{a,b,r} \\
  =\set{\Tr\Br{\br{\br{\alice}^*\br{\X_a}\otimes\br{\bob}^*\br{\X_b}\otimes\R_r}\br{\phi_{\textsf{in}}\otimes\psi^{\otimes n}}}}_{a,b,r}
\end{multline*}
By the fact that $\Tr=\Tr\circ\Tr_{\sR}$, we have
\begin{multline}\label{eqn:nonlocal}
  \set{\Tr\Br{\br{\X_a\otimes\X_{b}}{\br{\alice\otimes\bob}\br{\widetilde{\phi_{\textsf{in},r}}\otimes\psi^{\otimes n}}}}}_{a,b,r} \\
  =\set{\Tr\Br{\br{\br{\alice}^*\br{\X_a}\otimes\br{\bob}^*\br{\X_b}\otimes\R_r}\br{\widetilde{\phi_{\textsf{in},r}}\otimes\psi^{\otimes n}}}}_{a,b,r},
\end{multline}
where $\widetilde{\phi_{\textsf{in},r}}=\Tr_{\sR}\br{\id\otimes\R_r}\phi_{\textsf{in}}$. Notice that the difference between Eq.~\eqref{eqn:nonlocal} and Eq.~\eqref{eqn:1st} is that there is an additional operator $\widetilde{\phi_{\textsf{in},r}}$, which is a bounded-dimensional operator, but probably not a quantum state.
We can still work on the Fourier expansions of the Choi representations $J\br{\br{\alice}^*}$ and $J\br{\br{\bob}^*}$. We will show that the framework for local state transformation still works even if there is an additional bounded-dimensional operator $\widetilde{\phi_{\textsf{in},r}}$.

\section{Open problems}

In this work, we prove computable upper bounds on local state transformations with noisy MESs as source states. With some extra work, we further obtain computable upper bounds on the preshared entanglement for fully quantum nonlocal games where the players are only allowed to share noisy MESs. This implies that fully quantum nonlocal games with noisy MESs are decidable. We now list some interesting open problems for future work.

\begin{enumerate}
  \item If we compute the quantum values by $\varepsilon$-netting and searching over all the strategies, then the running time is at least doubly exponential in the size of the games. Can the upper bound on the entanglement or the time complexity be improved? It would be interesting to understand the exact complexity of fully quantum nonlocal games with noisy MESs. From the complexity-theoretic point of view, we may further investigate the complexity classes $\mathsf{MIP}^*_{\psi}$ and $\mathsf{QMIP}_{\psi}$. Here $\mathsf{MIP}^*_{\psi}$ is the set of languages that can be decided by entangled multiprover interactive proof systems, where the provers are only allowed to share arbitrarily many copies of $\psi$ and the provers exchange classical messages with verifiers. If the messages are quantum, then the class is $\mathsf{QMIP}_{\psi}$. What is the exact computational power of these complexity classes for constant-sized states $\psi$? We only know the answer if $\psi$ is an EPR state, for which it is $\textsf{RE}$, or a separable state, for which it is $\textsf{NEXP}$. When $\psi$ is a noisy entangled quantum state, would the computational power increase when the players share more copies of $\psi$? To what extent is the computational power of entangled multiprover interactive proof systems robust against noise?

  \item Local state transformation is one of the most basic quantum communication tasks. Beigi~\cite{Beigi:2013} initiated the study of the decidability of local state transformation and proved several sufficient conditions and necessary conditions~\cite{Beigi:2013,Delgosha2014}. However, to the best of our knowledge, this problem is still widely open. There are many other communication tasks with similar settings, which are also not well understood. For example, distillable entanglement measure and entanglement formation measure are two of the most well-studied entanglement measures~\cite{PhysRevA.53.2046} defined in a similar setting. The only difference in this setting is that classical communication between the players is free, and we aim to optimize the ratio between the number of target states and the number of source states. After decades of efforts, we still don't know how to compute the distillable entanglement measure or the entanglement formation measure of a given state. It is tempting to see whether the framework in this paper could provide new insights into the computability of these quantities.

  \item There are several "tensored" quantities in quantum information theory that are not known how to compute, such as quantum channel capacities~\cite{8242350}, various regularized entanglement measures~\cite{RevModPhys.81.865} and quantum information complexity~\cite{Touchette:2015:QIC:2746539.2746613}. Some of them look extremely simple but turn out to be notoriously hard, such as the quantum channel capacities of depolarizing channels~\cite{8119865}. Can we use the framework in this paper to design algorithms for these quantities?

    \item Given the wide range of applicability of Boolean analysis, Montanaro and Osborne initiated the study of its extensions to quantum setting, where they introduced {\em quantum boolean functions}~\cite{cj10-01}. Several key concepts and results have been successfully generalized to the quantum setting (readers may refer to the introduction in~\cite{RoueWZ:2022} for more details). Some fundamental problems are still open, such as quantum KKL conjecture~\cite{cj10-01,RoueWZ:2022}. Meanwhile, Fourier analysis in quantum settings has found applications in various topics, such as quantum communication complexity~\cite{4690981}, circuit complexity~\cite{BuGJKL:2022}, property testing of unitary operators~\cite{PhysRevA.84.052328}, learning quantum juntas~\cite{ChenNY:2022}, learning quantum dynamics~\cite{RoueWZ:2022}, etc. It is fascinating to see more applications of this growing field in quantum information and quantum computation.
\end{enumerate}

\section*{Acknowledgements}
The authors thank Zhengfeng Ji for helpful discussion. This work was supported by National Natural Science Foundation of China (Grant No. 61972191), Innovation Program for Quantum Science and Technology (Grant No. 2021ZD0302900) and the Program for Innovative Talents and Entrepreneur in Jiangsu.

\section{Preliminary}

For readers' convenience, we list all the notations used in this paper in \cref{app:notations}. Given $n\in\posint$, let $[n]$ and $[n]_{\geq0}$ represent the sets $\set{1,\dots,n}$ and $\set{0,\dots,n-1}$, respectively. For all $a\in\mathbb{Z}^n_{\geq0}$, we define $\abs{a}=\abs{\set{i:a_i>0}}$. In this paper, the lower-cased letters in bold $\mathbf{g},\mathbf{h},\dots$ are reserved for random variables, and the capital letters in bold $\mathbf{M},\mathbf{N}$ are reserved for random operators.
\begin{convention}\label{convention}
   We use $\sA,\sB,\dots$ to represent quantum systems, and the basis in the system is represented by the same letter in the calligraphy font. For instance, the basis in the quantum system $\sA$ is represented by $\set{\A_0,\A_1,\ldots}$. The dimension of quantum systems $\sA, \sB$ are denoted by $\abs{\sA}, \abs{\sB}$, respectively. To keep notations short, the dimension of a quantum system is also represented by the corresponding lower-cased letter in the sans serif font, e.g., the dimensions of quantum systems $\sA,\sB$ may be also represented by \dim{a},\dim{b}, respectively.
\end{convention}

\subsection{Quantum mechanics}
We first review the formalism of quantum mechanics. Readers may refer to \cite{NC00,Wat08} for a thorough treatment. A quantum system $\sS$ is associated with a finite-dimensional Hilbert space, known as the state space of the system. We consider the space of linear operators acting on the states and equip the space with the normalized Hilbert-Schmidt inner product
\[\innerproduct{P}{Q}=\frac{1}{\dim{s}}\Tr~P^\dagger Q.\]
where $P^\dagger$ denotes the conjugate transpose of $P$ and $\dim{s}$ is the dimension of $\sS$.

 A quantum state in $\sS$ can be completely described by a density operator, which is a positive semi-definite operator with trace one. We denote the set of all linear operators in the state space by $\M_\sS$, and the set of Hermitian operators by $\H_\sS$. The identity operator is denoted by $\id_\sS$. If the dimension of $\sS$ is $\dim{s}$, we may write $\H_\sS=\H_\dim{s}$ or $\id_\sS=\id_\dim{s}$. The subscripts $\sS$ and $\dim{s}$ may be dropped whenever it is clear from the context. The state of a composite quantum system is the Kronecker product of the state spaces of the component systems. In this paper, we use the shorthand $\sS\sA$ to represent $\sS\otimes\sA$. A state of a composite system with two components is called a bipartite state. The Hermitian space of the composition of $n$ Hermitian space $\H_\sS$ is denoted by $\H_\sS^{\otimes n}$, or $\H_\sS^{n}$ for short. Quantum measurements are described by a POVM, that is, a number of positive operators $\set{E_m}$ summing to identity. The index $m$ refers to the measurement outcomes that may occur in the experiment. If the state of the quantum system is $\rho$ immediately before the measurement then the probability that result $m$ occurs is given by $\Tr~E_m\rho$.

Given quantum systems $\sS,\sA$, let $\L\br{\sS,\sA}$ denote the set of all linear maps from $\M_\sS$ to $\M_\sA$, and if the input system $\sS$ and the output system are the same, we write $\L\br{\sS}$ for simplicity. A quantum operation from the input system $\sS$ to the output system $\sA$ is represented by a CPTP (completely positive and trace preserving) map $\Phi\in\L\br{\sS,\sA}$. An important example of quantum operations is partial trace. Given quantum systems $\sS,\sA$, and a bipartite state $\psi^{\sS\sA}\in\H_\sS\otimes\H_\sA$ ($\H_{\sS\sA}$ for short), the partial trace $\Tr_{\sA}$ derives the marginal state $\psi^{\sS}$ of the subsystem $\sS$ from $\psi^{\sS\sA}$. The partial trace $\Tr_{\sA}\in\L\br{\sS\sA,\sS}$ is given by
\[\psi^{\sS}=\Tr_\sA~\psi^{\sS\sA}=\sum_i\br{\id_\sS\otimes\bra{i}}\psi^{\sS\sA}\br{\id_\sS\otimes\ket{i}},\]
where $\set{\ket{i}}$ is an orthonormal basis in $\sA$. It is easy to verify that the operation is independent of the choice of basis $\set{\ket{i}}$.


For a given map $\Phi\in\L\br{\sS,\sA}$, the adjoint of $\Phi$ is defined to be the unique map $\Phi^*\in\L\br{\sA,\sS}$ that satisfies
\begin{equation}\label{eqn:adjointmap}
\Tr~\Phi^*(Q)^\dagger P=\Tr~Q^\dagger\Phi(P)
\end{equation}
for all $P\in\L(\sS)$ and $Q\in\L(\sA)$.

Given $\Psi\in\L\br{\sA,\sS}$, the Choi representation of $\Psi$ is a linear map $J:\L\br{\sA,\sS}\rightarrow\H\br{\sS\sA}$ defined as follows:
\begin{equation}\label{eqn:choi}
\choi{\Psi}=\sum_{a}\Psi\br{\atildea}\otimes\atildea,
\end{equation}
where $\atildea=\A_a/\sqrt{\abs{\sA}}$\ipfootnote, and $\set{\A_a:a\in\Br{\abs{\sA}^2}_{\geq0}}$ is an orthonormal basis in $\sA$. In \cite{Wat08} the Choi representation is defined using the basis $\set{E_{i,j}}_{i,j\in\Br{\abs{\sA}}}$, where the $(i,j)$-entry of $E_{i,j}$ is $1$ and the others are $0$. It is easy to verify that the definition is independent of the choice of basis. $J$ is a linear bijection. $\Psi$ can be recovered from its Choi representation $\choi{\Psi}$ as follows.
\begin{equation}\label{eqn:choitophi}
\Psi\br{P}=\Tr_\sA\br{\choi{\Psi}\br{\id_\sS\otimes P^\dagger}}.
\end{equation}


\begin{fact}\cite{Wat08}\label{fac:adjointchoi}
Given $\Phi\in\L\br{\sS,\sA}$, the following three statements are equivalent.
\begin{enumerate}
\item $\Phi$ is completely positive.
\item $\Phi^*$ is completely positive.
\item $\choi{\Phi^*}\geq0$.
\end{enumerate}
And the following four statements are equivalent as well.
\begin{enumerate}
\item $\Phi$ is trace preserving.
\item $\Phi^*$ is unital, that is, $\Phi^*\br{\id_\sA}=\id_\sS$.
\item $\Tr_\sA\choi{\Phi}=\id_\sS$.
\item $\Tr_\sA\choi{\Phi^*}=\id_\sS$.
\end{enumerate}
\end{fact}
\begin{proof}
  For the first part, item 1 and item 2 are equivalent by definition. Item 3 is equivalent to item 1 is by Theorem 2.22 in~\cite{Wat08}. For the second part, item 1 and item 2 are equivalent by definition. Item 1 and item 3 are equivalent by Theorem 2.26 in~\cite{Wat08}.  To see the equivalence between item 2 and item 4, let $\set{\A_a}_a$ be an orthonormal basis in $\sA$ with $\A_0=\id_{\sA}$. By the definition of the Choi representation,
  \[\\Tr_\sA\choi{\Phi^*}=\Tr_\sA\br{\sum_a\Phi^*\br{\atildea}\otimes\atildea}=\Phi^*\br{\A_0}=\Phi^*\br{\id_{\sA}},\]
  where the second equality is by the orthonormality of $\set{\A_a}_a$ and our choice $\A_0=\id_{\sA}$. It is easy to see item 2 and item 4 are equivalent.
\end{proof}

By the above fact, $\Phi$ is a quantum operation if and only if $\choi{\Phi^*}\geq0$ and $\Tr_\sA\choi{\Phi^*}=\id_\sS$.

We also need the following fact.
\begin{fact}\label{fac:cauchyschwartz}\cite[Fact 2.1]{qin2021nonlocal}
		Given quantum systems $\sS,\sT$, operators $P\in\H_{\sS}, Q\in\H_{\sT}$ and a bipartite state $\oneshared\in\H_{\sS\sT}$, it holds that
		\begin{enumerate}
			\item $\Tr\br{\br{P\otimes\id_{\sT}}\psi^{\sS}}=\Tr~P\psi^{\sS}$;
			
			\item $\abs{\Tr\br{\br{P\otimes Q}\psi^{\sS\sT}}}\leq\br{\Tr ~P^2\psi^{\sS}}^{1/2}\cdot\br{\Tr~Q^2\psi^{\sT}}^{1/2}$.
		\end{enumerate}
	\end{fact}
\subsection{Fourier analysis in Gaussian space}
Given $n\in\posint$, let $\gamma_n$ represent a standard $n$-dimensional normal distribution. A function $f:\reals^n\rightarrow\reals$ is in $L^2\br{\reals,\gamma_n}$ if
\[\int_{\reals^n}f(x)^2\gamma_n\br{\d x}<\infty.\]

All the functions involved in this paper are in $L^2\br{\reals,\gamma_n}$. We equip $L^2\br{\reals,\gamma_n}$ with an inner product
\[\innerproduct{f}{g}_{\gamma_n}=\expec{x\sim\gamma_n}{f(x)g(x)}.\]

Given $f\in L^2\br{\reals,\gamma_n}$, the 2-norm of $f$ is defined to be
\[\twonorm{f}=\sqrt{\innerproduct{f}{f}_{\gamma_n}}.\]

 The set of {\em Hermite polynomials} forms an orthonormal basis in $L^2\br{\reals,\gamma_1}$ with respect to the inner product $\innerproduct{\cdot}{\cdot}_{\gamma_1}$. The Hermite polynomials $H_r:\reals\rightarrow\reals$ for $r\in\mathbb{Z}_{\geq 0}$ are defined as
	\begin{equation*}
	H_0\br{x}=1; H_1\br{x}=x; H_r\br{x}=\frac{(-1)^r}{\sqrt{r!}}\e^{x^2/2}\frac{\d^r}{\d x^r}\e^{-x^2/2}.
	\end{equation*}
	For any $\sigma\in\br{\sigma_1,\ldots,\sigma_n}\in\mathbb{Z}_{\geq 0}^n$, define
	$H_{\sigma}:\reals^n\rightarrow\reals$ as \begin{equation*}
	H_{\sigma}\br{x}=\prod_{i=1}^nH_{\sigma_i}\br{x_i}.
	\end{equation*}
	The set $\set{H_{\sigma}:\sigma\in\mathbb{Z}_{\geq 0}^n}$ forms an orthonormal basis in $L^2\br{\reals,\gamma_n}$. Every function $f\in L^2\br{\reals,\gamma_n}$ has an {\em Hermite expansion}  as
$$f\br{x}=\sum_{\sigma\in\mathbb{Z}_{\geq 0}^n}\widehat{f}\br{\sigma}\cdot H_{\sigma}\br{x},$$
	where $\widehat{f}\br{\sigma}$'s are the {\em Hermite coefficients} of $f$, which can be obtained by $\widehat{f}\br{\sigma}=\innerproduct{H_{\sigma}}{f}_{\gamma_n}$. The degree of $f$ is defined to be \[\deg\br{f}=\max\set{\sum_{i=1}^n\sigma_i:~\widehat{f}\br{\sigma}\neq 0}.\]

We say $f\in L^2\br{\reals,\gamma_n}$ is {\em multilinear} if $\widehat{f}\br{\sigma}=0$ for $\sigma\notin\set{0,1}^n$.
	
	\begin{definition}\label{def:influenceGaussian}
		Given a function $f\in L^2\br{\reals,\gamma_n}$,
		the {\em variance} of $f$ is defined to be
		\begin{equation*}
		\var{}{f}=\expec{\mathbf{x}\sim \gamma_n}{\abs{f\br{\mathbf{x}}-\expec{}{f}}^2}.
		\end{equation*}
		The {\em influence} of the $i$-th coordinate(variable) on $f$, denoted by $\influence_i\br{f}$, is defined by
		\begin{equation*}
		\influence_i\br{f}=\expec{\mathbf{x}\sim \gamma_n}{\var{\mathbf{x}'_i\sim\gamma_1}{f\br{\mathbf{x}_1,\dots,\mathbf{x}_{i-1},\mathbf{x}'_i,\mathbf{x}_{i+1},\dots\mathbf{x}_n}}}.
		\end{equation*}
\end{definition}

The following fact summarizes some basic properties of variance and influence.
\begin{fact}\label{fac:Gaussianinf}\cite[Proposition 8.16 and Proposition 8.23]{Odonnell08}
Given $f\in L^2\br{\reals,\gamma_n}$, it holds that
\begin{enumerate}
\item
$\var{}{f}=\sum_{\sigma\ne0^n}\widehat{f}\br{\sigma}^2\leq\sum_{\sigma}\widehat{f}\br{\sigma}^2=\twonorm{f}^2.$
\item
$\influence_i\br{f}=\sum_{\sigma_i\ne0}\widehat{f}\br{\sigma}^2\leq\var{}{f}.$
\end{enumerate}	
\end{fact}

\begin{definition}\label{def:ornstein}
		Given $0\leq\nu\leq 1$ and $f\in L^2\br{\reals,\gamma_n}$, we define the {\em Ornstein-Uhlenbeck operator} $U_\nu$ to be
		\[U_\nu f\br{z}=\expec{\mathbf{x}\sim \gamma_n}{f\br{\nu z+\sqrt{1-\nu^2}\mathbf{x}}}.\]
	\end{definition}
		\begin{fact}\label{fac:Gaussiannoisy}~\cite[Page 338, Proposition 11.37]{Odonnell08}
		For any $0\leq\nu\leq 1$ and $f\in L^2\br{\reals,\gamma_n}$, it holds that
		\[U_\nu f=\sum_{\sigma\in\mathbb{Z}_{\geq0}^n}\widehat{f}\br{\sigma}\nu^{\sum_{i=1}^n\sigma_i}H_{\sigma}.\]
	\end{fact}

A vector-valued function $f=\br{f_1,\dots,f_k}:\reals^n\rightarrow\reals^k$ is in $L^2\br{\reals^k,\gamma_n}$ if $f_i\in L^2\br{\reals,\gamma_n}$ for all $i\in[n]$. The $2$-norm of $f$ is defined to be
\[\twonorm{f}=\br{\sum_{i=1}^k\twonorm{f_i}^2}^{1/2}.\]
The action of Ornstein-Uhlenbeck operator on $f$ is defined to be
\begin{equation}\label{eqn:vectorornstein}
U_\nu f=\br{U_\nu f_1,\dots,U_\nu f_k}.
\end{equation}

Given $\rho\in[0,1]$, $\G_\rho$ denotes the distribution of $\rho$-correlated Gaussians, that is,
\[\br{\mathbf{g},\mathbf{h}}\sim N\br{\begin{pmatrix}
       0 \\
       0
     \end{pmatrix},\begin{pmatrix}
                     1 & \rho \\
                     \rho & 1
                   \end{pmatrix}}.\]

Given $f,g\in L^2\br{\reals,\gamma_n}$, we denote
\[\innerproduct{f}{g}_{\G_\rho^{\otimes n}}=\expec{\br{\mathbf{x},\mathbf{y}}\sim\G_\rho^{\otimes n}}{f(\mathbf{x})g(\mathbf{y})}.\]

\subsection{Fourier analysis in matrix space}
Given $1\leq m,p\leq\infty$, and $H\in\H_m$, the $p$-norm of $H$ is defined to be
\[
\norm{H}_p=\br{\Tr~\abs{H}^p}^{1/p},
\]
where $\abs{H}=\br{H^2}^{1/2}$. It is easy to verify that for $1\leq q\leq p\leq \infty$,
\begin{equation}\label{eqn:normequivalence}
\norm{H}_p\leq\norm{H}_q\leq m^{1/q-1/p}\norm{H}_p.
\end{equation}

The {\em normalized $p$-norm} of $H$ is defined as
\[
\nnorm{H}_p=\br{\frac{1}{m}\Tr~\abs{H}^p}^{1/p}.
\]

For $1\leq q\leq p\leq \infty$, by \cref{eqn:normequivalence}, we have
\begin{equation}\label{eqn:nnormequivalence}
\nnorm{H}_q\leq\nnorm{H}_p\leq m^{1/q-1/p}\nnorm{H}_q.
\end{equation}

Given $P,Q\in\H_m$, we define an inner product over $\reals$:
\[\innerproduct{P}{Q}=\frac{1}{m}\Tr~PQ.\]


We need the following particular classes of bases in $\H_m$ on which our Fourier analysis is based.
\begin{definition}
Let $\set{\B_i}_{i\in\Br{m^2}_{\geq0}}$ be an orthonormal basis in $\H_m$ over $\reals$. We say $\set{\B_i}_{i\in\Br{m^2}_{\geq0}}$ is a standard orthonormal basis if $\B_0=\id_m$.
\end{definition}


\begin{fact}\label{fac:paulimutiplecopy}
	Let $\set{\B_i}_{i=0}^{m^2-1}$ be a standard orthonormal basis in $\H_m$. Then
	\[\set{\B_{\sigma}=\otimes_{i=1}^n\B_{\sigma_i}}_{\sigma\in[m^2]_{\geq 0}^n}\]
	is a standard orthonormal basis in $\H_m^{\otimes n}$.
\end{fact}

Given a standard orthonormal basis $\set{\B_i}_{i=0}^{m^2-1}$ in $\H_m$, every $H\in\H_m^{\otimes n}$ has a Fourier expansion:
\[H=\sum_{\sigma\in\Br{m^2}_{\geq0}^n}\widehat{H}\br{\sigma}\B_\sigma,\]
where $\widehat{H}\br{\sigma}\in\reals$ are the Fourier coefficients. The basic properties of $\widehat{H}\br{\sigma}$'s are summarized in the following fact, which can be easily derived from the orthonormality of $\set{\B_{\sigma}}_{\sigma\in[m^2]_{\geq 0}^n}$.

\begin{fact}\label{fac:basicfourier}\cite[Fact 2.11]{qin2021nonlocal}
	Given a standard orthonormal basis $\set{\B_i}_{i\in\Br{m^2}_{\geq0}}$ in $\H_m$ and $M,N\in\H_m$, it holds that
	\begin{enumerate}
		\item $\innerproduct{M}{N}=\sum_{\sigma}\widehat{M}\br{\sigma}\widehat{N}\br{\sigma}$.
		\item $\nnorm{M}_2^2=\innerproduct{M}{M}=\sum_{\sigma}\widehat{M}\br{\sigma}^2$.
		\item $\innerproduct{\id_m}{M}=\widehat{M}\br{0}$.	
	\end{enumerate}
\end{fact}

\begin{definition}
		Let $\B=\set{\B_i}_{i\in\Br{m^2}_{\geq0}}$ be a standard orthonormal basis in $\H_m$, $P,Q\in\H_m^{\otimes n}$
		\begin{enumerate}
			\item The degree of $P$ is defined to be \[\deg P=\max\set{\abs{\sigma}:\widehat{P}\br{\sigma}\neq 0}.\]
Recall that $\abs{\sigma}$ represents the number of nonzero entries of $\sigma$.
			\item For any $i\in[n]$, the influence of $i$-th coordinate is defined to be
			\begin{equation}\label{eqn:definf}
\influence_i\br{P}=\nnorm{P-\id_m\otimes\Tr_iP}_2^2,
\end{equation}

			where $\id_m$ is in the $i$'th quantum system, and the partial trace $\Tr_i$ derives the marginal state of the remaining $n-1$ quantum systems except for the $i$'th one.
			
			\item The total influence of $P$ is defined to be \[\influence\br{P}=\sum_i\influence_i\br{P}.\]
		\end{enumerate}
	\end{definition}

\begin{fact}\label{fac:partialvariance}\cite[Lemma 2.16]{qin2021nonlocal}
	Given $P\in\H_m^{\otimes n}$, a standard orthonormal basis $\B=\set{\B_i}_{i\in\Br{m^2}_{\geq0}}$ in $\H_m$, it holds that
	\begin{enumerate}
		\item $\influence_i\br{P}=\sum_{\sigma:\sigma_i\neq0}\abs{\widehat{P}\br{\sigma}}^2.$
		\item $\influence\br{P}=\sum_{\sigma}\abs{\sigma}\abs{\widehat{P}\br{\sigma}}^2\leq\deg P\cdot\nnorm{P}^2_2.$
	\end{enumerate}
\end{fact}

	With the notion of degrees, we define the low-degree part and the high-degree part of an operator.
	
	\begin{definition}\label{def:lowdegreehighdegree}
		Given $m,t\in\posint$, a standard orthonormal basis $\B=\set{\B_i}_{i\in\Br{m^2}_{\geq0}}$ in $\H_m$ and $P\in\H_m^{\otimes n}$, we define
		\[P^{\leq t}=\sum_{\sigma\in[m^2]_{\geq 0}^n:\abs{\sigma}\leq t}\widehat{P}\br{\sigma}\B_{\sigma};\]
		\[P^{\geq t}=\sum_{\sigma\in[m^2]_{\geq 0}^n:\abs{\sigma}\geq t}\widehat{P}\br{\sigma}\B_{\sigma}\]
		and
		\[P^{=t}=\sum_{\sigma\in[m^2]_{\geq 0}^n:\abs{\sigma}=t}\widehat{P}\br{\sigma}\B_{\sigma},\]
		where $\widehat{P}\br{\sigma}$'s are the Fourier coefficients of $P$ with respect to the basis $\B$.

	\end{definition}

\begin{fact}\label{fac:pt}\cite[Lemma 2.15]{qin2021nonlocal}
	The degree of $P$ is independent of the choices of bases. Moreover, $P^{\leq t}, P^{\geq t}$ and $P^{=t}$ are also independent of the choices of bases.
\end{fact}

\begin{definition}\label{def:bonamibeckner}
	Given a quantum system $\sS$ with dimension $\abs{\sS}=\dim{s}$, $\gamma\in[0,1]$, the depolarizing operation $\Delta_{\gamma}:\H_\sS\rightarrow\H_\sS$ is defined as follows. For any $P\in\H_\sS$,
	\[\Delta_{\gamma}\br{P}=\gamma P+\frac{1-\gamma}{\dim{s}}\br{\Tr~P}\cdot\id_\sS.\]
\end{definition}


\begin{fact}\label{lem:bonamibecknerdef}\cite[Lemma 3.6 and Lemma 6.1]{qin2021nonlocal}
	Given $n,m\in\posint$, $\gamma\in[0,1]$, a standard orthonormal basis of $\H_m$: $\B=\set{\B_i}_{i=0}^{m^2-1}$, the following holds:
	\begin{enumerate}
		\item For any $P\in\H_m^{\otimes n}$ with a Fourier expansion $P=\sum_{\sigma\in[m^2]_{\geq 0}^n}\widehat{P}\br{\sigma}\B_{\sigma}$, it holds that
		\[\Delta_{\gamma}\br{P}=\sum_{\sigma\in[m^2]_{\geq 0}^n}\gamma^{\abs{\sigma}}\widehat{P}\br{\sigma}\B_{\sigma}.\]
		\item For any $P\in\H_m^{\otimes n}$, $\nnorm{\Delta_{\gamma}\br{P}}_2\leq \nnorm{P}_2$.
		\item $\Delta_{\gamma}$ is a quantum operation.
		\item For any $d\in\posint,P\in\H_m^{\otimes n}$, it holds that
	\[\nnorm{\br{\Delta_\gamma(P)}^{>d}}_2\leq\gamma^d\nnorm{P}_2.\]
	\end{enumerate}
\end{fact}

\begin{definition}[Maximal correlation]~\cite{Beigi:2013}\label{def:maximalcorrelation}
	Given quantum systems ${\sS},{\sT}$ with dimensions $\dim{s}=\abs{\sS}$ and $\dim{t}=\abs{\sT}$, $\psi^{\sS\sT}\in\H_{\sS\sT}$ with $\psi^{\sS}=\id_{\sS}/\dim{s},\psi^{\sT}=\id_{\sT}/\dim{t}$, the maximal correlation of $\psi^{\sS\sT}$ is defined to be
	\[\rho\br{\psi^{\sS\sT}}=\sup\set{\abs{\Tr\br{\br{P\otimes Q}\psi^{\sS\sT}}}~:P\in\H_{\sS}, Q\in\H_{\sT},\atop\Tr~P=\Tr~Q=0, \nnorm{P}_2=\nnorm{Q}_2=1}.\]
\end{definition}

\begin{fact}~\cite{Beigi:2013}\label{fac:maximalcorrlationone}
	Given quantum systems ${\sS},{\sT}$ with dimensions $\dim{s}=\abs{\sS}$ and $\dim{t}=\abs{\sT}$, $\psi^{\sS\sT}\in\H_{\sS\sT}$ with $\psi^{\sS}=\id_{\sS}/\dim{s},\psi^{\sT}=\id_{\sT}/\dim{t}$, it holds that $\rho\br{\psi^{\sS\sT}}\leq 1$.
\end{fact}

\begin{definition}\label{def:noisyepr}
		Given quantum systems ${\sS},{\sT}$ with dimensions $\dim{s}=\abs{\sS}$ and $\dim{t}=\abs{\sT}$, a bipartite state $\psi^{\sS\sT}\in\H_{\sS\sT}$ is a noisy maximally entangled state (MES) if $\psi^{\sS}=\id_{\sS}/\dim{s},\psi^{\sT}=\id_{\sT}/\dim{t}$ and its maximal correlation $\rho<1$.
	\end{definition}
Beigi proved that depolarized maximally entangled states are noisy maximally entangled states.

	\begin{fact}~\cite[Page 5 in arXiv version]{Beigi:2013}~\cite[Lemma 3.9]{qin2021nonlocal}\label{lem:noisyeprmaximalcorrelation}
		For any $0\leq\epsilon<1$, an integer $m>1$, it holds that
		\[\rho\br{\br{1-\epsilon}\ketbra{\Psi}+\epsilon\frac{\id_{m}}{m}\otimes \frac{\id_{m}}{m}}=1-\epsilon,\]
		where $\ket{\Psi}=\frac{1}{\sqrt{m}}\sum_{i=0}^{m-1}|m,m\rangle$.
	\end{fact}


\begin{fact}\label{lem:normofM}\cite[Lemma 7.3]{qin2021nonlocal}
	Given quantum systems ${\sS},{\sT}$ with dimensions $\dim{s}=\abs{\sS}$ and $\dim{t}=\abs{\sT}$, if $\psi^{\sS\sT}\in\H_{\sS\sT}$ is a noisy MES with maximal correlation $\rho$, then there exist standard orthonormal bases $\set{\S_s}_{\srange{}}$  and $\set{\T_t}_{\trange{}}$ in $\sS$ and $\sT$, respectively, such that
	\begin{equation}\label{eqn:propersob}
\Tr\br{\br{\S_i\otimes\T_j}\psi^{\sS\sT}}=\begin{cases}c_i~&\mbox{if $i=j$}\\0~&\mbox{otherwise},\end{cases}
\end{equation}
	where $c_1=1, c_2=\rho$ and $c_1\geq c_2\geq c_3\geq\ldots\geq 0$.
\end{fact}

As described in \cref{subsec:comparison}, one of the difficulties is that the input of a fully quantum nonlocal game is a tripartite quantum state. The following lemma enables us to 'discretize' the input state by properly chosen orthonormal bases.

\begin{lemma}\label{lem:pqrdiag}
Given quantum systems $\sP,\sQ,\sR$ with $\abs{\sP}=\dim{p},\abs{\sQ}=\dim{q},\abs{\sR}=\dim{r}$, an orthonormal basis $\set{\R_r}_{\rrange}$ in $\H_\sR$, a tripartite quantum state $\abrshared\in\H_{\sP\sQ\sR}$ and an integer $\rrange$, there exist orthonormal bases $\set{\P_p}_{\prange}$ and $\set{\Q_q}_{\qrange}$ in $\H_\sP$ and $\H_\sQ$, respectively, which may depend on $r$, such that
\begin{equation}\label{eqn:pqdiag}
\Tr\Br{\br{\ptildep\otimes\qtildeq\otimes \rtilder}\abrshared}=
\begin{cases}
k_p&\text{if }p=q\\
0&\text{otherwise,}
\end{cases}
\end{equation}
where $\ptildep=\P_p/\sqrt{\dim{p}}$, $\qtildeq=\Q_q/\sqrt{\dim{q}}$, $\rtilder=\R_r/\sqrt{\dim{r}}$ and $k_0,\dots,k_{\dim{p}^2-1}\in[0,1]$. \footnote{Assume $\dim{p}\geq\dim{q}$ without loss of generality. If $\dim{p}>\dim{q}$, then $k_{\dim{q}^2}=\dots=k_{\dim{p}^2-1}=0.$}
\end{lemma}
\begin{remark}
  Notice that $\set{\P_p}_{\prange}, \set{\Q_q}_{\qrange}, \set{\R_r}_{\rrange}$ are not required to be  standard orthonormal bases.
\end{remark}
\begin{proof}[Proof of \cref{lem:pqrdiag}]
Let $\set{\P'_p}_{\prange}$ and $\set{\Q'_q}_{\qrange}$ be arbitrary orthonormal bases in $\H_\sP$ and $\H_\sQ$, respectively. Let $M$ be a $\dim{p}^2\times\dim{q}^2$ matrix such that
\[M_{p,q}=\Tr~\br{\widetilde{\P'_p}\otimes\widetilde{\Q'_q}\otimes \rtilder}\abrshared.\]
Then $M$ is a real matrix. Thus, it has a singular value decomposition
\[M=UDV^\dagger,\]
where $U\in\M_{\dim{p}^2}$ and $V\in\M_{\dim{q}^2}$ are orthonormal matrices (i.e., real unitary matrices) and $D$ is a $\dim{p}^2\times\dim{q}^2$ diagonal matrix with diagonal entries non-negative. Define
\[\P_p=\sum_{p'\in\Br{\dim{p}^2}_{\geq0}}U^\dagger_{p,p'}\P'_{p'}\]
and
\[\Q_q=\sum_{q'\in\Br{\dim{q}^2}_{\geq0}}V_{q',q}\Q'_{q'}.\]
Then $\set{\P_p}_{\prange}$ and $\set{\Q_q}_{\qrange}$ are orthonormal bases as well. We have
\[
\Tr\Br{\br{\ptildep\otimes\qtildeq\otimes \rtilder}\abrshared}\begin{cases}
\geq0&\text{if }p=q,\\
=0&\text{otherwise.}
\end{cases}
\]
In particular, since $\set{\P_p}_{\prange}$, $\set{\Q_q}_{\qrange}$ and $\set{\R_r}_{\rrange}$ are orthonormal bases,
\begin{align*}
&\abs{\Tr\Br{\br{\ptildep\otimes\qtildeq\otimes \rtilder}\abrshared}}\\
\leq~&\norm{\ptildep\otimes\qtildeq\otimes \rtilder}\norm{\abrshared}_1\quad\mbox{(H\"older's)}\\
=~&\norm{\ptildep}\norm{\qtildeq}\norm{\rtilder}\\
\leq~&\norm{\ptildep}_2\norm{\qtildeq}_2\norm{\rtilder}_2\\
=~&1,
\end{align*}
\cref{eqn:pqdiag} holds.
\end{proof}

\subsection{Random operators}\label{sec:randop}
In this subsection, we introduce random operators defined in~\cite{qin2021nonlocal}, which unifies Gaussian variables and operators.

\begin{definition}~\cite{qin2021nonlocal}\label{def:randomoperators}
		Given $p,h,n,m\in\posint$, we say $\mathbf{P}$ is a random operator if it can be expressed as
		\begin{equation*}
		\mathbf{P}=\sum_{\sigma\in[m^2]_{\geq 0}^h}p_{\sigma}\br{\mathbf{g}}\B_{\sigma},
		\end{equation*}
		where $\set{\B_i}_{i\in\Br{m^2}_{\geq0}}$ is a standard orthonormal basis in $\H_m$, $p_{\sigma}:\reals^n\rightarrow\reals$ for all $\sigma\in[m^2]_{\geq 0}^h$ and $\mathbf{g}\sim \gamma_n.$ $\mathbf{P}\in L^p\br{\H_m^{\otimes h},\gamma_n}$ if $p_{\sigma}\in L^p\br{\reals,\gamma_n}$ for all $\sigma\in[m^2]_{\geq 0}^h$. Define a vector-valued function \[p=\br{p_{\sigma}}_{\sigma\in[m^2]_{\geq 0}^h}:\reals^n\rightarrow\reals^{m^{2h}}.\] We say $p$ is the {\em associated vector-valued function} of $\mathbf{P}$ under the basis $\set{\B_i}_{i\in\Br{m^2}_{\geq0}}$.


The degree of $\mathbf{P}$, denoted by $\deg\br{\mathbf{P}}$, is \[\max_{\sigma\in[m^2]_{\geq 0}^h}\deg\br{p_{\sigma}}.\]


We say $\mathbf{P}$ is multilinear if $p_{\sigma}\br{\cdot}$ is multilinear for all $\sigma\in[m^2]_{\geq 0}^h$.
	\end{definition}

	\begin{fact}\label{lem:randoperator}\cite[Lemma 2.23]{qin2021nonlocal}
	Given $n,h,m\in\posint$, let $\mathbf{P}\in L^2\br{\H_m^{\otimes h},\gamma_n}$ with an  associated vector-valued function $p$ under a standard orthonormal basis. It holds that  $\expec{}{\nnorm{\mathbf{P}}_2^2}=\twonorm{p}^2.$
\end{fact}

We say a pair of random operators $\br{\mathbf{P},\mathbf{Q}}\in L^p\br{\H_m^{\otimes h},\gamma_n}\times L^p\br{\H_m^{\otimes h},\gamma_n} $ are {\em joint random operators} if the random variables $\br{\mathbf{g},\mathbf{h}}$ in $\br{\mathbf{P},\mathbf{Q}}$ are drawn from a joint distribution $\G_{\rho}^{\otimes n}$ for $0\leq\rho\leq 1$.

\subsection{Rounding maps}

Given a closed convex set $\Delta\subseteq\reals^k$, the rounding map of $\Delta$, denoted by $\R:\reals^k\rightarrow\reals^k$, is defined as follows:
\[\R(x)=\arg\min\set{\twonorm{y-x}:y\in\Delta}.\]
The following well-known fact states that the rounding maps of closed convex sets are Lipschitz continuous with Lipschitz constant being 1.

\begin{fact}\label{fac:rounding}\cite[Page 149, Proposition 3.2.1]{bertsekas2015convex}
	Let $\Delta$ be a nonempty closed convex set in $\reals^k$ with the rounding map $\R$. It holds that
	\[\twonorm{\R\br{x}-\R\br{y}}\leq\twonorm{x-y},\]
	for any $x,y\in\reals^k$.

\end{fact}

Define a function $\zeta:\reals\rightarrow\reals$ as follows.

\begin{equation}\label{eqn:zeta}
	\zeta\br{x}=\begin{cases}x^2~&\mbox{if $x\leq 0$}\\ 0~&\mbox{otherwise}\end{cases}.
\end{equation}
The function $\zeta$ measures the distance between an Hermitian operator and the set of positive semi-definite operators in $2$-norm.
\begin{fact}\label{lem:closedelta1}\cite[Lemma 9.1]{qin2021nonlocal}
	Given $m\in\posint$, $H\in\H_m$, it holds that
	\[\Tr~\zeta\br{H}=\min\set{\twonorm{H-X}^2:X\geq 0}.\]
\end{fact}
\begin{fact}\label{lem:zetaadditivity}\cite[Lemma 10.4]{qin2021nonlocal}
	For any Hermitian matrices $P$ and $Q$, it holds that \[\abs{\Tr~\br{\zeta\br{P+Q}-\zeta\br{P}}}\leq2\br{\twonorm{P}\twonorm{Q}+\twonorm{Q}^2}.\]
\end{fact}


%
%
%
\section{Main results}
\begin{theorem}\label{thm:nijs}
Given $\epsilon\in\br{0,1}$, $n,s\in\posint$, and quantum systems $\sP,\sQ,\sR,{\sS},{\sT},{\sA},{\sB}$ with dimensions
$\dim{p}=\abs{\sP},\dim{q}=\abs{\sQ},\dim{r}=\abs{\sR},\dim{s}=\abs{\sS},\dim{t}=\abs{\sT},\dim{a}=\abs{\sA},\dim{b}=\abs{\sB}.$
Let
$\set{\A_a}_{\arange}$, $\set{\B_b}_{\brange}$, $\set{\R_r}_{\rrange}$ be orthonormal bases in $\H_{\sA}$, $\H_{\sB}$ and $\H_{\sR}$, respectively. Let
$\oneshared\in\H_{\sS\sT}$ be a noisy MES with the maximal correlation $\rho=\rho\br{\oneshared}<1$, which is defined in \cref{def:maximalcorrelation}. Let $\phi_{\textsf{in}}^{\sP\sQ\sR}\in\H_{\sP\sQ\sR}$ be an arbitrary tripartite quantum state. Then there exists an explicitly computable $D=D\br{\rho,\epsilon,s,\dim{p},\dim{q},\dim{r},\dim{s},\dim{t},\dim{a},\dim{b}}$, such that for all quantum operations $\alice\in\L\br{\sS^n\sP,\sA}$, $\bob\in\L\br{\sT^n\sQ,\sB}$, there exist quantum operations
$\widetilde{\alice}\in\L\br{\sS^D\sP,\sA}$, $\widetilde{\bob}\in\L\br{\sT^D\sQ,\sB}$ such that for all $\arange$, \brange, \rrange, \remindfootnote
\begin{multline*}
\left|\Tr\Br{\br{\alice^*\br{\atildea}\otimes\bob^*\br{\btildeb}\otimes\rtilder}\br{\phi_{\textsf{in}}^{\sP\sQ\sR}\otimes\shared{n}}}\right.\\-\left.\Tr\Br{\br{\widetilde{\alice^*}\br{\atildea}\otimes\widetilde{\bob^*}\br{\btildeb}\otimes\rtilder}\br{\phi_{\textsf{in}}^{\sP\sQ\sR}\otimes\shared{D}}}\right|\leq\epsilon.
\end{multline*}

In particular, one may choose 
\[D=\exp\br{\poly{\mathsf{a},\mathsf{b},\mathsf{p},\mathsf{q},\mathsf{r},\log\mathsf{s},\log\mathsf{t},\frac{1}{1-\rho},\frac{1}{\epsilon}}}.\]

\end{theorem}
%
%
\begin{theorem}\label{thm:main}
Given parameters $0<\epsilon,\rho<1$, and a fully quantum nonlocal game	\[\mathfrak{G}=\br{\sP,\sQ,\sR,\sA,\sB,\abrshared,\set{M^{\sA\sB\sR},\id-M^{\sA\sB\sR}}},\]
with dimensions
$\dim{p}=\abs{\sP},\dim{q}=\abs{\sQ},\dim{r}=\abs{\sR},\dim{s}=\abs{\sS},\dim{t}=\abs{\sT},\dim{a}=\abs{\sA},\dim{b}=\abs{\sB}$,
suppose the two players are restricted to share an arbitrarily finite number of noisy MES states $\oneshared$, i.e., $\psi^{\sS}=\id_{\sS}/\dim{s}$, $\psi^{\sT}=\id_{\sT}/\dim{t}$ with the maximal correlation $\rho<1$ as defined in \cref{def:maximalcorrelation}. Let $\mathrm{val}_Q(\mathfrak{G},\oneshared)$ be the supremum of the winning probability that the players can achieve. Then there exists an explicitly computable bound $D=D\br{\rho,\epsilon,\dim{p},\dim{q},\dim{r},\dim{s},\dim{t},\dim{a},\dim{b}}$ such that it suffices for the players to share $D$ copies of $\oneshared$ to achieve the winning probability at least $\mathrm{val}_Q(\mathfrak{G},\oneshared)-\epsilon$. In particular, one may choose
		\[D=\exp\br{\poly{\dim{a},\dim{b},\dim{p},\dim{q},\dim{r},\log\dim{s},\log\dim{t},\frac{1}{1-\rho},\frac{1}{\epsilon}}}.\]
\end{theorem}

\begin{proof}
To keep the notations short, the superscripts will be omitted whenever it is clear from the context. Suppose the players share $n$ copies of $\oneshared$ and employ the strategy $\br{\alice,\bob}$ with the winning probability $\mathrm{val}_Q(\mathfrak{G},\oneshared)$. We apply \cref{thm:nijs} to $\br{\alice,\bob}$ with $\epsilon\leftarrow\epsilon/(\dim{a}\dim{b}\dim{r})^{3/2}$ to obtain $\br{\widetilde{\alice},\widetilde{\bob}}$. We claim that the strategy $\br{\widetilde{\alice},\widetilde{\bob}}$ wins the game with probability at least $\mathrm{val}_Q(\mathfrak{G},\oneshared)-\epsilon$.

Let
$\set{\A_a}_{\arange}$, $\set{\B_b}_{\brange}$, $\set{\R_r}_{\rrange}$ be orthonormal bases in $\H_{\sA}$, $\H_{\sB}$ and $\H_{\sR}$, respectively. From \cref{thm:nijs}, for all \arange, \brange, \rrange, we have
\begin{multline*}
\left|\Tr\Br{\br{\alice^*\br{\atildea}\otimes\bob^*\br{\btildeb}\otimes\rtilder}\br{\abrshared\otimes\psi^{\otimes n}}}\right.\\-\left.\Tr\Br{\br{\widetilde{\alice^*}\br{\atildea}\otimes\widetilde{\bob^*}\br{\btildeb}\otimes\rtilder}\br{\abrshared\otimes\psi^{\otimes D}}}\right|\leq\epsilon/(\dim{a}\dim{b}\dim{r})^{3/2}.
\end{multline*}
By \cref{eqn:adjointmap}, it is equivalent to
\begin{multline*}
\left|\Tr\Br{\br{\br{\alice\otimes\bob}\br{\abrshared\otimes\psi^{\otimes n}}}\br{\atildea\otimes\btildeb\otimes\rtilder}}\right.\\
\left.-\Tr\Br{\br{\br{\widetilde{\alice}\otimes\widetilde{\bob}}\br{\abrshared\otimes\psi^{\otimes D}}}\br{\atildea\otimes\btildeb\otimes\rtilder}}\right|\leq\epsilon/(\dim{a}\dim{b}\dim{r})^{3/2}.
\end{multline*}

We finally get the desired result:
\begin{align*}
&\abs{\Tr\Br{M^{\sA\sB\sR}\br{\br{\alice\otimes\bob}\br{\abrshared\otimes\psi^{\otimes n}}-\br{\widetilde{\alice}\otimes\widetilde{\bob}}\br{\abrshared\otimes\psi^{\otimes D}}}}}\\
\overset{(\star)}{\leq}~&\norm{M^{\sA\sB\sR}}\cdot\norm{\br{\alice\otimes\bob}\br{\abrshared\otimes\psi^{\otimes n}}-\br{\widetilde{\alice}\otimes\widetilde{\bob}}\br{\abrshared\otimes\psi^{\otimes D}}}_1\\
\overset{(\star\star)}{\leq}~&\br{\dim{a}\dim{b}\dim{r}}^{1/2}\norm{\br{\alice\otimes\bob}\br{\abrshared\otimes\psi^{\otimes n}}-\br{\widetilde{\alice}\otimes\widetilde{\bob}}\br{\abrshared\otimes\psi^{\otimes D}}}_2\\
=~&\left(\dim{a}\dim{b}\dim{r}\sum_{a,b,r}\left(\Tr\Br{\br{\br{\alice\otimes\bob}\br{\abrshared\otimes\psi^{\otimes n}}}\br{\atildea\otimes\btildeb\otimes\rtilder}}\right.\right.\\
&\left.\left.-\Tr\Br{\br{\br{\widetilde{\alice}\otimes\widetilde{\bob}}\br{\abrshared\otimes\psi^{\otimes D}}}\br{\atildea\otimes\btildeb\otimes\rtilder}}\right)^2\right)^{1/2}\\
\leq~&\epsilon,
\end{align*}
where $(\star)$ is by H\"older's inequality, and $(\star\star)$ is by \cref{eqn:normequivalence}.

\end{proof}

\subsection{Notations and setup}\label{subsec:notationsetup}
The proof of~\cref{thm:nijs} involves a number of notations. To keep the proof succinct, we introduce the setup and the notations that are used frequently in the rest of the paper. Some of the notations have been defined in~\cref{thm:nijs}. We collect them here for readers' convenience.

The notations used to represent quantum systems, bases, dimensions and operators are summarized in \cref{tab:sysbasisdimop} following \cref{convention}.

\begin{table}[!htbp]
\centering
\begin{tabular}{|cccccccc|}
\hline
System		&\sS 				&\sP 				&\sA 				&\sT 				&\sQ 				&\sB 				&\sR 	\\
Basis			&\set{\S_s}	&\set{\P_p}	&\set{\A_a}	&\set{\T_t}	&\set{\Q_q}	&\set{\B_b}	&\set{\R_r}\\
Dimension	&\dim{s}		&\dim{p}		&\dim{a}		&\dim{t}		&\dim{q}		&\dim{b}		&\dim{r}\\		
Operator 	&					&M 				&					&					&N 				&					&\\
\hline
\end{tabular}\caption{Some notations}\label{tab:sysbasisdimop}
\end{table}
\newcommand{\nofrmtab}[1]{\begin{tabular}{c}#1\\\end{tabular}}

\begin{setup}\label{setup}
Given quantum systems $\sP,\sQ,\sR,\sS,\sT,\sA,\sB$ with dimensions
\[\dim{p}=\abs{\sP},\dim{q}=\abs{\sQ},\dim{r}=\abs{\sR},\dim{s}=\abs{\sS},\dim{t}=\abs{\sT},\dim{a}=\abs{\sA},\dim{b}=\abs{\sB},\]
let $\phi_{\textsf{in}}^{\sP\sQ\sR}$ be the input state in $\sP\otimes\sQ\otimes\sR$ shared among Alice, Bob and the referee, where Alice, Bob and the referee hold $\sP$, $\sQ$ and $\sR$, respectively. Let $\oneshared\in\H_{\sS\sT}$ be the noisy MES shared between Alice and Bob, where Alice has $\sS$ and Bob has $\sT$.  Let $\rho<1$ be the maximal correlation of $\oneshared$. Let ${\sA}$ and ${\sB}$ be the answer registers of Alice and Bob, respectively.

%

Let $\set{\S_s}_{s\in\Br{\dim{s}^2}_{\geq0}},\set{\T_t}_{t\in\Br{\dim{t}^2}_{\geq0}}$ be standard orthonormal bases in $\H_\sS,\H_\sT$, respectively. Let $\set{\A_a}_{\arange},\set{\B_b}_{\brange},\set{\P_p}_{\prange},\set{\Q_q}_{\qrange},\set{\R_r}_{\rrange}$ be orthonormal bases (not necessary to be standard orthonormal) in $\H_{\sA},\H_{\sB},\H_{\sP},\H_{\sQ},\H_{\sR}$, respectively.  We require that  $\set{\S_s}_{s\in\Br{\dim{s}^2}_{\geq0}}$ and $\set{\T_t}_{t\in\Br{\dim{t}^2}_{\geq0}}$ satisfy \cref{eqn:propersob}. For convenience, we denote $\atildea$ to be $\A_a/\sqrt{\dim{a}}$. The same for \btildeb, \ptildep, \qtildeq, \rtilder.

When we use universal quantifiers, we omit the ranges of the variables whenever they are clear in the context. For example, we say ``for all $a$, $b$'' to mean ``for all \arange, \brange''.

Given $M\in\hspa{n}$, for all $p, a$, we define $M_a$ to be $\Tr_\sA\Br{\br{\id_{\sS^n\sP}\otimes \atildea}M}$, and $M_{p,a}$ to be $\Tr_\sP\Br{\br{\id_{\sS^n}\otimes \ptildep}M_a}$. Similar for $N, N_b, N_{q,b}$. In other words,
\begin{equation}\label{eqn:defMa}
M=\sum_{\arange}M_a\otimes\atildea,\quad N=\sum_{\brange}N_b\otimes\btildeb.
\end{equation}
and
\begin{equation}\label{eqn:defmpa}
M_a=\sum_{\prange}M_{p,a}\otimes\ptildep,\quad N_b=\sum_{\qrange}N_{q,b}\otimes\qtildeq.
\end{equation}

\end{setup}

\subsection{Proof of \cref{thm:nijs}}

\begin{proof}[Proof of \cref{thm:nijs}]
Let $\delta,\theta$ be parameters which are chosen later. The proof is composed of several steps.
\begin{itemize}
\item \textbf{Smoothing}

We apply \cref{lem:smoothing} to $\choi{\alice^*}$ and $\choi{\bob^*}$ with $\delta\leftarrow\delta$ to get $M^{(1)}$ and $N^{(1)}$, respectively. 
They satisfy the following.

\begin{enumerate}
\item  For all $a,b$:\[\nnorm{M^{(1)}_a}_2\leq1\quad\mbox{and}\quad\nnorm{N^{(1)}_b}_2\leq1,\]
    where $M^{(1)}_a$ and $N^{(1)}_b$ are defined in Eq.~\eqref{eqn:defMa}.
\item
 For all $a,b,r$:
\begin{multline*}
\left|\Tr\Br{\br{\alice^*\br{\atildea}\otimes\bob^*\br{\btildeb}\otimes\rtilder}\br{\abrshared\otimes\psi^{\otimes n}}}\right.\\-\left.\Tr\Br{\br{M^{(1)}_a\otimes N^{(1)}_b\otimes\rtilder}\br{\abrshared\otimes\psi^{\otimes n}}}\right|\leq\delta.
\end{multline*}
\item For all $a,b,p,q$, $M^{(1)}_{p,a}$ and $N^{(1)}_{q,b}$ have degree at most $d_1$, where $M^{(1)}_{p,a}$ and $N^{(1)}_{q,b}$ are defined in Eq.~\eqref{eqn:defmpa}.
\item \[\frac{1}{\dim{s}^n}\Tr~\zeta\br{M^{(1)}}\leq\delta\quad\mbox{and}\quad\frac{1}{\dim{t}^n}\Tr~\zeta\br{N^{(1)}}\leq\delta,\]
where $\zeta$ is defined in \cref{eqn:zeta}.
\item $M^{(1)}_{0}=\id_{\sS^n\sP}/\sqrt{\dim{a}}$ and $N^{(1)}_{0}=\id_{\sT^n\sQ}/\sqrt{\dim{b}}$.

\end{enumerate}
Here $d_1=O\br{\frac{\dim{a}^2\dim{b}^2\dim{p}\dim{q}}{\delta\br{1-\rho}}}$.

\item \textbf{Regularization}

Applying \cref{lem:regularization} to $M^{(1)}$ and $N^{(1)}$ with $\theta\leftarrow\theta, d\leftarrow d_1$, we obtain $H\subseteq[n]$ of size $h\leq d_1\br{\dim{a}+\dim{b}}/\theta$ such that for all $i\notin H$:
\[\mathrm{Inf}_i\br{M}\leq\theta,\quad\mathrm{Inf}_i\br{N}\leq\theta.\]

\item \textbf{Invariance to random operators}

Applying \cref{lem:mainIP} to $M^{(1)}$, $N^{(1)}$ and $H$, we obtain joint random operators $\mathbf{M}^{(2)}$ and $\mathbf{N}^{(2)}$ satisfying the following.

 \begin{enumerate}
\item For all $a,b,p,q$:\[\expec{}{\nnorm{\mathbf{M}^{(2)}_{p,a}}_2^2}^{1/2}=\nnorm{M^{(1)}_{p,a}}_2\quad\mbox{and}\quad \expec{}{\nnorm{\mathbf{N}^{(2)}_{q,b}}_2^2}^{1/2}=\nnorm{N^{(1)}_{q,b}}_2.\]
\item For all $a,b,r$:
\begin{multline*}
\expec{}{\Tr\Br{\br{\mathbf{M}^{(2)}_{a}\otimes \mathbf{N}^{(2)}_{b}\otimes\rtilder}\br{\abrshared\otimes\psi^{\otimes h}}}}\\=\Tr\Br{\br{M^{(1)}_a\otimes N^{(1)}_b\otimes\rtilder}\br{\abrshared\otimes\psi^{\otimes n}}}.
\end{multline*}

\item
\[\abs{\frac{1}{\dim{s}^h}\expec{}{\Tr~\zeta\br{\mathbf{M}^{(2)}}}-\frac{1}{\dim{s}^n}\Tr~\zeta\br{M^{(1)}}}\leq O\br{\dim{p}^{10/3}\dim{a}^{4}\br{3^{d_1}\dim{s}^{d_1/2}\sqrt{\theta}d_1}^{2/3}}\]
and
\[\abs{\frac{1}{\dim{t}^h}\expec{}{\Tr~\zeta\br{\mathbf{N}^{(2)}}}-\frac{1}{\dim{t}^n}\Tr~\zeta\br{N^{(1)}}}\leq O\br{\dim{q}^{10/3}\dim{b}^{4}\br{3^{d_1}\dim{t}^{d_1/2}\sqrt{\theta}d_1}^{2/3}}.\]
\item $\mathbf{M}^{(2)}_{0}=\id_{\sS^h\sP}/\sqrt{\dim{a}}$ and $\mathbf{N}^{(2)}_{0}=\id_{\sT^h\sQ}/\sqrt{\dim{b}}$.
\end{enumerate}

\item \textbf{Dimension Reduction}

Applying \cref{lem:dimensionreduction} to $\br{\mathbf{M}^{(2)},\mathbf{N}^{(2)}}$ with $\delta\leftarrow\delta/4\br{\dim{a}\dim{b}\dim{p}\dim{q}\dim{r}}^2$, $d\leftarrow d_1$, \linebreak$n\leftarrow\gsnuma$, $\alpha\leftarrow1/8$, if we sample $\mathbf{G}\sim\gamma_{n\times n_0}$, then item 1 to 3 in \cref{lem:dimensionreduction} hold with probability at least $3/4-\delta/2>0$. Thus we get joint random operators $\br{\mathbf{M}^{(3)},\mathbf{N}^{(3)}}$ such that the following holds:

\begin{enumerate}

\item For all $a,b,p,q$:

\[\expec{}{\nnorm{\mathbf{M}^{(3)}_{p,a}}_2^2}\leq\br{1+\delta}\expec{}{\nnorm{\mathbf{M}^{(2)}_{p,a}}_2^2}\quad\mbox{and}\quad \expec{}{\nnorm{\mathbf{N}^{(3)}_{q,b}}_2^2}\leq\br{1+\delta}\expec{}{\nnorm{\mathbf{N}^{(2)}_{q,b}}_2^2}.\]
\item \[\expec{\mathbf{x}}{\Tr~\zeta\br{\mathbf{M}^{(3)}}}\leq8\expec{\mathbf{g}}{\Tr~\zeta\br{\mathbf{M}^{(2)}}}~\mbox{and}~\expec{\mathbf{y}}{\Tr~\zeta\br{\mathbf{N}^{(3)}}}\leq8\expec{\mathbf{h}}{\Tr~\zeta\br{\mathbf{N}^{(2)}}}.\]

\item For all $a,b,r$:

\begin{multline*}
\left|\expec{\mathbf{x},\mathbf{y}}{\Tr\Br{\br{\mathbf{M}^{(3)}_{a}\otimes \mathbf{N}^{(3)}_{b}\otimes\rtilder}\br{\abrshared\otimes\shared{h}}}}\right.\\\left.-\expec{\mathbf{g},\mathbf{h}}{\Tr\Br{\br{\mathbf{M}^{(2)}_{a}\otimes \mathbf{N}^{(2)}_{b}\otimes\rtilder}\br{\abrshared\otimes\shared{h}}}}\right|\leq\delta.
\end{multline*}

\item $\mathbf{M}^{(3)}_{0}=\id_{\sS^h\sP}/\sqrt{\dim{a}}$ and $\mathbf{N}^{(3)}_{0}=\id_{\sT^h\sQ}/\sqrt{\dim{b}}$.
\end{enumerate}

Here $n_0=O\br{\frac{\br{\dim{a}\dim{b}\dim{r}}^{12}\br{\dim{p}\dim{q}}^{20}d_1^{O(d_1)}}{\delta^6}}$.

\item \textbf{Smoothing random operators}

Applying \cref{lem:smoothGaussian} to $\br{\mathbf{M}^{(3)},\mathbf{N}^{(3)}}$ with $\delta\leftarrow\delta,h\leftarrow h,n\leftarrow n_0$, we obtain joint random operators $\br{\mathbf{M}^{(4)},\mathbf{N}^{(4)}}$ satisfying the following.

\begin{enumerate}
\item For all $a,b,p,q$:
\[\deg\br{\mathbf{M}^{(4)}_{p,a}}\leq d_2\quad\mbox{and}\quad \deg\br{\mathbf{N}^{(4)}_{q,b}}\leq d_2.\]

\item For all $a,b,p,q$: \[\expec{}{\nnorm{\mathbf{M}^{(4)}_{p,a}}_2^2}^{1/2}\leq\expec{}{\nnorm{\mathbf{M}^{(3)}_{p,a}}_2^2}^{1/2}\quad\mbox{and}\quad \expec{}{\nnorm{\mathbf{N}^{(4)}_{q,b}}_2^2}^{1/2}\leq\expec{}{\nnorm{\mathbf{N}^{(3)}_{q,b}}_2^2}^{1/2}.\]
\item \[\expec{}{\Tr~\zeta\br{\mathbf{M}^{(4)}}}\leq\expec{}{\Tr~\zeta\br{\mathbf{M}^{(3)}}}+\delta~\mbox{and}~\expec{}{\Tr~\zeta\br{\mathbf{N}^{(4)}}}\leq\expec{}{\Tr~\zeta\br{\mathbf{N}^{(3)}}}+\delta.\]

\item For all $a,b,r$:
\begin{multline*}
\left|\expec{}{\Tr\Br{\br{\mathbf{M}^{(4)}_{a}\otimes \mathbf{N}^{(4)}_{b}\otimes\rtilder}\br{\abrshared\otimes\psi^{\otimes h}}}}\right.\\\left.-\expec{}{\Tr\Br{\br{\mathbf{M}^{(3)}_{a}\otimes \mathbf{N}^{(3)}_{b}\otimes\rtilder}\br{\abrshared\otimes\psi^{\otimes h}}}}\right|\leq\delta.
\end{multline*}
\item $\mathbf{M}^{(4)}_{0}=\id_{\sS^h\sP}/\sqrt{\dim{a}}$ and $\mathbf{N}^{(4)}_{0}=\id_{\sT^h\sQ}/\sqrt{\dim{b}}$.
\end{enumerate}
Here $d_2=O\br{\frac{\dim{a}^2\dim{b}^2\dim{p}\dim{q}}{\delta\br{1-\rho}}}$.

\item \textbf{Multilinearization}

Suppose that
\[\mathbf{M}^{(4)}_{p,a}=\sum_{\srange{h}}m^{(4)}_{s,p,a}\br{\mathbf{x}}\S_s\quad\mbox{and}\quad\mathbf{N}^{(4)}_{q,b}=\sum_{\trange{h}}n^{(4)}_{t,q,b}\br{\mathbf{y}}\T_t.\]

We apply \cref{lem:multiliniearization} to $\br{\mathbf{M}^{(4)},\mathbf{N}^{(4)}}$ with $d\leftarrow d_2,h\leftarrow h,n\leftarrow n_0,\delta\leftarrow\theta$, we obtain multilinear random operators $\br{\mathbf{M}^{(5)},\mathbf{N}^{(5)}}$ such that the following holds:

\begin{enumerate}
						\item For all $a,b,p,q$, $\mathbf{M}^{(5)}_{p,a}$ and $\mathbf{N}^{(5)}_{q,b}$ are degree-$d_2$ multilinear random operators.
			\item Suppose that
\[\mathbf{M}^{(5)}_{p,a}=\sum_{\srange{h}}m^{(5)}_{s,p,a}\br{\mathbf{x}}\S_s\quad\mbox{and}\quad\mathbf{N}^{(5)}_{q,b}=\sum_{\trange{h}}n^{(5)}_{t,q,b}\br{\mathbf{y}}\T_t,\]

where $\br{\mathbf{x},\mathbf{y}}\sim\G_\rho^{\otimes n_0\cdot n_1}$. For all $\br{i,j}\in[n_0]\times[n_1],a,b,p,q,s,t$,
			\[\influence_{(i-1)n_1+j}\br{m^{(5)}_{s,p,a}}\leq\theta\cdot\influence_i\br{m^{(4)}_{s,p,a}}\quad\mbox{and}\quad\influence_{(i-1)n_1+j}\br{n^{(5)}_{t,q,b}}\leq\theta\cdot\influence_i\br{n^{(4)}_{t,q,b}}.\]
			\item For all $a,b$: \[\expec{}{\nnorm{\mathbf{M}^{(5)}_a}_2^2}\leq \expec{}{\nnorm{\mathbf{M}^{(4)}_a}_2^2}\quad\mbox{and}\quad \expec{}{\nnorm{\mathbf{N}^{(5)}_b}_2^2}\leq \expec{}{\nnorm{\mathbf{N}^{(4)}_b}_2^2}.\]
			\item \[\frac{1}{\dim{s}^h}\abs{\expec{}{\Tr~\zeta\br{\mathbf{M}^{(5)}}}-\expec{}{\Tr~\zeta\br{\mathbf{M}^{(4)}}}}\leq\delta \]
			and
			\[\frac{1}{\dim{t}^h}\abs{\expec{}{\Tr~\zeta\br{\mathbf{N}^{(5)}}}-\expec{}{\Tr~\zeta\br{\mathbf{N}^{(4)}}}}\leq \delta.\]
\item For all $a,b,r$:

\begin{multline*}
\left|\expec{}{\Tr\Br{\br{\mathbf{M}^{(5)}_{a}\otimes \mathbf{N}^{(5)}_{b}\otimes\rtilder}\br{\abrshared\otimes\psi^{\otimes h}}}}\right.\\\left.-\expec{}{\Tr\Br{\br{\mathbf{M}^{(4)}_{a}\otimes \mathbf{N}^{(4)}_{b}\otimes\rtilder}\br{\abrshared\otimes\psi^{\otimes h}}}}\right|\leq\delta.
\end{multline*}

\item $\mathbf{M}^{(5)}_0=\id_{\sS^{h}\sP}/\sqrt{\dim{a}}$ and $\mathbf{N}^{(5)}_0=\id_{\sT^{h}\sQ}/\sqrt{\dim{b}}$.
		\end{enumerate}

Here $n_1=O\br{\frac{\dim{a}^4\dim{b}^4\dim{p}^2\dim{q}^2d_2^2}{\theta^2}}$.

\item \textbf{Invariance to operators}
Applying item 2 above, \cref{fac:Gaussianinf} and \cref{lem:randoperator}, we have
\[\sum_{s,p,a}\influence_{i}\br{m^{(5)}_{s,p,a}}\leq\theta\cdot\dim{p}\cdot\dim{a}\cdot \expec{}{\nnorm{\mathbf{M}^{(4)}}_2^2}.\]

Similarly, we have
\[\sum_{t,q,b}\influence_{i}\br{n^{(5)}_{t,q,b}}\leq\theta\cdot \dim{q}\cdot\dim{b}\cdot\expec{}{\nnorm{\mathbf{N}^{(4)}}_2^2}.\]

Let \[\theta_0=\max\set{\theta \expec{}{\nnorm{\mathbf{M}^{(4)}}_2^2},\theta \expec{}{\nnorm{\mathbf{N}^{(4)}}_2^2}}.
\]

We apply \cref{lem:invarianceback} to $\br{\mathbf{M}^{(5)},\mathbf{N}^{(5)}}$ with $n\leftarrow n_0n_1,h\leftarrow h,d\leftarrow d_2,\theta\leftarrow\theta_0$ to get $\br{M^{(6)},N^{(6)}}$ satisfying that:

\begin{enumerate}
\item For all $a,b,p,q$:\[\nnorm{M^{(6)}_{p,a}}_2=\expec{}{\nnorm{\mathbf{M}^{(5)}_{p,a}}_2^2}^{1/2}\quad\mbox{and}\quad \nnorm{N^{(6)}_{q,b}}_2=\expec{}{\nnorm{\mathbf{N}^{(5)}_{q,b}}_2^2}^{1/2}.\]
\item For all $a,b,r$:
\begin{multline*}
\Tr\Br{\br{M^{(6)}_a\otimes N^{(6)}_b\otimes\rtilder}\br{\abrshared\otimes\psi^{\otimes n_0n_1+h}}}\\=\expec{}{\Tr\Br{\br{\mathbf{M}^{(5)}_{a}\otimes \mathbf{N}^{(5)}_{b}\otimes\rtilder}\br{\abrshared\otimes\psi^{\otimes h}}}}.
\end{multline*}

\item
\[
\abs{\frac{1}{\dim{s}^{n_0n_1+h}}\Tr~\zeta\br{M^{(6)}}-\frac{1}{\dim{s}^h}\expec{}{\Tr~\zeta\br{\mathbf{M}^{(5)}}}}\leq O\br{\dim{p}^{10/3}\dim{a}^{4}\br{3^{d_2}\dim{s}^{d_2/2}\sqrt{\theta_0}d_2}^{2/3}}
\]
and
\[
\abs{\frac{1}{\dim{t}^{n_0n_1+h}}\Tr~\zeta\br{N^{(6)}}-\frac{1}{\dim{t}^h}\expec{}{\Tr~\zeta\br{\mathbf{N}^{(5)}}}}\leq O\br{\dim{q}^{10/3}\dim{b}^{4}\br{3^{d_2}\dim{t}^{d_2/2}\sqrt{\theta_0}d_2}^{2/3}}.
\]
\item $M^{(6)}_{0}=\id_{\sS^{n_0n_1+h}\sP}/\sqrt{\dim{a}}$ and $N^{(6)}_{0}=\id_{\sT^{n_0n_1+h}\sQ}/\sqrt{\dim{b}}$.
\end{enumerate}

\item
\textbf{Rounding}

Applying \cref{lem:mainrounding} to $M^{(6)}$ and $N^{(6)}$, we get $\widetilde{M}$ and $\widetilde{N}$ satisfying
\begin{equation}\label{eqn:around}
{\sum_a\nnorm{M^{(6)}_{a}-\widetilde{M_{a}}}_2^2}=\dim{a}\cdot\nnorm{M^{(6)}-\widetilde{M}}_2^2\leq O\br{\br{\frac{\dim{a}^7}{\dim{p}\dim{s^D}}\Tr~\zeta\br{M^{(6)}}}^{1/2}},
\end{equation}

\begin{equation}\label{eqn:bround}
{\sum_b\nnorm{N^{(6)}_{b}-\widetilde{N_{b}}}_2^2}=\dim{b}\cdot\nnorm{N^{(6)}-\widetilde{N}}_2^2\leq O\br{\br{\frac{\dim{b}^7}{\dim{q}\dim{t^D}}\Tr~\zeta\br{N^{(6)}}}^{1/2}}.
\end{equation}

 Let $D=h+n_0n_1$. Then
\begin{align*}
&\abs{\Tr\Br{\br{M^{(6)}_{a}\otimes N^{(6)}_{b}\otimes \rtilder-\widetilde{M_{a}}\otimes \widetilde{N_{b}}\otimes \rtilder}\br{\abrshared\otimes\psi^{\otimes D}}}}\\
\leq~&\abs{\Tr\Br{\br{M^{(6)}_{a}\otimes \br{N^{(6)}_{b}-\widetilde{N_{b}}}\otimes \rtilder}\br{\abrshared\otimes\psi^{\otimes D}}}}\\
&+\abs{\Tr\Br{\br{\br{M^{(6)}_{a}-\widetilde{M_{a}}}\otimes \widetilde{N_{b}}\otimes \rtilder}\br{\abrshared\otimes\psi^{\otimes D}}}}\\
\overset{(\star)}{\leq}~&\br{\dim{p}\dim{q}}^{1/2}\br{\nnorm{M^{(6)}_{a}}_2\nnorm{N^{(6)}_{b}-\widetilde{N_{b}}}_2+\nnorm{M^{(6)}_{a}-\widetilde{M_{a}}}_2\nnorm{\widetilde{N_{b}}}_2}\\
\leq~&\br{\dim{p}\dim{q}}^{1/2}\br{\nnorm{M^{(6)}_{a}}_2\br{\sum_b\nnorm{N^{(6)}_{b}-\widetilde{N_{b}}}_2^2}^{1/2}\hspace*{-1em}+\br{\sum_a\nnorm{M^{(6)}_{a}-\widetilde{M_{a}}}_2^2}^{1/2}\nnorm{\widetilde{N_{b}}}_2}\\
\overset{(\star\star)}{\leq}~&\br{\nnorm{M^{(6)}_{a}}_2O\br{\br{\frac{\dim{b}^{7}\dim{p}^2\dim{q}}{\dim{t}^D}\Tr~\zeta\br{N^{(6)}}}^{1/4}}+\nnorm{\widetilde{N_{b}}}_2O\br{\br{\frac{\dim{a}^{7}\dim{p}\dim{q}^2}{\dim{s}^D}\Tr~\zeta\br{M^{(6)}}}^{1/4}}},\\
\end{align*}
where $(\star)$ is by \cref{lem:roundingneed} and $(\star\star)$ is by \cref{eqn:around} and \cref{eqn:bround}.
\end{itemize}

Keeping track of the parameters in the construction, we are able to upper bound \linebreak$\Tr~\zeta\br{M^{(6)}}/\dim{s}^D$ and $\Tr~\zeta\br{N^{(6)}}/\dim{t}^D$. Finally, by the triangle inequality we are able to upper bound
\[
\left|\Tr\Br{\br{\alice^*\br{\atildea}\otimes\bob^*\br{\btildeb}\otimes\R_r}\br{\abrshared\otimes\psi^{\otimes n}}}-\Tr\Br{\br{\widetilde{M_{a}}\otimes \widetilde{N_{b}}\otimes\R_r}\br{\abrshared\otimes\psi^{\otimes D}}}\right|
\]
The dependency of the parameters is pictorially described in \cref{fig:dependence}.

\newcommand{\rectab}[1]{\begin{tabular}{|c|}\hline#1\\\hline\end{tabular}}
\setlength{\tabcolsep}{1pt}
\begin{figure}
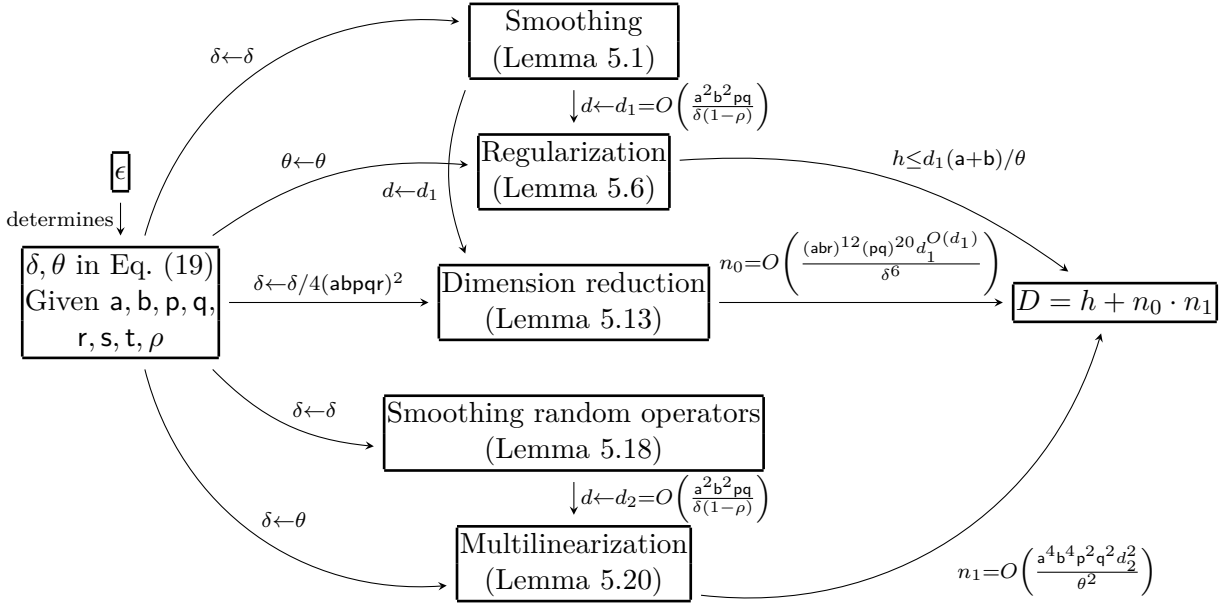

\begin{center}
\begin{codi}

\obj{
&[11em] |(B)|\rectab{Smoothing\\~~(\cref{lem:smoothing})~~~}&[14em]\\
|(H)|\rectab{$\epsilon$}&|(C)|\rectab{Regularization\\(\cref{lem:regularization})}&\\
|(A)|\rectab{$\delta,\theta$ in \cref{eqn:parameters}\\Given $\dim{a},\dim{b},\dim{p},\dim{q},$\\$\dim{r},\dim{s},\dim{t},\rho$} & |(D)|\rectab{Dimension reduction\\(\cref{lem:dimensionreduction})}&|(E)|\rectab{$D=h+n_0\cdot n_1$}\\
&|(F)|\rectab{Smoothing random operators\\(\cref{lem:smoothGaussian})}&\\
&|(G)|\rectab{Multilinearization\\(\cref{lem:multiliniearization})}&\\
};
\mor :[swap] H \mathrm{{determines}}:-> A;
\mor :[bend left = 40] A \delta\leftarrow\delta :-> B;
\mor :[bend left = 20] * \theta\leftarrow\theta :-> C;
\mor * \delta\leftarrow\delta/4(\dim{a}\dim{b}\dim{p}\dim{q}\dim{r})^2:-> D;
\mor :[bend right = 20] * \delta\leftarrow\delta:-> F;
\mor :[bend right = 40] * \delta\leftarrow\theta:-> G;
\mor B d\leftarrow d_1=O\br{{\frac{{\dim{a}^2\dim{b}^2\dim{p}\dim{q}}}{{\delta(1-\rho)}}}}:-> C;
\mor :[bend left = 20] C h\leq d_1(\dim{a}+\dim{b})/\theta :-> E;
\mor :D n_0=O\br{{\frac{{(\dim{a}\dim{b}\dim{r})^{12}(\dim{p}\dim{q})^{20}d_1^{O(d_1)}}}{{\delta^6}}}} :-> *;
\mor F d\leftarrow d_2=O\br{{\frac{{\dim{a}^2\dim{b}^2\dim{p}\dim{q}}}{{\delta(1-\rho)}}}} :-> G;
\mor :[swap, bend right = 40] G n_1=O\br{{\frac{{\dim{a}^4\dim{b}^4\dim{p}^2\dim{q}^2d_2^2}}{{\theta^2}}}} :-> E;
\mor :[swap, bend right = 20, shove = +3.1em] B ["\begin{tabular}{c}\\$d\leftarrow d_1$\end{tabular}"] :-> D;
\end{codi}
\end{center}\caption{Dependency of parameters in the proof of \cref{thm:nijs}}\label{fig:dependence}
\end{figure}

We define $\Psi_{\mathsf{Alice}}\in\L\br{\sA,\sS^{D}\sP}$, $\Psi_{\mathsf{Bob}}\in\L\br{\sB,\sT^{D}\sQ}$ as follows:
\[\Psi_{\mathsf{Alice}}\br{X}=\Tr_{\sA}\br{\widetilde{M}\br{\id_{\sS^{D}\sP}\otimes X^\dagger}},\]
\[\Psi_{\mathsf{Bob}}\br{Y}=\Tr_{\sB}\br{\widetilde{N}\br{\id_{\sT^{D}\sQ}\otimes Y^\dagger}},\]
just as \cref{eqn:choitophi}. Let $\widetilde{\alice}=\Psi_{\mathsf{Alice}}^*$ and $\widetilde{\bob}=\Psi_{\mathsf{Bob}}^*$. Then by \cref{fac:adjointchoi}, $\widetilde{\alice}$ and $\widetilde{\bob}$ are quantum operations. Furthermore,
\begin{multline*}
\Tr\Br{\br{\br{\widetilde{\alice}}^*\br{\atildea}\otimes\br{\widetilde{\bob}}^*\br{\btildeb}\otimes\rtilder}\br{\abrshared\otimes\psi^{\otimes D}}}\\=\Tr\Br{\br{\widetilde{M_{a}}\otimes \widetilde{N_{b}}\otimes\rtilder}\br{\abrshared\otimes\psi^{\otimes D}}}.
\end{multline*}

And by \cref{fac:adjointchoi} and \cref{lem:basisnormbound},
\[\nnorm{\widetilde{N_{b}}}_2=\nnorm{\br{\widetilde{\bob}}^*\br{\btildeb}}_2\leq1.\]

Choosing \begin{equation}\label{eqn:parameters}
\delta=O(\epsilon),\quad\theta=\frac{\epsilon^{12}}{\exp\br{\frac{\dim{a}^2\dim{b}^2\dim{p}\dim{q}\log\dim{s}\log\dim{t}}{\epsilon(1-\rho)}}},
\end{equation}
we finally conclude the result.
\end{proof}

\section{Construction}
\cref{thm:nijs} states that, given an arbitrarily dimensional strategy, it can be simulated by a strategy with a bounded number of shared noisy maximally entangled states. This section focuses on the construction of the new strategies, which consists of several steps. We adopt the notations and the setup given in~\cref{subsec:notationsetup}.

\subsection{Smoothing}

\begin{lemma}\label{lem:smoothing}
\addsetup let $\delta\in\br{0,1}$, $n\in\posint$. There exists an explicitly computable $d=d\br{\rho,\delta,\dim{a},\dim{b},\dim{p},\dim{q}}$  and maps $f:\hspa{n}\rightarrow \hspa{n}$, $g:\htqb{n}\rightarrow \htqb{n}$ such that for all quantum operations $\alice\in\L\br{\sS^n\sP,\sA},\bob\in\L\br{\sT^n\sQ,\sB}$, 
 denoting $M=\choi{\alice^*}$ and $N=\choi{\bob^*}$, the operators $M^{(1)}=f\br{M}$ and $N^{(1)}=g\br{N}$ satisfy the following.
\begin{enumerate}
\item  For all $a,b$:\[\nnorm{M^{(1)}_a}_2\leq1\quad\mbox{and}\quad\nnorm{N^{(1)}_b}_2\leq1,\]
    where $M^{(1)}_a$ and $N^{(1)}_b$ are defined in Eq.~\eqref{eqn:defMa}.
\item
 For all $a,b,r$:
\begin{multline*}
\left|\Tr\Br{\br{\alice^*\br{\atildea}\otimes\bob^*\br{\btildeb}\otimes\rtilder}\br{\abrshared\otimes\psi^{\otimes n}}}\right.\\-\left.\Tr\Br{\br{M^{(1)}_a\otimes N^{(1)}_b\otimes\rtilder}\br{\abrshared\otimes\psi^{\otimes n}}}\right|\leq\delta.
\end{multline*}
\item For all $a,b,p,q$, $M^{(1)}_{p,a}$ and $N^{(1)}_{q,b}$ have degree at most $d$, where $M^{(1)}_{p,a}$ and $N^{(1)}_{q,b}$ are defined in Eq.~\eqref{eqn:defmpa}.
\item \[\frac{1}{\dim{s}^n}\Tr~\zeta\br{M^{(1)}}\leq\delta\quad\mbox{and}\quad\frac{1}{\dim{t}^n}\Tr~\zeta\br{N^{(1)}}\leq\delta,\]
where $\zeta$ is defined in \cref{eqn:zeta}.
\item $M^{(1)}_{0}=\id_{\sS^n\sP}/\sqrt{\dim{a}}$ and $N^{(1)}_{0}=\id_{\sT^n\sQ}/\sqrt{\dim{b}}$.
\end{enumerate}
In particular, one may take $d=O\br{\frac{\dim{a}^2\dim{b}^2\dim{p}\dim{q}}{\delta\br{1-\rho}}}$.
\end{lemma}

\begin{remark}
By \cref{lem:closedelta1}, $\zeta(\cdot)$ describes the distance to the set of positive operators. Thus, item 4 implies that $M^{(1)}$ and $N^{(1)}$ are still close to positive operators after the smoothing operation. By \cref{fac:adjointchoi}, item 4 and item 5 together imply that $M^{(1)}$ and $N^{(1)}$ are close to the Choi representations of adjoints of quantum operations.
\end{remark}

\begin{proof}[Proof of \cref{lem:smoothing}]

By \cref{eqn:defMa} and the definition of the Choi representation in \cref{eqn:choi}, we have
\[M_a=\alice^*\br{\atildea},\quad N_b=\bob^*\br{\btildeb},\]
for all $a,b$. By \cref{lem:basisnormbound}, $\nnorm{M_a}_2\leq1$ and $\nnorm{N_b}_2\leq1$.

Define parameters
\[\delta'=\frac{\delta}{4\dim{a}^2\dim{b}^2\dim{p}\dim{q}},\quad\gamma=1-C\frac{\br{1-\rho}\delta'}{\log(1-\delta')},\quad d=O\br{\frac{\log^2(1/\delta')}{\delta'(1-\rho)}}.\]
Set
\[M'=\br{\Delta_\gamma\otimes\id_{\sP\sA}}M,\]
where the noise operator acts on $\sS^n$. Similarly, define
\[N'=\br{\Delta_\gamma\otimes\id_{\sQ\sB}}N,\]
where the noise operator acts on $\sT^n$. That is, for all $a,b,r$,

\[M'_{p,a}=\Delta_\gamma\br{M_{p,a}},\quad N'_{q,b}=\Delta_\gamma\br{N_{q,b}}.\]

By \cref{fac:adjointchoi}, $M,N\geq0$. By \cref{lem:bonamibecknerdef} item 3, $M',N'\geq0$ as well.

For all $a,b,p,q$, define
\begin{equation}\label{eqn:pqlowedeg}
  M^{(1)}_{p,a}=\br{M'_{p,a}}^{\leq d}, \quad N^{(1)}_{q,b}=\br{N'_{q,b}}^{\leq d}.
\end{equation}

Each item in Lemma~\ref{lem:smoothing} is proved as follows.

\begin{enumerate}
  \item By \cref{fac:basicfourier} and \cref{lem:bonamibecknerdef} item 2,
      \begin{equation}\label{eqn:pprimesigmaupperbound}
        \nnorm{M^{(1)}_a}_2\leq\nnorm{M'_{a}}_2\leq\nnorm{M_a}_2\leq1,\quad\nnorm{N^{(1)}_b}_2\leq\nnorm{N'_{b}}_2\leq\nnorm{N_b}_2\leq1.
      \end{equation}

\item By \cref{lem:Tsmooth} and the choice of parameters, we have
\[
\abs{\Tr\Br{\br{M_a\otimes N_b\otimes\rtilder-M'_a\otimes N'_b\otimes\rtilder}\br{\abrshared\otimes\psi^{\otimes n}}}}\leq\delta/2.
\]
    Using \cref{lem:cutoff}, we obtain
    \[
\abs{\Tr\Br{\br{M'_a\otimes N'_b\otimes\rtilder-M^{(1)}_a\otimes N^{(1)}_b\otimes\rtilder}\br{\abrshared\otimes\psi^{\otimes n}}}}\leq\delta/2.
\]
    By the triangle inequality, we have
\begin{align*}
&\abs{\Tr\Br{\br{M_a\otimes N_b\otimes\rtilder-M^{(1)}_a\otimes N^{(1)}_b\otimes\rtilder}\br{\abrshared\otimes\psi^{\otimes n}}}}\\
\leq~&\abs{\Tr\Br{\br{M_a\otimes N_b\otimes\rtilder-M'_a\otimes N'_b\otimes\rtilder}\br{\abrshared\otimes\psi^{\otimes n}}}}\\
&+~\abs{\Tr\Br{\br{M'_a\otimes N'_b\otimes\rtilder-M^{(1)}_a\otimes N^{(1)}_b\otimes\rtilder}\br{\abrshared\otimes\psi^{\otimes n}}}}\\
\leq~&\delta/2+\delta/2=\delta.
\end{align*}

  \item It holds by the definition of $M^{(1)}_{p,a}$ and $N^{(1)}_{q,b}$ in Eq.~\eqref{eqn:pqlowedeg}.

  \item  From Eq.~\eqref{eqn:pprimesigmaupperbound}
  \begin{equation*}
    \nnorm{M'}_2^2=\frac{1}{\dim{a}}\sum_{a\in[\dim{a}^2]_{\geq 0}}\nnorm{M'_{a}}_2^2\leq\dim{a},\quad\nnorm{N'}_2^2=\frac{1}{\dim{b}}\sum_{b\in[\dim{b}^2]_{\geq 0}}\nnorm{N'_{b}}_2^2\leq\dim{b}.
  \end{equation*}
  By \cref{lem:bonamibecknerdef} item 4,
\[\nnorm{\br{M_a'}^{> d}}_2\leq\gamma^{d},\quad \nnorm{\br{N_b'}^{> d}}_2\leq\gamma^{d}.\]

Thus
\[\nnorm{\br{M'}^{>d}}_2^2=\frac{1}{\dim{a}}\sum_{a\in[\dim{a}^2]_{\geq 0}}\nnorm{\br{M'_{a}}^{>d}}_2^2\leq\dim{a}\gamma^{2d},\]
\[\nnorm{\br{N'}^{>d}}_2^2=\frac{1}{\dim{b}}\sum_{b\in[\dim{b}^2]_{\geq 0}}\nnorm{\br{N'_{b}}^{>d}}_2^2\leq\dim{b}\gamma^{2d}.\]
By \cref{lem:zetaadditivity}, and the fact that $\zeta\br{P'}=\zeta\br{Q'}=0$ since $P'$ and $Q'$ are positive,
\[\frac{1}{\dim{s}^n}\Tr~\zeta\br{M^{(1)}}\leq4\dim{a}\nnorm{M'}_2\nnorm{\br{M'}^{>d}}_2\leq4\dim{a}^2\gamma^{d}\leq\delta\]and\[\frac{1}{\dim{t}^n}\Tr~\zeta\br{N^{(1)}}\leq4\dim{b}\nnorm{N'}_2\nnorm{\br{N'}^{>d}}_2\leq4\dim{b}^2\gamma^{d}\leq\delta.\]

  \item By \cref{fac:adjointchoi}, $\alice^*$ and $\bob^*$ are unital. Thus \[M_0=\alice^*\br{\widetilde{\A_0}}=\id_{\sS^n\sP}/\sqrt{\dim{a}},\quad N_0=\bob^*\br{\widetilde{\B_0}}=\id_{\sT^n\sQ}/\sqrt{\dim{b}}.\] Since $\Delta_{\gamma}$ is also unital, item 5 follows.
\end{enumerate}

%

\end{proof}
\begin{lemma}\label{lem:basisnormbound}
Given quantum systems $\sA$ and $\sS$, let $\Phi\in\L\br{\sA,\sS}$ be positive and unital, and $H\in\H_\sA$ satisfy $\Tr\,H^2\leq1$. Then
\[\nnorm{\Phi\br{H}}_2\leq1.\]
\end{lemma}
\begin{proof}
Since $\Tr\,H^2\leq1$, we have $H\leq\id_{\sA}$. By the positivity of $\Phi$, $\Phi\br{\id_{\sA}-H}\geq0$. The fact that $\Phi$ is unital implies $\Phi\br{H}\leq\id_{\sS}$, from which the result concludes.
\end{proof}

%
%

\begin{lemma}\label{lem:Tsmooth}
	\addsetup let $\epsilon\in(0,1)$, $n\in\posint$, $M\in\hspa{n}$, $N\in\htqb{n}$. It holds that for all $a,b,r$,
	\[
\abs{\Tr\Br{\br{M_a\otimes N_b\otimes\rtilder-\Delta_\gamma^{\sS}\br{M_a}\otimes \Delta^{\sT}_\gamma\br{N_b}\otimes\rtilder}\br{\abrshared\otimes\psi^{\otimes n}}}}\leq\epsilon\br{\dim{p}\dim{q}}^{1/2}\nnorm{M_a}_2\nnorm{N_b}_2
\]
	for
\begin{equation*}
  \br{1-\epsilon}^{\log\rho/\br{\log\epsilon+\log\rho}}\leq \gamma\leq1,
\end{equation*}
where $\Delta^{\sS}_\gamma$ and $\Delta^{\sT}_\gamma$ are the depolarizing channels $\Delta_{\gamma}$ acting on the $n$-copies of systems $\sS$ and $\sT$, respectively.

In particular, there exists an absolute constant $C$ such that it suffices to take \[\gamma=1-C\frac{\br{1-\rho}\epsilon}{\log(1-\epsilon)}.\]
\end{lemma}
\begin{proof}
By \cref{lem:pqrdiag}, for any $r$, we may choose bases $\set{\P_p}_{\prange}$, $\set{\Q_q}_{\qrange}$ satisfying \cref{eqn:pqdiag}.
Suppose that for all $a,b,p,q$, $M_{p,a}$ and $N_{q,b}$ have Fourier expansions
\[M_{p,a}=\sum_{\srange{n}}\widehat{M_{p,a}}\br{s}\S_{s},\quad N_{q,b}=\sum_{\trange{n}}\widehat{N_{q,b}}\br{t}\T_{t}.\]

Then by \cref{lem:bonamibecknerdef} item 1,
\[\Delta_\gamma\br{M_{p,a}}=\sum_{\srange{n}}\widehat{M_{p,a}}\br{s}\gamma^{\abs{s}}\S_{s},\quad \Delta_\gamma\br{N_{q,b}}=\sum_{\trange{n}}\widehat{N_{q,b}}\br{t}\gamma^{\abs{t}}\T_{t}.\]
Note that our choice of $\gamma$ gives us that, $\rho^d\br{1-\gamma^{2d}}\leq\epsilon$ for all $d\in\posint.$ Then
\begin{align*}
&\abs{\Tr\Br{\br{M_a\otimes N_b\otimes\rtilder-\Delta^{\sS}_\gamma\br{M_a}\otimes \Delta^{\sT}_\gamma\br{N_b}\otimes\rtilder}\br{\abrshared\otimes\psi^{\otimes n}}}}\\
=~&\abs{\sum_{p,q,s,t}\widehat{M_{p,a}}\br{s}\widehat{N_{q,b}}\br{t}\br{1-\gamma^{\abs{s}+\abs{t}}}\Tr\Br{\br{\S_s\otimes\T_t}\psi^{\otimes n}}\cdot\Tr\Br{\br{\ptildep\otimes\qtildeq\otimes\rtilder}\abrshared}}\\
\overset{(\star)}{=}~&\abs{\sum_{p,s}\widehat{M_{p,a}}\br{s}\widehat{N_{p,b}}\br{s}\br{1-\gamma^{2\abs{s}}}c_s\cdot k_p}\\
\overset{(\star\star)}{\leq}~&\br{\sum_{p,s}\widehat{M_{p,a}}\br{s}^2\br{1-\gamma^{2\abs{s}}}\rho^{\abs{s}}}^{1/2}\cdot\br{\sum_{p,s}\widehat{N_{p,b}}\br{s}^2\br{1-\gamma^{2\abs{s}}}\rho^{\abs{s}}}^{1/2}\\
\leq~&\epsilon\br{\sum_{p,s}\widehat{M_{p,a}}\br{s}^2}^{1/2}\cdot\br{\sum_{p,s}\widehat{N_{p,b}}\br{s}^2}^{1/2}\\
\overset{(\star\star\star)}{\leq}~&\epsilon\br{\sum_{p}\nnorm{M_{p,a}}_2^2}^{1/2}\cdot\br{\sum_{p}\nnorm{N_{p,b}}_2^2}^{1/2}\\
\leq~&\epsilon\br{\dim{p}\dim{q}}^{1/2}\nnorm{M_a}_2\nnorm{N_b}_2,
\end{align*}
where ($\star$) is by the choices of $\set{\P_p}_{\prange},\set{\Q_q}_{\qrange},\set{\R_r}_{\rrange},\set{\S_s}_{s\in\Br{\dim{s}^2}_{\geq0}},\set{\T_t}_{t\in\Br{\dim{t}^2}_{\geq0}}$ in \cref{setup}, which satisfy \cref{eqn:propersob} and \cref{eqn:pqdiag}; ($\star\star$) is by Cauchy-Schwartz inequality; ($\star\star\star$) is by \cref{fac:basicfourier} item 2.
\end{proof}

\begin{lemma}\label{lem:cutoff}
	Given $d\in\posint$,
operators $M,M'\in\hspa{n}, N,N'\in\htqb{n}$ satisfying that for all $a,b,p,q$,
\[M'_{p,a}=M_{p,a}^{\leq d},\quad N'_{q,b}=N_{q,b}^{\leq d},\]
where $M_{p,a}, N_{p,a}, M_{p,a}^{\leq d}$ and $N_{q,b}^{\leq d}$ are defined in~\Cref{def:lowdegreehighdegree} and Eq.~\eqref{eqn:defmpa}.
 it holds that
\[
\abs{\Tr\Br{\br{M_a\otimes N_b\otimes\rtilder-M'_a\otimes N'_b\otimes\rtilder}\br{\abrshared\otimes\psi^{\otimes n}}}}\leq\rho^d\br{\dim{p}\dim{q}}^{1/2}\nnorm{M_a}_2\nnorm{N_b}_2
\]
\end{lemma}

\begin{proof}
By \cref{lem:pqrdiag}, for any $r$, we can choose bases $\set{\P_p}_{\prange}$,
$\set{\Q_q}_{\qrange}$ satisfying \cref{eqn:pqdiag}.
By~\cref{fac:pt}, $M',N'$ are independent of the choice of the bases $\set{\S_s}_{\srange{}}$ and $\set{\T_t}_{\trange{}}$ in $\sS$ and $\sT$, respectively.
Then
\begin{align*}
&\abs{\Tr\Br{\br{M_a\otimes N_b\otimes\rtilder-M'_a\otimes N'_b\otimes\rtilder}\br{\abrshared\otimes\psi^{\otimes n}}}}\\
=~&\left|\sum_{p,q,s,t}\widehat{M_{p,a}}\br{s}\widehat{N_{q,b}}\br{t}\Tr\Br{\br{\S_s\otimes\T_t}\psi^{\otimes n}}\cdot\Tr\Br{\br{\ptildep\otimes\qtildeq\otimes\rtilder}\abrshared}\right.\\
&-\left.\sum_{p,q,s,t:\abs{s}\leq d,\abs{t}\leq d}\widehat{M_{p,a}}\br{s}\widehat{N_{q,b}}\br{t}\Tr\Br{\br{\S_s\otimes\T_t}\psi^{\otimes n}}\cdot\Tr\Br{\br{\ptildep\otimes\qtildeq\otimes\rtilder}\abrshared}\right|\\
\overset{(\star)}{=}~&\abs{\sum_{p,s:\abs{s}>d}\widehat{M_{p,a}}\br{s}\widehat{N_{p,b}}\br{s}c_s\cdot k_p}\\
\overset{(\star\star)}{\leq}~&\rho^d\br{\sum_{p,s:\abs{s}>d}\widehat{M_{p,a}}\br{s}^2}^{1/2}\cdot\br{\sum_{p,s:\abs{s}>d}\widehat{N_{p,b}}\br{s}^2}^{1/2}\\
\overset{(\star\star\star)}{\leq}~&\rho^d\br{\sum_{p}\nnorm{M_{p,a}}_2^2}^{1/2}\cdot\br{\sum_{p}\nnorm{N_{p,b}}_2^2}^{1/2}\\
\leq~&\rho^d\br{\dim{p}\dim{q}}^{1/2}\nnorm{M_a}_2\nnorm{N_b}_2,
\end{align*}
where ($\star$) is by \cref{eqn:propersob} and \cref{eqn:pqdiag}, ($\star\star$) is by Cauchy-Schwarz inequality and ($\star\star\star$) is by \cref{fac:basicfourier} item 2.
\end{proof}

\subsection{Regularization}
\begin{lemma}\label{lem:regularization}
\addsetup let $\theta\in\br{0,1}$, $d,n\in\posint$, operators $M\in\hspa{n}$ and $N\in\htqb{n}$ satisfy that $M_{p,a},$ $N_{q,b}$ have degree at most $d$ for all $a,b,p,q,r$. Furthermore, $\nnorm{M_a}_2\leq1$, $\nnorm{N_b}_2\leq1$ for all $a,b$.
Then there exists a subset $H\subseteq[n]$ of size $h\leq d\br{\dim{a}+\dim{b}}/\theta$ such that for all $i$-th $\sS$ or $\sT$ system, $i\notin H$,
\[\mathrm{Inf}_i\br{M}\leq\theta,\quad\mathrm{Inf}_i\br{N}\leq\theta.\]

\end{lemma}
\begin{proof}
Set \[H_{M}=\set{i:\mathrm{Inf}_i\br{M}>\theta}\]
of size $h_M$. Then
\begin{align*}
\theta h_M~\leq~&\sum_{i=1}^n\influence_i\br{M}\\
=~&\sum_{i=1}^n\nnorm{M-\id_\sS\otimes\Tr_iM}_2^2\quad\quad\mbox{(by \cref{eqn:definf})}\\
=~&\frac{1}{\dim{p}\dim{a}}\sum_{i=1}^n\sum_{\arange\atop p\in\Br{\dim{p}^2}_{\geq 0}}\nnorm{M_{p,a}-\id_\sS\otimes\Tr_iM_{p,a}}_2^2\quad\quad\mbox{(by Eqs.~\eqref{eqn:defMa}\&\eqref{eqn:defmpa})}\\
=~&\frac{1}{\dim{p}\dim{a}}\sum_{i=1}^n\sum_{\arange\atop p\in\Br{\dim{p}^2}_{\geq 0}}\influence_i\br{M_{p,a}}\quad\quad\mbox{(by \cref{eqn:definf})}\\
\leq~&\frac{d}{\dim{p}\dim{a}}\sum_{\arange\atop p\in\Br{\dim{p}^2}_{\geq 0}}\nnorm{M_{p,a}}_2^2\quad\quad\mbox{(by \cref{fac:partialvariance} item 2)}\\
=~&\frac{d}{\dim{a}}\sum_{\arange}\nnorm{M_{a}}_2^2\quad\quad\mbox{(by Eqs.~\eqref{eqn:defMa}\&\eqref{eqn:defmpa})}\\
\leq~&d\dim{a}.
\end{align*}
Similarly,
\[H_{N}=\set{i:\mathrm{Inf}_i\br{N}>\theta}\]
is of size $h_N\leq\frac{d\dim{b}}{\theta}$. Define
$H=H_M\cup H_N$, the conclusion holds by a union bound.
\end{proof}


\subsection{Invariance principle}

The following is the main result in this subsection.

\begin{lemma}\label{lem:mainIP}
\addsetup let $\theta\in\br{0,1}$, $n,d\in\posint$, $H\subseteq[n]$ of size $h$. There exist maps
$f:\hspa{n}\times\reals^{\gsnuma}\rightarrow L^2\br{\hspa{h},\gamma_{\gsnuma}}$ and $g:\htqb{n}\times\reals^{\gsnumb}\rightarrow L^2\br{\htqb{h},\gamma_{\gsnumb}}$ satisfying the following.

For any $M\in\hspa{n}$, $N\in\htqb{n}$ we define
\[\br{\mathbf{M},\mathbf{N}}=\br{f\br{M,\mathbf{g}},g\br{N,\mathbf{h}}}_{\br{\mathbf{g},\mathbf{h}}\sim\G_\rho^{\otimes\gsnuma}}.\]
If $M,N$ satisfy that
\begin{enumerate}
\item For all $a,b$, $\nnorm{M_a}_2^2\leq1$ and $\nnorm{N_b}_2^2\leq1$;
\item For all $a,b,p,q$, $M_{a,p}$ and $N_{b,q}$ have degree at most $d$;
\item For all $i\notin H$, $\mathrm{Inf}_i\br{M}\leq\theta$, $\mathrm{Inf}_i\br{N}\leq\theta$;
\item $M_{0}=\id_{\sS^n\sP}/\sqrt{\dim{a}},N_{0}=\id_{\sT^n\sQ}/\sqrt{\dim{b}}$,
\end{enumerate}
where $M_a$ and $N_b$ are defined in Eq.~\eqref{eqn:defMa}, then the following holds:
 \begin{enumerate}
\item For all $a,b,p,q$, $\mathbf{M}_{p,a}\in L^2\br{\H_{\sS^h},\gamma_{\gsnuma}}$ and $\mathbf{N}_{q,b}\in L^2\br{\H_{\sT^h},\gamma_{\gsnumb}}$ are degree-$d$ multilinear joint random operators with the joint random variables drawn from $\G_\rho^{\otimes \gsnuma}$, where $\mathbf{M}_{p,a}$ and $\mathbf{N}_{q,b}$ are defined in Eq.~\eqref{eqn:defmpa}.

\item For all $a,b,p,q$:\[\expec{}{\nnorm{\mathbf{M}_{p,a}}_2^2}=\nnorm{M_{p,a}}_2^2\quad\mbox{and}\quad \expec{}{\nnorm{\mathbf{N}_{q,b}}_2^2}=\nnorm{N_{q,b}}_2^2.\]

\item For all $a,b,r$:
\[\expec{}{\Tr\Br{\br{\mathbf{M}_{a}\otimes \mathbf{N}_{b}\otimes\rtilder}\br{\abrshared\otimes\psi^{\otimes h}}}}=\Tr\Br{\br{M_a\otimes N_b\otimes\rtilder}\br{\abrshared\otimes\psi^{\otimes n}}}.\]

\item
\[\abs{\frac{1}{\dim{s}^h}\expec{}{\Tr~\zeta\br{\mathbf{M}}}-\frac{1}{\dim{s}^n}\Tr~\zeta\br{M}}\leq O\br{\dim{p}^{10/3}\dim{a}^{4}\br{3^d\dim{s}^{d/2}\sqrt{\theta}d}^{2/3}}\]
and
\[\abs{\frac{1}{\dim{t}^h}\expec{}{\Tr~\zeta\br{\mathbf{N}}}-\frac{1}{\dim{t}^n}\Tr~\zeta\br{N}}\leq O\br{\dim{q}^{10/3}\dim{b}^{4}\br{3^d\dim{t}^{d/2}\sqrt{\theta}d}^{2/3}}.\]

\item $\mathbf{M}_{0}=\id_{\sS^h\sP}/\sqrt{\dim{a}}$ and $\mathbf{N}_{0}=\id_{\sT^h\sQ}/\sqrt{\dim{b}}$.
\end{enumerate}
\end{lemma}

Before proving the lemma, we need to introduce the quantum invariance principle for the function $\zeta\br{\cdot}$ defined as \cref{eqn:zeta}.

\begin{definition}\label{def:ud}
Given $k,l\in\posint$, for all multilinear random operator $\mathbf{M}\in L^2\br{\hspa{k},\gamma_l}$, assume that for all $p,a$,
\[\mathbf{M}_{p,a}=\sum_{\srange{k}}m_{s,p,a}\br{\mathbf{x}}\S_s,\]
where $x\sim\gamma_l$, and $m_{s,p,a}$ is multilinear for all $s,p,a$. For all $d\in\posint$, define \[u_d=\sup\set{\frac{\expec{}{\nnorm{\mathbf{M}}_4^4}^{1/4}}{\expec{}{\nnorm{\mathbf{M}}_2^2}^{1/2}}:\mathbf{M}\in L^2\br{\hspa{k},\gamma_l},k>0,\ell>0, \deg\br{m_{s,p,a}}+\abs{s}\leq d\text{ for all }s,p,a}.\]

\end{definition}

The following is the quantum invariance principle for the function $\zeta\br{\cdot}$.

\begin{lemma}\label{lem:invarianceH}\cite[Lemma 10.10]{qin2021nonlocal}\footnote{The statement is slightly different from that of \cite{qin2021nonlocal}. But the proof is exactly the same.}
	\addsetup let $\theta\in\br{0,1}$, $n,d,h\in\posint$, $H\subseteq[n]$ of size $\abs{H}=h$, $M\in\hspa{n}$ with the expansion as Eqs.~\eqref{eqn:defMa}\&\eqref{eqn:defmpa} satisfies the following.
	\begin{enumerate}
\item For all $a$, $\nnorm{M_a}_2^2\leq1$.
\item For all $a,p$, $M_{a,p}$ have degree at most $d$.
\item For all $i\notin H$,\[\mathrm{Inf}_i\br{M}\leq\theta.\]
\end{enumerate}
	Suppose $M_{p,a}$ has a Fourier expansion
\[M_{p,a}=\sum_{\srange{k}}\widehat{M_{p,a}}\br{s}\S_s.\]
Define
	 \[\mathbf{M}_{p,a}=\sum_{\srange{k}}\widehat{M_{p,a}}\br{s}\prod_{i\notin H}\mathbf{g}_{i,s_i}\br{\bigotimes_{i\in H}\S_{s_i}}.\]
 Then it holds that
\[\abs{\frac{1}{\dim{s}^h}\expec{}{\Tr~\zeta\br{\mathbf{M}}}-\frac{1}{\dim{s}^n}\Tr~\zeta\br{M}}\leq O\br{\dim{p}\dim{a}\br{u_d^2\sqrt{\theta} d\dim{a}}^{2/3}},\]
where $u_d$ is defined in \Cref{def:ud}.
\end{lemma}

The following is a hypercontractive inequality for random operators in $L^2\br{\H_m^{\otimes h},\gamma_n}$.

\begin{lemma}\label{lem:Xhypercontractivity}\cite[Lemma 8.12]{qin2021nonlocal}
		\addsetup integers $h,n\geq 0$, and a quantum system $\S$,
	it holds that
		\[\expec{}{\nnorm{\mathbf{M}}_4^4}^{1/4}\leq 3^{d/2}m^{d/4}\expec{}{\nnorm{\mathbf{M}}_2^2}^{1/2},\]
		for all multilinear random operator $\mathbf{M}\in L^2\br{\H_m^{\otimes h},\gamma_n}$, where $d=\max_{a\in[m^2]_{\geq 0}^h}\br{\deg\br{p_{a}}+\abs{a}}$.
	\end{lemma}
%

The following lemma generalizes the hypercontractivity above to the operators in space $L^2\br{\hspa{h},\gamma_n}$.
\begin{lemma}\label{lem:hypercontractivity}
\addsetup let $n,d,h\in\posint$ and a random operator $\mathbf{M}\in L^2\br{\hspa{h},\gamma_n}$. Suppose
\[\mathbf{M}_{p,a}=\sum_{\srange{h}}m_{s,p,a}\br{\mathbf{g}}\S_s\]
for any $p,a,$ where $\mathbf{g}\sim\gamma_n$, $m_{s,p,a}$ is a multilinear polynomial, and $\deg\br{m_{s,p,a}}+\abs{s}\leq d$ for all $s,p,a$. Then
\[\expec{}{\nnorm{\mathbf{M}}_4^4}^{1/4}\leq 3^{d/2}\dim{p}^{7/4}\dim{a}^{7/4}\dim{s}^{d/4}\expec{}{\nnorm{\mathbf{M}}_2^2}^{1/2}.\]
Namely, $u_d\leq 3^{d/2}\dim{p}^{3/4}\dim{a}^{3/4}\dim{s}^{d/4}$, where $u_d$ is defined in \Cref{def:ud}.
\end{lemma}
\begin{proof}

\begin{align*}
\expec{}{\nnorm{\mathbf{M}}_4^4}^{1/4}~\leq~&\sum_{\prange\atop\arange}\expec{}{\nnorm{\mathbf{M}_{p,a}\otimes\ptildep\otimes\atildea}_4^4}^{1/4}\\
=~&\sum_{\prange\atop\arange}\expec{}{\nnorm{\mathbf{M}_{p,a}}_4^4}^{1/4}\nnorm{\ptildep}_4\nnorm{\atildea}_4\\
\overset{(\star)}{\leq}~&\sum_{p,a}3^{d/2}\dim{s}^{d/4}\cdot \expec{}{\nnorm{\mathbf{M}_{p,a}}_2^2}^{1/2}\dim{p}^{1/4}\dim{a}^{1/4}\cdot\nnorm{\ptildep}_2\nnorm{\atildea}_2\\
\overset{(\star\star)}{\leq}~&3^{d/2}\dim{p}^{5/4}\dim{a}^{5/4}\dim{s}^{d/4}\br{\sum_{\prange\atop\arange}\expec{}{\nnorm{\mathbf{M}_{p,a}}_2^2}}^{1/2}\\
=~&3^{d/2}\dim{p}^{7/4}\dim{a}^{7/4}\dim{s}^{d/4}\expec{}{\nnorm{\mathbf{M}}_2^2}^{1/2},
\end{align*}
where ($\star$) is by \cref{lem:Xhypercontractivity} and \cref{eqn:nnormequivalence}, and ($\star$) is by the well-known fact that $\sum_{i=1}^nx_i\leq\br{n\sum_{i=1}^nx_i^2}^{1/2}$ for all $x_1,\dots,x_n\geq0$.
\end{proof}

We are now ready to prove \cref{lem:mainIP}.

\begin{proof}[Proof of \cref{lem:mainIP}]

 Suppose for all $a,b,p,q$, $M_{p,a}$ and $N_{q,b}$ have Fourier expansions
\[M_{p,a}=\sum_{\srange{n}}\widehat{M_{p,a}}\br{s}\S_{s}\]
and
\[N_{q,b}=\sum_{\trange{n}}\widehat{N_{q,b}}\br{t}\T_{t}.\]
By the convention in \Cref{setup}, $\set{\S_{s}}_{\srange{n}}$ and $\set{\T_{t}}_{\trange{n}}$ are standard orthonormal bases in $\H_\sS^{\otimes n}$ and $\H_\sT^{\otimes n}$, respectively, which satisfy Eq.~\eqref{eqn:propersob}.

Without loss of generality, we assume that $\dim{s}\geq\dim{t}$. We further assume for now that the Gaussian random variables with different correlations are allowed. This assumption will be removed later. Define $n\dim{s}^2$ independent joint random variables $\set{\br{\mathbf{g}'_{i,j},\mathbf{h}'_{i,j}}}_{i\in[n], j\in\Br{\dim{s}^2}_{\geq0}}$, where
$\mathbf{g}'_{i,0}=\mathbf{h}'_{i,0}=1$ for all $i\in[n]$, $\br{\mathbf{g}'_{i,j},\mathbf{h}'_{i,j}}\sim\G_{c_j}$ for all $i\in[n], j\in\Br{\dim{s}^2-1}$, and $c_j$ is defined in \cref{eqn:propersob}.

Define

\begin{equation}\label{eqn:mpa}
\mathbf{M}_{p,a}=\sum_{\srange{n}}\widehat{M_{p,a}}\br{s}\prod_{i\notin H}\mathbf{g}'_{i,s_i}\br{\bigotimes_{i\in H}\S_{s_i}}
\end{equation}
and
\begin{equation}\label{eqn:nqb}
\mathbf{N}_{q,b}=\sum_{\trange{n}}\widehat{N_{q,b}}\br{t}\prod_{i\notin H}\mathbf{h}'_{i,t_i}\br{\bigotimes_{i\in H}\T_{t_i}},
\end{equation}
where $\set{\S_{s}}_{s\in\Br{\dim{s}^2}_{\geq 0}}$ and $\set{\T_{t}}_{t\in\Br{\dim{t}^2}_{\geq 0}}$ are standard orthonormal bases in $\H_\sS$ and $\H_\sT$, respectively, which satisfy Eq.~\eqref{eqn:propersob}, by the convention in \cref{setup}.  Then we have the correlation matching:
\begin{equation}\label{eqn:momentmatching}
  \Tr~\br{\S_j\otimes\T_{j'}}\psi^{\sS\sT}=\expec{}{\mathbf{g}'_{i,j}\mathbf{h}'_{i,j'}}=\delta_{j,j'}c_j
\end{equation}
for $i\notin[H]$ and $j,j'\in\srange{}$.

Each item in \cref{lem:mainIP} is proved as follows.

\begin{enumerate}
\item It follows trivially by Eqs.~\eqref{eqn:mpa}\&\eqref{eqn:nqb}.

\item By direct calculation, we have
\[\expec{}{\nnorm{\mathbf{M}_{p,a}}_2^2}=\sum_{\srange{n}}\abs{\widehat{M_{p,a}}\br{s}}^2=\nnorm{M_{p,a}}_2^2\]and \[\expec{}{\nnorm{\mathbf{M}_{p,a}}_2^2}=\sum_{\trange{n}}\abs{\widehat{N_{q,b}}\br{t}}^2=\nnorm{N_{q,b}}_2^2.\]

\item 
By the convention of $\set{\P_p}_{\prange},\set{\Q_q}_{\qrange},\set{\R_r}_{\rrange}, \set{\S_s}_{s\in\Br{\dim{s}^2}_{\geq0}},\set{\T_t}_{t\in\Br{\dim{t}^2}_{\geq0}}$ in \cref{setup} and the correlation matching Eq.~\eqref{eqn:momentmatching}, we have
\begin{align*}
&\expec{}{\Tr\Br{\br{\mathbf{M}_{a}\otimes \mathbf{N}_{b}\otimes\rtilder}\br{\abrshared\otimes\psi^{\otimes h}}}}\\
=~&\sum_{s,p}\widehat{M_{p,a}}\br{s}\widehat{N_{p,b}}\br{s}c_sk_p\\
=~&\Tr\Br{\br{M_a\otimes N_b\otimes\rtilder}\br{\abrshared\otimes\psi^{\otimes n}}}.
\end{align*}

\item It holds by \cref{lem:invarianceH} and \cref{lem:hypercontractivity}.

\item It holds trivially by the definition of $\mathbf{M}$ and $\mathbf{N}$.

\end{enumerate}

It remains to show that, given $c\in[0,\rho]$, $\G_c$ can be simulated by sampling from $\G_\rho$. Indeed, let $\br{\mathbf{g}^{(1)},\mathbf{h}^{(1)}}$ and $\br{\mathbf{g}^{(2)},\mathbf{h}^{(2)}}$ be drawn from $\G_\rho$ independently. Define
\begin{align*}
\mathbf{g}&=\mathbf{g}^{(1)},\\
\mathbf{h}&=\frac{c}{\rho}\mathbf{h}^{(1)}+\sqrt{1-\frac{c^2}{\rho^2}}\mathbf{h}^{(2)}.
\end{align*}
Then $\br{\mathbf{g},\mathbf{h}}\sim\G_c$. It is easy to see that all the items still hold with the replacement.

\end{proof}

The following lemma converts random operators to operators.
\begin{lemma}\label{lem:invarianceback}
Given $\theta\in\br{0,1}$, $n,d\in\posint$, $H\subseteq[n]$ of size $h$,
there exist maps \linebreak$f:L^2\br{\hspa{h},\gamma_{n}}\rightarrow \hspa{n+h},g:L^2\br{\hspa{h},\gamma_{n}}\rightarrow \htqb{n+h}$ such that the following holds:

Let $\br{\mathbf{M},\mathbf{N}}\in L^2\br{\hspa{h},\gamma_n}\times L^2\br{\htqb{h},\gamma_n}$ be any joint random operators with the expansions as  Eqs.~\eqref{eqn:defMa}\&\eqref{eqn:defmpa} satisfying the following.
\begin{enumerate}
\item For all $a,b$, $\expec{}{\nnorm{\mathbf{M}_{p,a}}_2^2}\leq1$ and $\expec{}{\nnorm{\mathbf{N}_{q,b}}_2^2}\leq1$.
\item For all $a,b,p,q$, $\mathbf{M}_{p,a}$ and $\mathbf{N}_{q,b}$ are degree-$d$ multilinear random operators.
\item For all $a,b,p,q$, suppose that
\[\mathbf{M}_{p,a}=\sum_{\srange{h}}m_{s,p,a}\br{\mathbf{g}}\S_s,\quad\mathbf{N}_{q,b}=\sum_{\trange{h}}n_{t,q,b}\br{\mathbf{h}}\T_t,\]
where $\br{\mathbf{g},\mathbf{h}}\sim\G_\rho^{\otimes n}$, and $\set{\S_{s}}_{s\in\Br{\dim{s}^2}_{\geq 0}}$ and $\set{\T_{t}}_{t\in\Br{\dim{t}^2}_{\geq 0}}$ are standard orthonormal bases in $\H_\sS$ and $\H_\sT$, respectively, which satisfy Eq.~\eqref{eqn:propersob}, it holds that
\[\frac{1}{\dim{p}\dim{a}}\sum_{s,p,a}\influence_i\br{m_{s,p,a}}\leq\theta,\quad\frac{1}{\dim{q}\dim{b}}\sum_{t,q,b}\influence_i\br{n_{t,q,b}}\leq\theta.\]
\item $\mathbf{M}_{0}=\id_{\sS^n\sP}/\sqrt{\dim{a}},\mathbf{N}_{0}=\id_{\sT^n\sQ}/\sqrt{\dim{b}}$.
\end{enumerate}
Then
 \[\br{M,N}=\br{f\br{\mathbf{M}},g\br{\mathbf{N}}}\in\hspa{n+h}\times\htqb{n+h}\]
satisfies the following.
 \begin{enumerate}
\item For all $a,b,p,q$:\[\nnorm{M_{p,a}}_2^2=\expec{}{\nnorm{\mathbf{M}_{p,a}}_2^2}\quad\mbox{and}\quad \nnorm{N_{q,b}}_2^2=\expec{}{\nnorm{\mathbf{N}_{q,b}}_2^2}.\]
\item For all $a,b,r$:
\[\Tr\Br{\br{M_a\otimes N_b\otimes\rtilder}\br{\abrshared\otimes\psi^{\otimes n+h}}}=\expec{}{\Tr\Br{\br{\mathbf{M}_{a}\otimes \mathbf{N}_{b}\otimes\rtilder}\br{\abrshared\otimes\psi^{\otimes h}}}}.\]
\item
\[\abs{\frac{1}{\dim{s}^{n+h}}\Tr~\zeta\br{M}-\frac{1}{\dim{s}^h}\expec{}{\Tr~\zeta\br{\mathbf{M}}}}\leq O\br{\dim{p}^{10/3}\dim{a}^{4}\br{3^d\dim{s}^{d/2}\sqrt{\theta}d}^{2/3}}\]
and
\[\abs{\frac{1}{\dim{t}^{n+h}}\Tr~\zeta\br{N}-\frac{1}{\dim{t}^h}\expec{}{\Tr~\zeta\br{\mathbf{N}}}}\leq O\br{\dim{q}^{10/3}\dim{b}^{4}\br{3^d\dim{t}^{d/2}\sqrt{\theta}d}^{2/3}}.\]
\item $M_{0}=\id_{\sS^{n+h}\sP}/\sqrt{\dim{a}}$ and $N_{0}=\id_{\sT^{n+h}\sQ}/\sqrt{\dim{b}}$.
\end{enumerate}
\end{lemma}
\begin{proof}
For all $a,b,p,q$, since $\mathbf{M}_{p,a}$ and $\mathbf{N}_{q,b}$ are multilinear random operators, we can assume that $$m_{s,p,a}\br{\mathbf{g}}=\sum_{\mu\in\{0,1\}^n}\widehat{m_{s,p,a}}(\mu)\prod_{j=1}^n\mathbf{g}_j^{\mu_j};$$
$$n_{t,q,b}\br{\mathbf{h}}=\sum_{\mu\in\{0,1\}^n}\widehat{n_{t,q,b}}(\mu)\prod_{j=1}^n\mathbf{h}_j^{\mu_j},$$
where $\widehat{m_{s,p,a}}(\mu)$, $\widehat{n_{t,q,b}}(\mu)\in\reals$ for all $s,t$. Then $\mathbf{M}_{p,a}$ and $\mathbf{N}_{q,b}$ can be expressed as
$$\mathbf{M}_{p,a}=\sum_{\srange{h}}\sum_{\mu\in\{0,1\}^n}\widehat{m_{s,p,a}}(\mu)\prod_{j=1}^n\mathbf{g}_j^{\mu_j}\S_s;$$
$$\mathbf{N}_{q,b}=\sum_{\trange{h}}\sum_{\mu\in\{0,1\}^n}\widehat{n_{t,q,b}}(\mu)\prod_{j=1}^n\mathbf{h}_j^{\mu_j}\T_t.$$

Define
$$M_{p,a}=\sum_{\srange{h}}\sum_{\mu\in\{0,1\}^n}\widehat{m_{s,p,a}}(\mu)\br{\bigotimes_{j=1}^n\S_{\mu_j}}\otimes\S_s;$$
$$N_{q,b}=\sum_{\trange{h}}\sum_{\mu\in\{0,1\}^n}\widehat{n_{t,q,b}}(\mu)\br{\bigotimes_{j=1}^n\T_{\mu_j}}\otimes\T_t.$$

Each item of \cref{lem:invarianceback} is proved as follows.
\begin{enumerate}
\item By direct calculation, we have
\[\expec{}{\nnorm{\mathbf{M}_{p,a}}_2^2}=\nnorm{M_{p,a}}_2^2=\sum_{\srange{h}}\sum_{\mu\in\{0,1\}^n}\widehat{m_{s,p,a}}(\mu)^2\]

and

\[\expec{}{\nnorm{\mathbf{N}_{q,b}}_2^2}=\nnorm{N_{q,b}}_2^2=\sum_{\trange{h}}\sum_{\mu\in\{0,1\}^n}\widehat{n_{t,q,b}}(\mu)^2.\]

\item
By the same argument as item 3 in the proof of \cref{lem:mainIP}, we have
\begin{align*}
&\Tr\Br{\br{M_a\otimes N_b\otimes\rtilder}\br{\abrshared\otimes\psi^{\otimes n}}}\\
=~&\sum_{s,p}\sum_{\mu\in\{0,1\}^n}\widehat{m_{s,p,a}}(\mu)\widehat{n_{s,p,b}}(\mu)\rho^{\abs{\mu}}c_sk_p\\
=~&\expec{}{\Tr\Br{\br{\mathbf{M}_{a}\otimes \mathbf{N}_{b}\otimes\rtilder}\br{\abrshared\otimes\psi^{\otimes h}}}}.
\end{align*}

\item Observe that for all $i\notin H$,
\begin{align*}
\influence_i\br{M}\overset{(\star)}{=}~&\frac{1}{\dim{p}\dim{a}}\sum_{a,p}\influence_i\br{M_{p,a}}\\
=~&\frac{1}{\dim{p}\dim{a}}\sum_{s,p,a}\sum_{\mu:\mu_i=1}\widehat{m_{s,p,a}}(\mu)^2\\
=~&\frac{1}{\dim{p}\dim{a}}\sum_{s,p,a}\influence_i\br{m_{s,p,a}}\\
\leq~&\theta,
\end{align*}

where $(\star)$ is by \cref{eqn:definf} combined with direct calculation.

Similarly, $\influence_i\br{N}\leq\theta$.
Then from \cref{lem:invarianceH} and \cref{lem:hypercontractivity}, item 3 holds.

\item It follows immediately from the constructions of $f$ and $g$.
\end{enumerate}


%
\end{proof}

\subsection{Dimension reduction}

The following is the main result in this subsection.

\begin{lemma}\label{lem:dimensionreduction}
\addsetup let $\delta,\alpha>0,d,n,h\in\posint$,
\[
\br{\mathbf{M},\mathbf{N}}\in L^2\br{\hspa{h},\gamma_n}\times L^2\br{\htqb{h},\gamma_n},
\]
be degree-$d$ multilinear joint random operators
satisfying the following.
\begin{enumerate}
\item For all $a,b,p,q$,
\begin{equation}\label{eqn:dimredmpa}
\mathbf{M}_{p,a}=\sum_{S\subseteq[n]}\mathbf{g}_SM_{S,p,a},\quad\mathbf{N}_{q,b}=\sum_{S\subseteq[n]}\mathbf{h}_SN_{S,q,b},
\end{equation}
where $\br{\mathbf{g},\mathbf{h}}\sim\G_\rho^{\otimes n}$, $\mathbf{g}_S=\prod_{i\in S}\mathbf{g}_i$,$\mathbf{h}_S=\prod_{i\in S}\mathbf{h}_i$, $M_{S,p,a}\in\H_{\sS^h}$, $N_{S,q,b}\in\H_{\sT^h}$ for all $S\subseteq[n]$;
\item For all $a,b,p,q$, $\expec{}{\nnorm{\mathbf{M}_{p,a}}_2^2}\leq1$ and $\expec{}{\nnorm{\mathbf{N}_{q,b}}_2^2}\leq1$;
\item $\mathbf{M}_0=\id_{\sS^{h}\sP}/\sqrt{\dim{a}}$ and $\mathbf{N}_0=\id_{\sT^{h}\sQ}/\sqrt{\dim{b}}$.
\end{enumerate}


Then there exists an explicitly computable $n_0=n_0\br{d,\delta,\dim{p},\dim{q}}$ and maps \linebreak$f_G: L^2\br{\hspa{h},\gamma_n}\rightarrow L^2\br{\hspa{h},\gamma_{n_0}}$, $g_G:L^2\br{\htqb{h},\gamma_n}\rightarrow L^2\br{\htqb{h},\gamma_{n_0}}$ for $G\in\reals^{n\times n_0}$, such that the following holds:

The joint random operators $\br{\mathbf{M}_G,\mathbf{N}_G}=\br{f_G\br{\mathbf{M}},g_G\br{\mathbf{N}}}$ with  expansions as Eqs.~\eqref{eqn:defMa}~and~\eqref{eqn:defmpa} satisfy

\[
\mathbf{M}^G_{p,a}=\sum_{S\subseteq[n]}\prod_{i\in S}\frac{G_i^T\mathbf{x}}{\twonorm{\mathbf{x}}}\cdot M_{S,p,a}
\]
and
\[
\mathbf{N}^G_{q,b}=\sum_{S\subseteq[n]}\prod_{i\in S}\frac{G_i^T\mathbf{y}}{\twonorm{\mathbf{y}}}\cdot N_{S,q,b},
\]
where $G_i$ is the $i$'th row of $G$. If we sample $\mathbf{G}\sim\gamma_{n\times n_0}$, then with probability at least $1-2\br{\dim{a}\dim{b}\dim{p}\dim{q}\dim{r}}^2\delta-2\alpha$, the following holds:
\begin{enumerate}

\item For all $a,b,p,q$: \[\expec{\mathbf{x}}{\nnorm{\mathbf{M}^{\mathbf{G}}_{p,a}}_2^2}\leq\br{1+\delta}\expec{}{\nnorm{\mathbf{M}_{p,a}}_2^2}\quad\mbox{and}\quad \expec{\mathbf{y}}{\nnorm{\mathbf{N}^{\mathbf{G}}_{q,b}}_2^2}\leq\br{1+\delta}\expec{}{\nnorm{\mathbf{N}_{q,b}}_2^2}.\]
\item \[\expec{\mathbf{x}}{\Tr~\zeta\br{\mathbf{M}^\mathbf{G}}}\leq\frac{1}{\alpha}\expec{\mathbf{g}}{\Tr~\zeta\br{\mathbf{M}}}~\mbox{and}~\expec{\mathbf{y}}{\Tr~\zeta\br{\mathbf{N}^\mathbf{G}}}\leq\frac{1}{\alpha}\expec{\mathbf{h}}{\Tr~\zeta\br{\mathbf{N}}}.\]
\item For all $a,b,r$:
%
%
\begin{multline*}
\left|\expec{\mathbf{x},\mathbf{y}}{\Tr\Br{\br{\mathbf{M}^{\mathbf{G}}_{a}\otimes \mathbf{N}^{\mathbf{G}}_{b}\otimes\rtilder}\br{\abrshared\otimes\psi^{\otimes h}}}}\right.\\\left.-\expec{\mathbf{g},\mathbf{h}}{\Tr\Br{\br{\mathbf{M}_{a}\otimes \mathbf{N}_{b}\otimes\rtilder}\br{\abrshared\otimes\psi^{\otimes h}}}}\right|\leq\delta.
\end{multline*}

\item $\mathbf{M}^{\mathbf{G}}_{0}=\id_{\sS^h\sP}/\sqrt{\dim{a}}$ and $\mathbf{N}^{\mathbf{G}}_{0}=\id_{\sT^h\sQ}/\sqrt{\dim{b}}$.

\end{enumerate}

In particular, one may take $n_0=O\br{\frac{\dim{p}^8\dim{q}^8d^{O(d)}}{\delta^6}}$.
\end{lemma}

\begin{remark}
\cite[Lemma 11.1]{qin2021nonlocal} uses a rough union bound to give an upper bound of $n_0$ that is exponential to $h$. This paper takes a refined analysis following~\cite{Ghazi:2018:DRP:3235586.3235614} and obtains an upper bound of $n_0$ independent of $h$. This leads to an exponential upper bound instead of a doubly-exponential one in the main result.
\end{remark}

To keep the proof succinct, we define for $a,b,r$,
\begin{align*}
F_{a,b,r}(G)&=\expec{\mathbf{x},\mathbf{y}}{\Tr\Br{\br{\mathbf{M}^G_{a}\otimes \mathbf{N}^G_{b}\otimes\rtilder}\br{\abrshared\otimes\psi^{\otimes h}}}}\\
G_{a,b,r}&=\expec{\mathbf{g},\mathbf{h}}{\Tr\Br{\br{\mathbf{M}_{a}\otimes \mathbf{N}_{b}\otimes\rtilder}\br{\abrshared\otimes\psi^{\otimes h}}}}\\
\mathbf{u}_S&= \prod_{i\in S}\frac{G_i^T\mathbf{x}}{\twonorm{\mathbf{x}}},\\
\mathbf{v}_S&= \prod_{i\in S}\frac{G_i^T\mathbf{y}}{\twonorm{\mathbf{y}}}.
\end{align*}
Remind that the randomness of $\mathbf{M}^G_{a}, \mathbf{M}_{a}$ and $\mathbf{N}^G_{b},\mathbf{N}_{b}$ is from the random variables $\mathbf{x}$ and $\mathbf{y}$, respectively.

To prove \cref{lem:dimensionreduction} item 3, we need the following lemma.
\begin{lemma}\label{lem:meanvar}
\addsetup let $\delta>0,d,n,h\in\posint$. There exists $n_0(d,\delta,\dim{p},\dim{q})$ such that the following holds for all $a,b,r$: For $\mathbf{G}\sim\gamma_{n\times n_0}$,
\begin{align*}
&\mbox{(Mean bound)}\quad\quad \abs{\expec{\mathbf{G}}{F_{a,b,r}(\mathbf{G})}-G_{a,b,r}}\leq\delta,\\
&\mbox{(Variance bound)}\quad\quad \var{\mathbf{G}}{F_{a,b,r}(\mathbf{G})}\leq\delta .
\end{align*}
In particular, one may take $n_0=O\br{\frac{\dim{p}^8\dim{q}^8d^{O(d)}}{\delta^2}}$.
\end{lemma}

We borrow the following technical lemma from~\cite{Ghazi:2018:DRP:3235586.3235614}.
\begin{lemma}\cite[Lemma A.8, Lemma A.9]{Ghazi:2018:DRP:3235586.3235614}\label{lem:intermediate}
Given parameters $\rho\in[0,1]$, $d\in\posint$ and $\delta>0$, there exists an explicitly computable $n_0(d,\delta)$ such that the following hold:
\begin{enumerate}
\item For all subsets $S,T\subseteq[n]$ satisfying $\abs{S},\abs{T}\leq d$, it holds that
\begin{align*}
&\text{if }S\ne T:\quad\expec{\mathbf{G},\mathbf{x},\mathbf{y}}{\mathbf{u}_S\mathbf{v}_T}=0,\\
&\text{if }S= T:\quad\abs{\expec{\mathbf{G},\mathbf{x},\mathbf{y}}{\mathbf{u}_S\mathbf{v}_T}-\rho^{\abs{S}}}\leq\delta.\\
\end{align*}

\item For all subsets $S,T,S',T'\subseteq[n]$ satisfying $\abs{S},\abs{T},\abs{S'},\abs{T'}\leq d$, it holds that
\begin{multline*}
\text{if } S\triangle T\triangle S'\triangle T'\ne\emptyset:\\\abs{\expec{\mathbf{G},\mathbf{x},\mathbf{y},\mathbf{x'},\mathbf{y'}}{\mathbf{u}_S\mathbf{v}_T\mathbf{u'}_{S'}\mathbf{v'}_{T'}}-\br{\expec{\mathbf{G},\mathbf{x},\mathbf{y}}{\mathbf{u}_S\mathbf{v}_T}}\br{\expec{\mathbf{G},\mathbf{x'},\mathbf{y'}}{\mathbf{u'}_{S'}\mathbf{v'}_{T'}}}}=0,
\end{multline*}

\begin{multline*}
\text{if } S\triangle T\triangle S'\triangle T'=\emptyset:\\\abs{\expec{\mathbf{G},\mathbf{x},\mathbf{y},\mathbf{x'},\mathbf{y'}}{\mathbf{u}_S\mathbf{v}_T\mathbf{u'}_{S'}\mathbf{v'}_{T'}}-\br{\expec{\mathbf{G},\mathbf{x},\mathbf{y}}{\mathbf{u}_S\mathbf{v}_T}}\br{\expec{\mathbf{G},\mathbf{x'},\mathbf{y'}}{\mathbf{u'}_{S'}\mathbf{v'}_{T'}}}}\leq\delta.
\end{multline*}

Here, $S\triangle T\triangle S'\triangle T'$ is the symmetric difference of the sets $S,T,S',T'$.
\end{enumerate}

In particular, one may take $n_0=\frac{d^{O(d)}}{\delta^2}.$
\end{lemma}

\begin{proof}[Proof of \cref{lem:meanvar}]

By \cref{lem:pqrdiag}, for any $r$, we can choose bases $\set{\P_p}_{\prange}$, $\set{\Q_q}_{\qrange}$ satisfying \cref{eqn:pqdiag}. Applying \cref{lem:intermediate} with parameters $d$ and $\delta\leftarrow\delta/\br{\dim{p}\dim{q}}^{1/2}$, we have
\begin{align*}
~&\abs{\expec{\mathbf{G}}{F_{a,b,r}(\mathbf{G})}-G_{a,b,r}}\\
=~&\left|\expec{\mathbf{G},\mathbf{x},\mathbf{y}}{\Tr\Br{\br{\mathbf{M}^G_{a}\otimes \mathbf{N}^G_{b}\otimes\rtilder}\br{\abrshared\otimes\psi^{\otimes h}}}}\right.\\
&\left.-\expec{\mathbf{g},\mathbf{h}}{\Tr\Br{\br{\mathbf{M}_{a}\otimes \mathbf{N}_{b}\otimes\rtilder}\br{\abrshared\otimes\psi^{\otimes h}}}}\right|\\
\overset{(\star)}{=}~&\abs{\sum_{S,T\subseteq[n],p}\br{\expec{\mathbf{G},\mathbf{x},\mathbf{y}}{\mathbf{u}_S\mathbf{v}_T}-\expec{\mathbf{g},\mathbf{h}}{\mathbf{g}_S\mathbf{h}_T}}\Tr\Br{\br{M_{S,p,a}\otimes N_{T,p,b}}\psi^{\otimes h}}\cdot k_p}\\
\overset{(\star\star)}{\leq}~&\frac{\delta}{\br{\dim{p}\dim{q}}^{1/2}}\sum_{S\subseteq[n],p}\abs{\Tr\Br{\br{M_{S,p,a}\otimes N_{S,p,b}}\psi^{\otimes h}}\cdot k_p}\\
\overset{(\star\star\star)}{\leq}~&\frac{\delta}{\br{\dim{p}\dim{q}}^{1/2}}\sum_{S\subseteq[n],p}\nnorm{M_{S,p,a}}_2\nnorm{N_{S,p,b}}_2\\
\leq~&\frac{\delta}{\br{\dim{p}\dim{q}}^{1/2}}\br{\sum_{S\subseteq[n],p}\nnorm{M_{S,p,a}}_2^2\cdot\sum_{S\subseteq[n],p}\nnorm{N_{S,p,b}}_2^2}^{1/2}\\
\overset{(4\star)}{\leq}~&\delta \br{\expec{}{\nnorm{\mathbf{M}_{a}}_2^2}\expec{}{\nnorm{\mathbf{N}_{b}}_2^2}}^{1/2}\\
\leq~&\delta,
\end{align*}
where $(\star)$ is by \cref{eqn:pqdiag}, $(\star\star)$ is by \cref{lem:intermediate} item 1, $(\star\star\star)$ is by
\cref{fac:cauchyschwartz} and ~\cref{lem:pqrdiag}, and $(4\star)$ is by \cref{eqn:defmpa} and \cref{eqn:dimredmpa}.

Using \cref{lem:intermediate} with parameters $d$ and $\delta\leftarrow\delta/\br{9^d\dim{p}^4}$, we have

\begin{align*}
&\var{\mathbf{G}}{F_{a,b,r}(\mathbf{G})}\\
=~&\expec{\mathbf{G}}{F_{a,b,r}(\mathbf{G})^2}-\br{\expec{\mathbf{G}}{F_{a,b,r}(\mathbf{G})}}^2\\
\leq~&\sum_{S,T,S',T'\subseteq[n],p,p'}\abs{\expec{\mathbf{G},\mathbf{x},\mathbf{y},\mathbf{x'},\mathbf{y'}}{\mathbf{u}_S\mathbf{v}_T\mathbf{u'}_{S'}\mathbf{v'}_{T'}}-\br{\expec{\mathbf{G},\mathbf{x},\mathbf{y}}{\mathbf{u}_S\mathbf{v}_T}}\br{\expec{\mathbf{G},\mathbf{x'},\mathbf{y'}}{\mathbf{u'}_{S'}\mathbf{v'}_{T'}}}}\\
~&\abs{\Tr\Br{\br{M_{S,p,a}\otimes N_{T,p,b}}\psi^{\otimes h}}\Tr\Br{\br{M_{S',p',a}\otimes N_{T',p',b}}\psi^{\otimes h}}k_pk_{p'}}\\
\overset{(\star)}{\leq}~&\frac{\delta}{9^d\dim{p}^4}\sum_{S,T,S',T'\subseteq[n]\atop S\triangle T\triangle S'\triangle T'=\emptyset,p,p'}\nnorm{M_{S,p,a}}_2\nnorm{N_{T,p,b}}_2\nnorm{M_{S',p',a}}_2\nnorm{N_{T',p',b}}_2,
\end{align*}
where $(\star)$ is by~\cref{lem:intermediate} item 2,~\cref{fac:cauchyschwartz} and~\cref{lem:pqrdiag}.
\[
\sum_{S,T,S',T'\subseteq[n]\atop S\triangle T\triangle S'\triangle T'=\emptyset,p,p'}\nnorm{M_{S,p,a}}_2\nnorm{N_{T,p,b}}_2\nnorm{M_{S',p',a}}_2\nnorm{N_{T',p',b}}_2\leq9^d\dim{p}^4\nnorm{M_{a}}_2^2\nnorm{N_{b}}_2^2.
\]
For all $a,b,p$, define functions $f_{p,a},g_{p,b}:\set{-1,1}^n\rightarrow\reals$ over the boolean hypercube as,
\[f_{p,a}(x)=\sum_{S\subseteq[n]\atop\abs{S}\leq d}\nnorm{M_{S,p,a}}_2x_S\quad\text{and}\quad g_{p,b}(x)=\sum_{T\subseteq[n]\atop\abs{T}\leq d}\nnorm{N_{T,p,b}}_2x_T.\]
By the hypercontractive inequality \cite[Page 243, Bonami Lemma]{Odonnell08}
\[\expec{x}{f_{p,a}(x)^4}\leq9^d\br{\expec{x}{f_{p,a}(x)^2}}^2\quad\text{and}\quad\expec{x}{g_{p,b}(x)^4}\leq9^d\br{\expec{x}{g_{p,b}(x)^2}}^2,\]
we finish the proof as follows.
\begin{align*}
&\sum_{S,T,S',T'\subseteq[n]\atop S\triangle T\triangle S'\triangle T'=\emptyset,p,p'}\nnorm{M_{S,p,a}}_2\nnorm{N_{T,p,b}}_2\nnorm{M_{S',p',a}}_2\nnorm{N_{T',p',b}}_2\\
=&\sum_{S,T,S',T'\subseteq[n],p,p'}\nnorm{M_{S,p,a}}_2\nnorm{N_{T,p,b}}_2\nnorm{M_{S',p',a}}_2\nnorm{N_{T',p',b}}_2\expec{x}{x_Sx_{S'}x_Tx_{T'}}\\
=~&\expec{x}{\br{\sum_{p}f_{p,a}(x)g_{p,b}(x)}^2}\\
\leq~&\dim{p}^2\sum_{p}\expec{x}{\br{f_{p,a}(x)g_{p,b}(x)}^2}\\
\leq~&\dim{p}^2\sum_{p}\br{\expec{x}{f_{p,a}(x)^4}\expec{x}{g_{p,b}(x)^4}}^{1/2}\\
\leq~&9^d\dim{p}^2\sum_{p}\expec{x}{f_{p,a}(x)^2}\expec{x}{g_{p,b}(x)^2}\\
=~&9^d\dim{p}^2\sum_p\br{\br{\sum_{S\subseteq[n]}\nnorm{M_{S,p,a}}_2^2}\br{\sum_{T\subseteq[n]}\nnorm{N_{T,p,b}}_2^2}}\\
\leq~&9^d\dim{p}^2\sum_{S\subseteq[n],p}\nnorm{M_{S,p,a}}_2^2\sum_{T\subseteq[n],p}\nnorm{N_{T,p,b}}_2^2\\
\overset{(\star)}{\leq}~&9^d\dim{p}^4\nnorm{M_{a}}_2^2\nnorm{N_{b}}_2^2,\\
\end{align*}
where $(\star)$ is by \cref{eqn:defmpa} and \cref{eqn:dimredmpa}.
\end{proof}

The following fact is for item 1 in \cref{lem:dimensionreduction}.
\begin{fact}~\cite[Theorem 3.1]{Ghazi:2018:DRP:3235586.3235614}\label{fac:dimensionreduction}
	Given parameters $n,d\in\mathbb{N}_{+}$, $\rho\in[0,1]$ and $\delta>0$, there exists an explicitly computable $D=D\br{d,\delta}$ such that the following holds:
	
	For all $n$ and any degree-$d$ multilinear polynomials $\alpha,\beta:\reals^n\rightarrow\reals$, and $G\in\reals^{n\times D}$, define functions $\alpha_G, \beta_G:\reals^{D}\rightarrow\reals$ as
	\begin{equation*}
	\alpha_G\br{x}= \alpha\br{\frac{Gx}{\twonorm{x}}}~\mbox{and}~\beta_G\br{x}= \beta\br{\frac{Gx}{\twonorm{x}}}.
	\end{equation*}
	Then
	\[\Pr_{\mathbf{G}\sim \gamma_{n\times D}}\Br{\abs{\innerproduct{\alpha_{\mathbf{G}}}{\beta_{\mathbf{G}}}_{\G_{\rho}^{\otimes D}}-\innerproduct{\alpha}{\beta}_{\G_{\rho}^{\otimes n}}}<\delta\twonorm{\alpha}\twonorm{\beta}}\geq 1-\delta.\]

If $\alpha$ and $\beta$ are identical and $\rho=1$, we have

\begin{eqnarray*}
&&\Pr_{\mathbf{G}\sim \gamma_{n\times D}}\Br{\abs{\twonorm{\alpha_{\mathbf{G}}}^2-\twonorm{\alpha}^2}\leq\delta\twonorm{\alpha}^2}\geq1-\delta; \\
&&\Pr_{\mathbf{G}\sim \gamma_{n\times D}}\Br{\abs{\twonorm{\beta_{\mathbf{G}}}^2-\twonorm{\beta}^2}\leq\delta\twonorm{\beta}^2}\geq1-\delta.
\end{eqnarray*}

	In particular, one may take $D=\frac{d^{O\br{d}}}{\delta^6}$.
\end{fact}

\begin{proof}[Proof of \cref{lem:dimensionreduction}]

Rewrite
\[\mathbf{M}_{p,a}=\sum_{\srange{h}}m_{s,p,a}\br{\mathbf{g}}\S_s\]
and
\[\mathbf{M}^G_{p,a}=\sum_{\srange{h}}m^G_{s,p,a}\br{\mathbf{x}}\S_s.\]
where $m_{s,p,a}:\reals^n\rightarrow\reals$ is a degree-$d$ multilinear polynomial and $m^G_{s,p,a}:\reals^{n_0}\rightarrow\reals$ for all $s,p,a$.
By  the definition of $\mathbf{M}^G_{p,a}$, we have
$m^G_{s,p,a}\br{\mathbf{x}}=m_{s,p,a}\br{G\mathbf{x}/\twonorm{x}}$.
Each item in Lemma~\ref{lem:smoothing} is proved as follows.
\begin{enumerate}
\item By \cref{fac:dimensionreduction}, with probability at least $1-\delta$, we get
\[\expec{\mathbf{x}}{\nnorm{\mathbf{M}^{\mathbf{G}}_{p,a}}_2^2}=\sum_{\srange{h}}\twonorm{m^G_{s,p,a}}^2\leq\br{1+\delta}\sum_{\srange{h}}\twonorm{m_{s,p,a}}^2=\br{1+\delta}\expec{}{\nnorm{\mathbf{M}_{p,a}}_2^2}.\]

Similarly, $\expec{\mathbf{y}}{\nnorm{\mathbf{N}^{\mathbf{G}}_{q,b}}_2^2}\leq\br{1+\delta}\expec{}{\nnorm{\mathbf{N}_{q,b}}_2^2}$, so item 1 holds.

\item We first observe that for all fixed $\mathbf{x}\in\reals^{n_0}$, the distribution of $G\mathbf{x}/\twonorm{\mathbf{x}}$ is identical to that of a standard $n$-variate Gaussian distribution. Thus, we immediately have that,

\[\expec{\mathbf{G},\mathbf{x}}{\Tr~\zeta\br{\mathbf{M}_\mathbf{G}}}=\expec{\mathbf{g}}{\Tr~\zeta\br{\mathbf{M}}}.\]

Thus, using Markov's inequality, we get that with probability at least $1-\alpha,$
\[\expec{\mathbf{x}}{\Tr~\zeta\br{\mathbf{M}_\mathbf{G}}}\leq\frac{1}{\alpha}\expec{\mathbf{g}}{\Tr~\zeta\br{\mathbf{M}}}.\]

\item We invoke \cref{lem:meanvar} with $\delta\leftarrow\delta^3/2$ and
$n_0=O\br{\dim{p}^8\dim{q}^8d^{O(d)}/\delta^6}$
. Using Chebyshev's inequality and the Variance bound in  \cref{lem:meanvar}, we have that for all $\eta>0$,
\[\Pr_{\mathbf{G}}\Br{\abs{F_{a,b,r}(\mathbf{G})-\expec{\mathbf{G}}{F_{a,b,r}(\mathbf{G})}}>\eta}\leq\frac{\delta^3}{2\eta^2}.\]
Using the triangle inequality, and the Mean bound in \cref{lem:meanvar}, we get
\begin{align*}
~&\Pr_{\mathbf{G}}\Br{\abs{F_{a,b,r}(\mathbf{G})-G_{a,b,r}}>\delta}\\
\leq~&\Pr_{\mathbf{G}}\Br{\abs{F_{a,b,r}(\mathbf{G})-\expec{\mathbf{G}}{F_{a,b,r}(\mathbf{G})}}+\abs{\expec{\mathbf{G}}{F_{a,b,r}(\mathbf{G})}-G_{a,b,r}}>\delta}\\
\leq~&\Pr_{\mathbf{G}}\Br{\abs{F_{a,b,r}(\mathbf{G})-\expec{\mathbf{G}}{F_{a,b,r}(\mathbf{G})}}>\delta-\delta^2}\\
\leq~&\delta.
\end{align*}

\item It follows directly by the construction of $\mathbf{M}^\mathbf{G}$ and $\mathbf{N}^\mathbf{G}$.
\end{enumerate}
Items 1 to 4 hold simultaneously with probability $1-2\br{\dim{a}\dim{b}\dim{p}\dim{q}\dim{r}}^2\delta-2\alpha$ by a union bound.
\end{proof}

\subsection{Smoothing random operators}

\begin{lemma}\label{lem:smoothGaussian}
\addsetup let $\delta>0,n,h\in\posint$,

\[
\br{\mathbf{M},\mathbf{N}}\in L^2\br{\hspa{h},\gamma_n}\times L^2\br{\htqb{h},\gamma_n},
\]
be joint random operators with the expansions Eqs.~\eqref{eqn:defMa}\&\eqref{eqn:defmpa}, where
\begin{equation}\label{eqn:rmpa}
\mathbf{M}_{p,a}=\sum_{\srange{h}}m_{s,p,a}\br{\mathbf{g}}\S_s\quad\mbox{and}\quad\mathbf{N}_{q,b}=\sum_{\trange{h}}n_{t,q,b}\br{\mathbf{h}}\T_t.
\end{equation}
Suppose they satisfy that
\begin{enumerate}
\item $\br{\mathbf{g},\mathbf{h}}\sim\G_\rho^{\otimes n}$;
\item For all $a,b$, $\expec{}{\nnorm{\mathbf{M}_{a}}_2^2}\leq2$, $\expec{}{\nnorm{\mathbf{N}_{b}}_2^2}\leq2$;
\item $\mathbf{M}_{0}=\id_{\sS^h\sP}/\sqrt{\dim{a}}$ and $\mathbf{N}_{0}=\id_{\sT^h\sQ}/\sqrt{\dim{b}}$.
\end{enumerate}
Then there exists an explicitly computable $d=d\br{\rho,\delta,\dim{a},\dim{b},\dim{p},\dim{q}}$ and maps \linebreak$f: L^2\br{\hspa{h},\gamma_n}\rightarrow L^2\br{\hspa{h},\gamma_{n}},g:L^2\br{\htqb{h},\gamma_n}\rightarrow L^2\br{\htqb{h},\gamma_{n}}$ such that
\[
\br{\mathbf{M}^{(1)},\mathbf{N}^{(1)}}=\br{f(\mathbf{M}),g(\mathbf{N})}
\]
satisfies the following.
\begin{enumerate}
\item For all $a,b,p,q$:
\[\deg\br{\mathbf{M}^{(1)}_{p,a}}\leq d\quad\mbox{and}\quad \deg\br{\mathbf{N}^{(1)}_{q,b}}\leq d.\]

\item For all $a,b,p,q$: \[\expec{}{\nnorm{\mathbf{M}^{(1)}_{p,a}}_2^2}^{1/2}\leq \expec{}{\nnorm{\mathbf{M}_{p,a}}_2^2}^{1/2}\quad\mbox{and}\quad \expec{}{\nnorm{\mathbf{N}^{(1)}_{q,b}}_2^2}^{1/2}\leq \expec{}{\nnorm{\mathbf{N}_{q,b}}_2^2}^{1/2}.\]
\item \[\expec{}{\Tr~\zeta\br{\mathbf{M}^{(1)}}}\leq\expec{}{\Tr~\zeta\br{\mathbf{M}}}+\delta~\mbox{and}~\expec{}{\Tr~\zeta\br{\mathbf{N}^{(1)}}}\leq\expec{}{\Tr~\zeta\br{\mathbf{N}}}+\delta.\]

\item For all $a,b,r$:
\begin{multline*}
\left|\expec{}{\Tr\Br{\br{\mathbf{M}^{(1)}_{a}\otimes \mathbf{N}^{(1)}_{b}\otimes\rtilder}\br{\abrshared\otimes\psi^{\otimes h}}}}\right.\\\left.-\expec{}{\Tr\Br{\br{\mathbf{M}_{a}\otimes \mathbf{N}_{b}\otimes\rtilder}\br{\abrshared\otimes\psi^{\otimes h}}}}\right|\leq\delta.
\end{multline*}
\item $\mathbf{M}^{(1)}_{0}=\id_{\sS^h\sP}/\sqrt{\dim{a}}$ and $\mathbf{N}^{(1)}_{0}=\id_{\sT^h\sQ}/\sqrt{\dim{b}}$.
\end{enumerate}

In particular, one may take $d=O\br{\frac{\dim{a}^2\dim{b}^2\dim{p}\dim{q}}{\delta\br{1-\rho}}}$.
\end{lemma}

To prove \cref{lem:smoothGaussian}, we need the following fact.
\begin{fact}~\cite[Lemma 4.1]{Ghazi:2018:DRP:3235586.3235614}\label{fac:smoothGaussian}
	Let $\rho\in[0,1),\delta>0,k,n\in\posint$ be any given constant parameters, $f,g\in L^2\br{\reals^k,\gamma_n}$; $\Lambda_1,\Lambda_2\subseteq\reals^k$ be closed convex sets. Set $\R_1$ and $\R_2$ be rounding maps of $\Lambda_1$ and $\Lambda_2$, respectively. Then there exists explicitly computable $\nu=\nu\br{\rho,\delta},d=d(\rho,\delta)$ such that the following holds:
	\begin{enumerate}
		\item For all $i\in[k]$, it holds that \[\twonorm{\br{U_\nu(f_i)}^{\leq d}}\leq\twonorm{U_\nu(f_i)}\leq\twonorm{f_i}~\mbox{and}~\twonorm{\br{U_\nu(g_i)}^{\leq d}}\leq\twonorm{U_\nu(g_i)}\leq\twonorm{g_i}.\]
		\item
\[\twonorm{\R_1\circ U_\nu(f_i)-U_\nu(f_i)}\leq\twonorm{\R_1\circ f-f}\]
		and
\[\twonorm{\R_2\circ U_\nu(g_i)-U_\nu(g_i)}\leq\twonorm{\R_2\circ g-g}.\]
\item For all $i\in [k]$,
\[\twonorm{\br{U_\nu(f_i)}^{>d}}^2\leq\delta\twonorm{U_\nu(f_i)}^2\quad\mbox{and}\quad\twonorm{\br{U_\nu(g_i)}^{>d}}^2\leq\delta\twonorm{U_\nu(g_i)}^2\]
		\item For all $i\in [k]$,
		\[\abs{\innerproduct{f_i\br{\mathbf{x}}}{g_i\br{\mathbf{y}}}_{\G_{\rho}^{\otimes n}}-\innerproduct{\br{U_\nu(f_i)}^{\leq d}\br{\mathbf{x}}}{\br{U_\nu(g_i)}^{\leq d}\br{\mathbf{y}}}_{\G_{\rho}^{\otimes n}}}\leq\delta\twonorm{f_i}\twonorm{g_i};\]
	\end{enumerate}
	In particular, one may take $$\nu=(1-\delta)^{\log\rho/\br{\log\delta+\log\rho}},d=\frac{\log(1/\delta)}{2\log(1/\nu)}=O\br{\frac{\log^2\frac{1}{\delta}}{\delta\br{1-\rho}}}.$$
\end{fact}

We are now ready to prove \cref{lem:smoothGaussian}.
\begin{proof}[Proof of \cref{lem:smoothGaussian}]

Set
\[\delta'=\frac{\delta}{16\dim{a}^2\dim{b}^2\dim{p}\dim{q}},\quad\nu=(1-\delta')^{\frac{\log\rho}{\log\delta'+\log\rho}},\quad d=\frac{\log(1/\delta')}{2\log(1/\nu)}=O\br{\frac{\log^2\frac{1}{\delta'}}{\delta'\br{1-\rho}}}.\]
 For any $a,b,p,q,s,t$, let $m'_{s,p,a}=U_\nu(m_{s,p,a}),n'_{t,q,b}=U_\nu(n_{t,q,b}),m^{(1)}_{s,p,a}=\br{m'_{s,p,a}}^{\leq d},n^{(1)}_{t,q,b}=\br{n'_{t,q,b}}^{\leq d}$, where $U_\nu$ is the Ornstein-Uhlenbeck operator defined in \cref{def:ornstein} and \cref{eqn:vectorornstein}. For all $a,b,p,q$, define
\begin{align}
\mathbf{M}'_{p,a}=\sum_{\srange{h}}m'_{s,p,a}\br{\mathbf{g}}\S_s\quad&\mbox{and}\quad\mathbf{N}'_{q,b}=\sum_{\trange{h}}n'_{t,q,b}\br{\mathbf{h}}\T_t;\nonumber\\
\mathbf{M}^{(1)}_{p,a}=\sum_{\srange{h}}m^{(1)}_{s,p,a}\br{\mathbf{g}}\S_s\quad&\mbox{and}\quad\mathbf{N}^{(1)}_{q,b}=\sum_{\trange{h}}n^{(1)}_{t,q,b}\br{\mathbf{h}}\T_t.\label{eqn:mpa1}
\end{align}
Set vector-valued functions $m=\br{m_{s,p,a}}_{\srange{h},\prange,\arange},m'=\br{m'_{s,p,a}}_{\srange{h},\prange,\arange}$\linebreak and $m^{(1)}=\br{m^{(1)}_{s,p,a}}_{\srange{h},\prange,\arange}$.
Each item in \cref{lem:smoothGaussian} is proved as follows:
\begin{enumerate}
\item It follows directly by the definition of $\mathbf{M}^{(1)}$ and $\mathbf{N}^{(1)}$.

\item We apply \cref{lem:randoperator} and \cref{fac:smoothGaussian} item 1,
\[
\expec{}{\nnorm{\mathbf{M}^{(1)}_{p,a}}_2^2}=\twonorm{m^{(1)}_{p,a}}^2\leq\twonorm{m_{p,a}}^2=\expec{}{\nnorm{\mathbf{M}_{p,a}}_2^2}.
\]
where $m_{p,a}$ is the associated vector-valued function (in \cref{def:randomoperators}) of $\mathbf{M}_{p,a}$. Similarly, $
\expec{}{\nnorm{\mathbf{N}^{(1)}_{q,b}}_2^2}\leq \expec{}{\nnorm{\mathbf{N}_{q,b}}_2^2}$.

\item Define a convex set
\[\Lambda=\set{x\in\reals^{\dim{s}^h\dim{p}\dim{a}}:\sum_{s,p,a}x_{s,p,a}\S_s\otimes\ptildep\otimes\atildea\geq0},\]
and let $\R$ be a rounding map of $\Lambda$. Note that $0\in\Lambda$, thus by \cref{fac:rounding}, for all $x\in\reals^{\dim{s}^h\dim{p}\dim{a}}$, we have
\begin{equation}\label{eqn:Rcontraction}
\twonorm{\R(x)}\leq\twonorm{x}.
\end{equation}

By \cref{lem:randoperator}, \cref{lem:closedelta1} and \cref{fac:smoothGaussian} item 2, we have
\begin{equation}\label{eqn:funcnoise}
\frac{1}{\dim{s}^h}\expec{}{\Tr~\zeta\br{\mathbf{M}'}}=\twonorm{\R\circ m'-m'}^2\leq\twonorm{\R\circ m-m}^2=\frac{1}{\dim{s}^h}\expec{}{\Tr~\zeta\br{\mathbf{M}}}.
\end{equation}

By \cref{lem:randoperator} and \cref{lem:closedelta1}, we have
\begin{align*}
&\frac{1}{\dim{s}^h}\abs{\expec{}{\Tr~\zeta\br{\mathbf{M}^{(1)}}}-\expec{}{\Tr~\zeta\br{\mathbf{M}'}}}\\
=~&\abs{\twonorm{\R\circ m^{(1)}-m^{(1)}}^2-\twonorm{\R\circ m'-m'}^2}\\
=~&\abs{\twonorm{\R\circ m^{(1)}-m^{(1)}}-\twonorm{\R\circ m'-m'}}\br{\twonorm{\R\circ m^{(1)}-m^{(1)}}+\twonorm{\R\circ m'-m'}}\\
\leq~&\br{\twonorm{\R\circ m^{(1)}-\R\circ m'}\hspace*{-0.3em}+\hspace*{-0.1em}\twonorm{m^{(1)}-m'}}\br{\twonorm{\R\circ m^{(1)}}\hspace*{-0.2em}+\twonorm{m^{(1)}}\hspace*{-0.2em}+\hspace*{-0.1em}\twonorm{\R\circ m'}\hspace*{-0.3em}+\hspace*{-0.1em}\twonorm{m'}}\\
\leq~&8\twonorm{m^{(1)}-m'}\twonorm{m'}\quad\mbox{(\cref{fac:rounding} and \cref{eqn:Rcontraction})}\\
\leq~&8\delta'\twonorm{m'}^2\quad\mbox{(\cref{fac:smoothGaussian} item 3)}\\
\leq~&8\delta'\twonorm{m}^2\quad\mbox{(\cref{fac:smoothGaussian} item 1)}\\
=~&8\delta' \sum_{\srange{h}}\sum_{\prange,\arange}\twonorm{m_{s,p,a}}^2\\
=~&8\delta' \sum_{\prange,\arange}\expec{}{\nnorm{\mathbf{M}_{p,a}}^2}\quad\mbox{(\cref{lem:randoperator})}\\
=~&8\delta' \dim{p}\sum_{\arange}\expec{}{\nnorm{\mathbf{M}_a}_2^2}\quad\mbox{(\cref{eqn:defmpa})}\\
\leq~&16\delta'\dim{p}\dim{a}^2~\leq~\delta.
\end{align*}

Combined with \cref{eqn:funcnoise}, we conclude the first inequality in item 3. The second inequality follows by the same argument.

\item 
    By Eqs.~\eqref{eqn:propersob}\eqref{eqn:pqdiag} and the definitions of $\mathbf{M}_a,\mathbf{N}_b,\mathbf{M}^{(1)}_{a}, \mathbf{N}^{(1)}_{b}$, we have
\begin{align*}
&\left|\expec{}{\Tr\Br{\br{\mathbf{M}^{(1)}_{a}\otimes \mathbf{N}^{(1)}_{b}\otimes\rtilder-\mathbf{M}_{a}\otimes \mathbf{N}_{b}\otimes\rtilder}\br{\abrshared\otimes\shared{h}}}}\right|\\
=~&\abs{\sum_{s,p}\br{\innerproduct{m^{(1)}_{s,p,a}}{n^{(1)}_{s,p,b}}_{\G_\rho^{\otimes n}}-\innerproduct{m_{s,p,a}}{n_{s,p,b}}_{\G_\rho^{\otimes n}}}c_sk_p}\\
\leq~&\delta'\sum_{s,p}\twonorm{m_{s,p,a}}\twonorm{n_{s,p,b}}\quad\mbox{(\cref{fac:smoothGaussian} item 4)}\\
\leq~&\delta'\br{\sum_{s,p}\twonorm{m_{s,p,a}}^2}^{1/2}\br{\sum_{s,p}\twonorm{n_{s,p,b}}^2}^{1/2}\\
\leq~&\br{\dim{p}\dim{q}}^{1/2}\delta'\expec{}{\nnorm{\mathbf{M}_a}_2^2}^{1/2}\expec{}{\nnorm{\mathbf{N}_b}_2^2}^{1/2}\quad\mbox{(\cref{lem:randoperator} and \cref{eqn:defmpa})}\\
\leq~&\delta.
\end{align*}

\item It follows directly by the definition of $\mathbf{M}^{(1)}$ and $\mathbf{N}^{(1)}$.
\end{enumerate}

\end{proof}

\subsection{Multilinearization}
\begin{lemma}\label{lem:multiliniearization}
\addsetup let $\delta>0,d,n,h\in\posint$,
\[
\br{\mathbf{M},\mathbf{N}}\in L^2\br{\hspa{h},\gamma_n}\times L^2\br{\htqb{h},\gamma_n},
\]
be joint random operators with
\[\mathbf{M}_{p,a}=\sum_{\srange{h}}m_{s,p,a}\br{\mathbf{g}}\S_s\quad\mbox{and}\quad\mathbf{N}_{q,b}=\sum_{\trange{h}}n_{t,q,b}\br{\mathbf{h}}\T_t,\]
as defined in Eqs.~\eqref{eqn:defmpa} satisfying that
\begin{enumerate}
\item $\br{\mathbf{g},\mathbf{h}}\sim\G_\rho^{\otimes n}$.
\item For all $a,b$, $\expec{}{\nnorm{\mathbf{M}_{a}}_2^2}\leq2$ and $\expec{}{\nnorm{\mathbf{N}_{b}}_2^2}\leq2$.
\item For all $a,b$, $\deg\br{\mathbf{M}'_{a}}\leq d,\deg\br{\mathbf{N}'_{b}}\leq d$.
\item $\mathbf{M}_0=\id_{\sS^{h}\sP}/\sqrt{\dim{a}}$ and $\mathbf{N}_0=\id_{\sT^{h}\sQ}/\sqrt{\dim{b}}$.
\end{enumerate}
 Then there exists an explicitly computable $t=t\br{\rho,\delta,\dim{a},\dim{b},\dim{p},\dim{q}}$ and maps \linebreak$f: L^2\br{\hspa{h},\gamma_n}\rightarrow L^2\br{\hspa{h},\gamma_{n\cdot n_1}},g:L^2\br{\htqb{h},\gamma_{n}}\rightarrow L^2\br{\htqb{h},\gamma_{n\cdot n_1}}$ such that
\[
\br{\mathbf{M}',\mathbf{N}'}=\br{f(\mathbf{M}),g(\mathbf{N})}
\]
satisfies the following:
		\begin{enumerate}
			\item For all $a,b,p,q$, $\mathbf{M}'_{p,a}$ and $\mathbf{N}'_{q,b}$ are degree-$d$ multilinear random operators.
			\item Suppose that
\[\mathbf{M}'_{p,a}=\sum_{\srange{h}}m'_{s,p,a}\br{\mathbf{x}}\S_s\quad\mbox{and}\quad\mathbf{N}'_{q,b}=\sum_{\trange{h}}n'_{t,q,b}\br{\mathbf{y}}\T_t,\]

where $\br{\mathbf{x},\mathbf{y}}\sim\G_\rho^{\otimes n\cdot n_1}$. For all $\br{i,j}\in[n]\times[n_1],a,b,p,q,s,t$,
			\[\influence_{(i-1)n_1+j}\br{m'_{s,p,a}}\leq\delta\cdot\influence_i\br{m_{s,p,a}}\quad\mbox{and}\quad\influence_{(i-1)n_1+j}\br{n'_{t,q,b}}\leq\delta\cdot\influence_i\br{n_{t,q,b}}.\]
			\item For all $a,b$: \[\expec{}{\nnorm{\mathbf{M}'_{a}}_2^2}^{1/2}\leq \expec{}{\nnorm{\mathbf{M}_a}_2^2}^{1/2}\quad\mbox{and}\quad \expec{}{\nnorm{\mathbf{N}'_{b}}_2^2}^{1/2}\leq \expec{}{\nnorm{\mathbf{N}_b}_2^2}^{1/2}.\]
			\item \[\frac{1}{\dim{s}^h}\abs{\expec{}{\Tr~\zeta\br{\mathbf{M}'}}-\expec{}{\Tr~\zeta\br{\mathbf{M}}}}\leq\delta \]
			and
			\[\frac{1}{\dim{t}^h}\abs{\expec{}{\Tr~\zeta\br{\mathbf{N}'}}-\expec{}{\Tr~\zeta\br{\mathbf{N}}}}\leq \delta.\]
\item For all $a,b,r$:

\begin{multline*}
\left|\expec{}{\Tr\Br{\br{\mathbf{M}'_{a}\otimes \mathbf{N}'_{b}\otimes\rtilder}\br{\abrshared\otimes\psi^{\otimes h}}}}\right.\\\left.-\expec{}{\Tr\Br{\br{\mathbf{M}_{a}\otimes \mathbf{N}_{b}\otimes\rtilder}\br{\abrshared\otimes\psi^{\otimes h}}}}\right|\leq\delta.
\end{multline*}

\item $\mathbf{M}'_0=\id_{\sS^{h}\sP}/\sqrt{\dim{a}}$ and $\mathbf{N}'_0=\id_{\sT^{h}\sQ}/\sqrt{\dim{b}}$.
		\end{enumerate}
		
In particular, one may take $n_1=O\br{\dim{a}^4\dim{b}^4\dim{p}^2\dim{q}^2d^2/\delta^2}$.
\end{lemma}

We need the following notion introduced in~\cite{Ghazi:2018:DRP:3235586.3235614} before constructing $\mathbf{M}'$ and $\mathbf{N}'$.
\begin{definition}~\cite{Ghazi:2018:DRP:3235586.3235614}\label{def:linear}
	Suppose $f\in L^2\br{\reals,\gamma_n}$ is given with a Hermite expansion \[f=\sum_{\sigma\in\mathbb{Z}_{\geq 0}^n}\widehat{f}\br{\sigma}H_{\sigma}.\] The {\em multilinear truncation} of $f$ is defined to be the function $f^{\mathpzc{ml}}\in L^2\br{\reals,\gamma_n}$ given by
	\[f^{\mathpzc{ml}}=\sum_{\sigma\in\set{0,1}^n}\widehat{f}\br{\sigma}H_{\sigma}.\]
\end{definition}

The following fact is crucial to our proof.
\begin{fact}~\cite[Fact 13.3]{qin2021nonlocal}\cite[Lemma 5.1]{Ghazi:2018:DRP:3235586.3235614}\label{fac:mulilinear}
	Given parameters $\rho\in[0,1], \delta>0$ and $d\in\posint$, there exists $n_1=n_1\br{d,\delta}$ such that the following holds:
	
	Let $f,g\in L^2\br{\reals,\gamma_n}$ be degree-$d$ polynomials. There exist polynomials $\bar{f},\bar{g}\in L^2\br{\reals,\gamma_{n\cdot n_1}}$ over variables \[\bar{x}=\set{x_j^{\br{i}}:\br{i,j}\in[n]\times[n_1]}~\mbox{and}~\bar{y}=\set{y_j^{\br{i}}:\br{i,j}\in[n]\times[n_1]}\]
satisfying the following:
	\begin{enumerate}
		\item $\bar{f}^{\mathpzc{ml}}$ and $\bar{g}^{\mathpzc{ml}}$ are multilinear with degree $d$.
		\item $\twonorm{\bar{f}^{\mathpzc{ml}}}\leq\twonorm{\bar{f}}=\twonorm{f}$ and $\twonorm{\bar{g}^{\mathpzc{ml}}}\leq\twonorm{\bar{g}}=\twonorm{g}$.
		\item Given two independent random variables $\mathbf{g}\sim \gamma_n$ and $\mathbf{x}\sim \gamma_{n\cdot n_1}$, the distributions of $f\br{\mathbf{g}}$ and $\bar{f}\br{\mathbf{x}}$ are identical and the distributions of $g\br{\mathbf{g}}$ and $\bar{g}\br{\mathbf{x}}$ are identical.
		
		\item $\twonorm{\bar{f}-\bar{f}^{\mathpzc{ml}}}\leq\frac{\delta}{2}\twonorm{f}$ and $\twonorm{\bar{g}-\bar{g}^{\mathpzc{ml}}}\leq\frac{\delta}{2}\twonorm{g}$.
		\item For all $\br{i,j}\in[n]\times[n_1]$, it holds that \[\influence_{x^{\br{i}}_j}\br{\bar{f}^{\mathpzc{ml}}}\leq\delta\cdot\influence_i\br{f}~\mbox{and}~\influence_{y^{\br{i}}_j}\br{\bar{g}^{\mathpzc{ml}}}\leq\delta\cdot \influence_i\br{g}.\]
		\item $\abs{\innerproduct{\bar{f}^{\mathpzc{ml}}}{\bar{g}^{\mathpzc{ml}}}_{\G^{\otimes n\cdot n_1}_{\rho}}-\innerproduct{f}{g}_{\G^{\otimes n}_{\rho}}}\leq\delta\twonorm{f}\twonorm{g}$.
	\end{enumerate}
	In particular, we may take $n_1=O\br{d^2/\delta^2}.$
\end{fact}

We are now ready to prove \cref{lem:multiliniearization}.

\begin{proof}[Proof of \cref{lem:multiliniearization}]
	Take $\delta'=\delta/\dim{a}^2\dim{b}^2\dim{p}\dim{q}$. For all $a,b,p,q$, applying \cref{fac:mulilinear} to \linebreak$\set{m_{s,p,a}}_{\srange{h}}$ and $\set{n_{t,q,b}}_{\trange{h}}$ with $\delta\leftarrow\delta'$, we get $\set{\overline{m_{s,p,a}}}_{\srange{h}}$ and $\set{\overline{n_{t,q,b}}}_{\trange{h}}$.
	Let $m'_{s,p,a}\br{\cdot}=\overline{m_{s,p,a}}^{\mathpzc{ml}}\br{\cdot}$ and
	$n'_{t,q,b}\br{\cdot}=\overline{n_{t,q,b}}^{\mathpzc{ml}}\br{\cdot}$.
Each item of \cref{lem:multiliniearization} is proved as follows.
\begin{enumerate}
\item It is implied by \cref{fac:mulilinear} item 1.
\item It is implied by \cref{fac:mulilinear} item 5.
\item It follows from \cref{lem:randoperator} and the item 3 in \cref{fac:mulilinear}.
\item We prove the first inequality in item 4. The second inequality follows from the same argument. Define a convex set
\[\Lambda=\set{x\in\reals^{\dim{s}^h\dim{p}\dim{a}}:\sum_{s,p,a}x_{s,p,a}\S_s\otimes\ptildep\otimes\atildea\geq0},\]
and let $\R$ be a rounding map of $\Lambda$. Note that $0\in\Lambda$, thus  by \cref{fac:rounding}, for all $x\in\reals^{\dim{s}^h\dim{a}}$, we have
\begin{equation}
\twonorm{\R(x)}\leq\twonorm{x}.
\end{equation}

Let $m=\br{m_{s,p,a}}_{\srange{h},\prange,\arange}$, similar for $m'$, $\bar{m}$ and $\bar{m}^{\mathpzc{ml}}$. Then by \cref{lem:randoperator} and \cref{lem:closedelta1},
	\begin{align*}
		&\frac{1}{\dim{s}^h}\abs{\expec{}{\Tr~\zeta\br{\mathbf{M}'}}-\expec{}{\Tr~\zeta\br{\mathbf{M}}}}\\ \\
		=~&\abs{\twonorm{m'-\R\circ m'}^2-\twonorm{m-\R\circ m}^2}\\
		=~&\abs{\twonorm{\bar{m}^{\mathpzc{ml}}-\R\circ \bar{m}^{\mathpzc{ml}}}^2-\twonorm{\bar{m}-\R\circ \bar{m}}^2}\quad\quad\mbox{(\cref{fac:mulilinear} item 3)}\\
		=~&\abs{\br{\twonorm{\bar{m}^{\mathpzc{ml}}-\R\circ \bar{m}^{\mathpzc{ml}}}-\twonorm{\bar{m}-\R\circ \bar{m}}}\br{\twonorm{\bar{m}^{\mathpzc{ml}}-\R\circ \bar{m}^{\mathpzc{ml}}}+\twonorm{\bar{m}-\R\circ \bar{m}}}}\\
\leq~&\abs{\br{\twonorm{\bar{m}-\bar{m}^{\mathpzc{ml}}}\hspace*{-0.3em}+\hspace*{-0.1em}\twonorm{\R\circ\bar{m}-\R\circ\bar{m}^{\mathpzc{ml}}}}\br{\twonorm{\bar{m}^{\mathpzc{ml}}}\hspace*{-0.3em}+\hspace*{-0.1em}\twonorm{\R\circ \bar{m}^{\mathpzc{ml}}}\hspace*{-0.3em}+\hspace*{-0.1em}\twonorm{\bar{m}}+\twonorm{\R\circ \bar{m}}}}\\
		\leq~&4\twonorm{m}\br{\twonorm{\bar{m}-\bar{m}^{\mathpzc{ml}}}+\twonorm{\R\circ\bar{m}-\R\circ\bar{m}^{\mathpzc{ml}}}}\quad\quad\mbox{(\cref{fac:mulilinear} item 2 and \cref{eqn:Rcontraction})}\\
		\leq~&8\twonorm{m}\twonorm{\bar{m}-\bar{m}^{\mathpzc{ml}}}\quad\quad\mbox{(\cref{fac:rounding})}\\
		\leq~&4\delta'\twonorm{m}^2\quad\quad\mbox{(\cref{fac:mulilinear} item 4)}\\
		=~&4\delta' \sum_{\prange,\arange}\sum_s\twonorm{m_{s,p,a}}^2\\
=~&4\delta' \sum_{\prange,\arange}\expec{}{\nnorm{\mathbf{M}_{p,a}}^2}\quad\mbox{(\cref{lem:randoperator})}\\
=~&4\delta' \dim{p}\sum_{\arange}\expec{}{\nnorm{\mathbf{M}_a}_2^2}\quad\mbox{(\cref{eqn:defmpa})}\\
\leq~&8\delta'\dim{p}\dim{a}^2~\leq~\delta.
	\end{align*}

\item 
By Eqs.~\eqref{eqn:propersob}\eqref{eqn:pqdiag} and the definitions of $\mathbf{M}_a,\mathbf{N}_b,\mathbf{M}'_{a}, \mathbf{N}'_{b}$, we have
\begin{align*}
&\left|\expec{}{\Tr\Br{\br{\mathbf{M}'_{a}\otimes \mathbf{N}'_{b}\otimes\rtilder-\mathbf{M}_{a}\otimes \mathbf{N}_{b}\otimes\rtilder}\br{\abrshared\otimes\shared{h}}}}\right|\\
=~&\abs{\sum_{s,p}\br{\innerproduct{m'_{s,p,a}}{n'_{s,p,b}}_{\G_\rho^{\otimes n}}-\innerproduct{m_{s,p,a}}{n_{s,p,b}}_{\G_\rho^{\otimes n}}}c_sk_p}\\
\leq~&\delta'\sum_{s,p}\twonorm{m_{s,p,a}}\twonorm{n_{s,p,b}}\quad\mbox{(\cref{fac:mulilinear} item 6)}\\
\leq~&\delta'\br{\sum_{s,p}\twonorm{m_{s,p,a}}^2}^{1/2}\br{\sum_{s,p}\twonorm{n_{s,p,b}}^2}^{1/2}\\
\leq~&\br{\dim{p}\dim{q}}^{1/2}\delta'\expec{}{\nnorm{\mathbf{M}_a}^2_2}^{1/2}\expec{}{\nnorm{\mathbf{N}_b}^2_2}^{1/2}\quad\mbox{(\cref{lem:randoperator} and \cref{eqn:defmpa})}\\
\leq~&\delta.
\end{align*}

	\item It holds by the definition of $\mathbf{M}'$ and $\mathbf{N}'$.
\end{enumerate}
\end{proof}

\subsection{Rounding}
For a Hermitian matrix $X$, suppose it has a spectral decomposition $U^\dagger\Lambda U$, where $U$ is unitary and $\Lambda$ is diagonal. Define
\begin{equation*}
\pos{X}= U^{\dagger}\pos{\Lambda} U,
\end{equation*}
where \pos{\Lambda} is diagonal and $\pos{\Lambda}_{ii}=\Lambda_{ii}$ if $\Lambda_{ii}\geq0$, and $\pos{\Lambda}_{ii}=0$, otherwise. Let $\pinv{X}$ be the {\em Moore-Penrose inverse} of $X$. Namely,
\begin{equation*}
  \pinv{X}= U^{\dagger}\pinv{\Lambda} U,
\end{equation*}
	where \pinv{\Lambda} is diagonal and $\pinv{\Lambda}_{ii}=\Lambda_{ii}^{-1}$ if $\Lambda_{ii}\ne0$, and $\pinv{\Lambda}_{ii}=0$ otherwise.

\begin{lemma}\label{lem:mainrounding}
Let $J\in\H_{\sS\sA}$ with $\dim{s}=\abs{\sS}$ and $\dim{a}=\abs{\sA}$ satisfy that $\Tr~\zeta(J)\leq\dim{s}\epsilon$ and $\Tr_\sA J=\id_\sS$. Then there exists $\widetilde{J}\in\H_{\sS\sA}$ such that the following holds:
\begin{enumerate}
\item $\widetilde{J}\geq0$;
\item $\Tr_\sA \widetilde{J}=\id_\sS$;
\item $\nnorm{J-\widetilde{J}}_2^2\leq O\br{\dim{a}^{5/2}\sqrt{\epsilon}}$.
\end{enumerate}
\end{lemma}

\begin{proof}
Let $\pos{J}_\sS$ denote $\Tr_\sA \pos{J}$, and $\Pi_\sS$ be the projector onto the support of $\pos{J}_\sS$. Note that $\pos{J}\geq0$ implies $\pos{J}_\sS\geq0$. Define
\[\widetilde{J}=\br{\sqrt{\br{\pos{J}_\sS}^+}\otimes\id_\sA}\pos{J}\br{\sqrt{\br{\pos{J}_\sS}^+}\otimes\id_\sA}+\br{\id_\sS-\Pi_\sS}\otimes\frac{\id_\sA}{\dim{a}}.\]

It is easy to verify that $\widetilde{J}$ satisfies item 1 and 2. To prove item 3,
\begin{align*}
&\nnorm{J-\widetilde{J}}_2^2\\
\leq~&4\left(\nnorm{\br{\sqrt{\br{\pos{J}_\sS}^+}\otimes\id_\sA}\pos{J}\br{\sqrt{\br{\pos{J}_\sS}^+}\otimes\id_\sA}-\br{\sqrt{\br{\pos{J}_\sS}^+}\otimes\id_\sA}\pos{J}}_2^2\right.\\
&+\left.\nnorm{\br{\sqrt{\br{\pos{J}_\sS}^+}\otimes\id_\sA}\pos{J}-\pos{J}}_2^2+\nnorm{\pos{J}-J}_2^2+\nnorm{\br{\id_\sS-\Pi_\sS}\otimes\frac{\id_\sA}{\dim{a}}}_2^2\right)\\
\leq~&O\br{\dim{a}^2\nnorm{\pos{J}_\sS-\id_\sS}_2+\epsilon/\dim{a}}\quad\quad\mbox{(\cref{claim:term1} and \cref{claim:term2})}\\
\leq~&O\br{\dim{a}^{5/2}\sqrt{\epsilon}}\quad\quad\mbox{(Using \cref{lem:CommonTerm} with $\pos{J}-J$).}
\end{align*}

\end{proof}

\begin{claim}\label{claim:term1}
\[
\nnorm{\br{\sqrt{\br{\pos{J}_\sS}^+}\otimes\id_\sA}\pos{J}\br{\sqrt{\br{\pos{J}_\sS}^+}\otimes\id_\sA}-\br{\sqrt{\br{\pos{J}_\sS}^+}\otimes\id_\sA}\pos{J}}_2^2\leq\nnorm{\pos{J}_\sS-\id_\sS}_2.\]
\end{claim}

\begin{claim}\label{claim:term2}
\[\nnorm{\br{\sqrt{\br{\pos{J}_\sS}^+}\otimes\id_\sA}\pos{J}-\pos{J}}_2^2\leq\dim{a}^2\nnorm{\pos{J}_\sS-\id_\sS}_2\br{\nnorm{\pos{J}_\sS-\id_\sS}_2+1}.\]
\end{claim}

\begin{lemma}\label{lem:CommonTerm}
Let $J\in\H_\sS\otimes\H_\sA$. It holds that
\[
\nnorm{\Tr_\sA J}_2\leq\dim{a}\nnorm{J}_2.
\]
\end{lemma}

\begin{proof}
\begin{align*}
\br{\Tr_\sA J}_{i,j}=\br{\sum_{k=1}^{\dim\sA}\br{\id_\sS\otimes\bra{k}_\sA}J\br{\id_\sS\otimes\ket{k}_\sA}}_{i,j}
=\sum_{k=1}^{\dim\sA}J_{(i,k),(j,k)}.
\end{align*}
Thus,
\begin{align*}
&\nnorm{\Tr_\sA J}_2^2\\
=~&\frac{1}{\dim{s}}\sum_{i,j}\abs{\Tr_\sA J}_{i,j}^2\\
=~&\frac{1}{\dim{s}}\sum_{i,j}\abs{\sum_{k\in\Br{\dim{a}}}J_{(i,k),(j,k)}}^2\\
\leq~&\frac{\dim{a}}{\dim{s}}\sum_{i,j}\sum_{k\in\Br{\dim{a}}}\abs{J_{(i,k),(j,k)}}^2\\
\leq~&\dim{a}^2\nnorm{J}_2^2.
\end{align*}
\end{proof}

\begin{proof}[Proof of \cref{claim:term1}]
\begin{align*}
&\nnorm{\br{\sqrt{\br{\pos{J}_\sS}^+}\otimes\id_\sA}\pos{J}\br{\sqrt{\br{\pos{J}_\sS}^+}\otimes\id_\sA}-\br{\sqrt{\br{\pos{J}_\sS}^+}\otimes\id_\sA}\pos{J}}_2^2\\
=~&\frac{1}{\dim{s}\dim{a}}\Tr\br{\id_{\sS\sA}-\br{\sqrt{\br{\pos{J}_\sS}^+}\otimes\id_\sA}}^2\pos{J}\br{\br{\pos{J}_\sS}^+\otimes\id_\sA}\pos{J}\\
\leq~&\frac{1}{\dim{s}}\Tr\br{\br{\id_\sS-\sqrt{\br{\pos{J}_\sS}^+}}^2\otimes\id_{\sA}}\pos{J}\quad\mbox{(\cref{lem:pinvinequality} item 1 and \cref{lem:Jupbound})}\\
=~&\frac{1}{\dim{s}}\Tr\br{\id_\sS-\sqrt{\br{\pos{J}_\sS}^+}}^2\pos{J}_\sS\\
\leq~&\frac{1}{\dim{s}}\Tr\abs{\id_\sS-\br{\pos{J}_\sS}^+}\pos{J}_\sS\quad\quad\mbox{(\cref{lem:pinvinequality} item 2)}\\
=~&\frac{1}{\dim{s}}\Tr\abs{\pos{J}_\sS-\Pi_\sS}\\
\leq~&\frac{1}{\dim{s}}\Tr\abs{\pos{J}_\sS-\id_\sS}\\
\leq~&\nnorm{\pos{J}_\sS-\id_\sS}_2.
\end{align*}
\end{proof}
\begin{proof}[Proof of \cref{claim:term2}]
\begin{align*}
&\nnorm{\br{\sqrt{\br{\pos{J}_\sS}^+}\otimes\id_\sA}\pos{J}-\pos{J}}_2^2\\
=~&\frac{1}{\dim{s}\dim{a}}\Tr\br{\pos{J}}^2\br{\br{\Pi_\sS-\sqrt{\br{\pos{J}_\sS}^+}}^2\otimes\id_\sA}\\
\leq~&\frac{1}{\dim{s}\dim{a}}\Tr\br{\pos{J}}^2\br{\abs{\Pi_\sS-\br{\pos{J}_\sS}^+}\otimes\id_\sA}\quad\quad\mbox{(\cref{lem:pinvinequality} item 2)}\\
=~&\frac{1}{\dim{s}\dim{a}}\Tr\br{\pos{J}}^2\br{\sqrt{\br{\pos{J}_\sS}^+}\otimes\id_\sA}\br{\abs{\pos{J}_\sS-\Pi_\sS}\otimes\id_\sA}\br{\sqrt{\br{\pos{J}_\sS}^+}\otimes\id_\sA}\\
=~&\frac{1}{\dim{s}\dim{a}}\Tr\br{\sqrt{\br{\pos{J}_\sS}^+}\otimes\id_\sA}\br{\pos{J}}^2\br{\sqrt{\br{\pos{J}_\sS}^+}\otimes\id_\sA}\br{\abs{\pos{J}_\sS-\Pi_\sS}\otimes\id_\sA}\\
\leq~&\frac{1}{\dim{s}\dim{a}}\norm{\br{\sqrt{\br{\pos{J}_\sS}^+}\otimes\id_\sA}\sqrt{\pos{J}}}^2\norm{\pos{J}}_2\norm{\abs{\pos{J}_\sS-\Pi_\sS}\otimes\id_\sA}_2\quad\mbox{(H\"older's)}\\
=~&\frac{1}{\dim{s}\dim{a}}\norm{\br{\sqrt{\br{\pos{J}_\sS}^+}\otimes\id_\sA}\pos{J}\br{\sqrt{\br{\pos{J}_\sS}^+}\otimes\id_\sA}}\norm{\pos{J}}_2\norm{\abs{\pos{J}_\sS-\Pi_\sS}\otimes\id_\sA}_2\\
\leq~&\frac{1}{\dim{s}}\norm{\pos{J}}_2\norm{\abs{\pos{J}_\sS-\Pi_\sS}\otimes\id_\sA}_2\quad\quad\mbox{(\cref{lem:Jupbound})}\\
\leq~&\frac{\dim{a}^2}{\dim{s}}\norm{\pos{J}_\sS}_2\norm{\pos{J}_\sS-\Pi_\sS}_2\quad\quad\mbox{(\cref{lem:Jupbound})}\\
\leq~&\dim{a}^2\nnorm{\pos{J}_\sS-\id_\sS}_2\br{\nnorm{\pos{J}_\sS-\id_\sS}_2+1}.
\end{align*}
\end{proof}

\begin{fact}\label{lem:Jupbound}\cite[Proposition 3.4]{zhang2019some}
Let $J\in\H_\sS\otimes\H_\sA$ be positive, then
\[\dim{a}\br{\Tr_\sA J}\otimes\id_\sA\geq J.\]
\end{fact}
\begin{fact}\label{lem:pinvinequality}\cite[Lemma 9.5]{qin2021nonlocal}
Given Hermitian matrices $A$ and $B$ the following holds:
\begin{enumerate}
  \item If $A\geq B\geq0$, then $B\pinv{A}B\leq B$.
  \item If $A\geq0$, then $(\id-A)^2\leq\abs{\id-A^2}$.
\end{enumerate}
\end{fact}
\begin{lemma}\label{lem:roundingneed}
\addsetup let $n\in\posint$, $M\in\hspa{n}$ and $N\in\htqb{n}$. Then for all $a,b,r$, we have
\[\abs{\Tr\Br{\br{M_a\otimes N_b\otimes \rtilder}\br{\abrshared\otimes\psi^{\otimes n}}}}\leq\br{\dim{p}\dim{q}}^{1/2}\nnorm{M_a}_2\nnorm{N_b}_2\]
\end{lemma}
\begin{proof}
By \cref{lem:pqrdiag}, for any $r$, we can choose bases $\set{\P_p}_{\prange}$,
$\set{\Q_q}_{\qrange}$ satisfying \cref{eqn:pqdiag}.
\begin{align*}
&\abs{\Tr\Br{\br{M_a\otimes N_b\otimes \rtilder}\br{\abrshared\otimes\psi^{\otimes n}}}}\\
=~&\abs{\sum_{p,q}\Tr\Br{\br{M_{p,a}\otimes N_{q,b}}\psi^{\otimes n}}\cdot\Tr\Br{\br{\ptildep\otimes\qtildeq\otimes \rtilder}\abrshared}}\\
\leq~&\sum_{p}\abs{\Tr\Br{\br{M_{p,a}\otimes N_{p,b}}\psi^{\otimes n}}}\quad\mbox{(\cref{eqn:pqdiag})}\\
\leq~&\sum_{p}\nnorm{M_{p,a}}_2\nnorm{N_{p,b}}_2\quad\mbox{(\cref{fac:cauchyschwartz})}\\
\leq~&\br{\sum_{p}\nnorm{M_{p,a}}_2^2}^{1/2}\br{\sum_{p}\nnorm{N_{p,b}}_2^2}^{1/2}\\
\leq~&\br{\dim{p}\dim{q}}^{1/2}\nnorm{M_a}_2\nnorm{N_b}_2
\end{align*}
\end{proof}

\appendix
\section{List of notations}\label{app:notations}

\begin{longtable}{ll}
$\Delta_{\gamma}\br{P}$&noise operator,

$\gamma P+\frac{1-\rho}{m}\br{\Tr P}\cdot\id_m$\\
$\Phi^*$&the adjoint of $\Phi$\\
$\gamma_n$&standard $n$-dimensional normal distribution\\
$|\sigma|$&the number of nonzeros in $\sigma$\\
$\zeta\br{x}$&$\begin{cases}x^2~&\mbox{if $x\leq 0$}\\ 0~&\mbox{otherwise}\end{cases}$\\
$A\geq B$&the matrix $A-B$ is positive semi-definite \\
$\A_{a}$&$\otimes_{i=1}^n\A_{a_i}$\\
$\deg P$&$\max\set{\abs{\sigma}:\widehat{P}\br{\sigma}\neq 0}$\\
$\deg\br{\mathbf{P}}$&$\max_{\sigma\in[m^2]_{\geq 0}^h}\deg\br{p_{\sigma}}$\\
$\deg\br{f}$&$\max\set{\sum_i\sigma_i:~\widehat{f}\br{\sigma}\neq 0}$\\
$f\in L^p\br{\reals,\gamma_n}$&$f:\reals^n\rightarrow\reals, \int_{\reals^n}\abs{f(x)}^p\gamma_n\br{dx}<\infty$\\
$f\in L^p\br{\reals^k,\gamma_n}$&$f_1,\dots,f_k\in L^p\br{\reals,\gamma_n}$ for $f=\br{f_1,\ldots,f_k}$\\
$\innerproduct{f}{g}_{\gamma_n}$&$\expec{\mathbf{x}\sim\gamma_n}{f\br{\mathbf{x}}g\br{\mathbf{x}}}$\\
$\norm{f}_p$&$\br{\int_{\reals^n}\abs{f(x)}^p\gamma_n\br{dx}}^{\frac{1}{p}}$\\
&$\br{\sum_{t=1}^k\norm{f_t}_p^p}^{1/p}$ for $f=\br{f_1,\ldots,f_k}$\\
$\G_{\rho}$&$\rho$-correlated Gaussian distribution $N\br{\begin{pmatrix}
                                                            0 \\
                                                            0
                                                          \end{pmatrix},\begin{pmatrix}
                                                                          1 & \rho \\
                                                                          \rho & 1
                                                                        \end{pmatrix}}$\\
$\H_\sS$&the set of all Hermitian operators in the system $\sS$\\
$\H_m$&the set of all Hermitian operators of dimension $m$ \\
$\H_m^{\otimes n}$&$\underbrace{\H_m\otimes\cdots\otimes\H_m}_{n\text{ times}}$\\
$H_r\br{x}$&Hermite polynomial, $\frac{(-1)^r}{\sqrt{r!}}e^{x^2/2}\frac{d^r}{dx^r}e^{-x^2/2}$\\
$H_{\sigma}\br{x}$&$\prod_{i=1}^nH_{\sigma_i}\br{x_i}$\\
$\id_\sS$&the identity operator in the system $\sS$\\
$\id_m$&the identity operator of dimension $m$\\
$\influence\br{P}$&$\sum_i\influence_i\br{P}$\\
$\influence\br{f}$&$\sum_i\influence_i\br{f}$\\
$\influence_i\br{P}$&$\nnorm{P-\id_m\otimes\Tr_iP}_2^2$\\
$\influence_i\br{f}$&$\expec{\mathbf{x}\sim \gamma_n}{\var{\mathbf{x}'_i\sim\gamma_1}{f\br{\mathbf{x}_1,\dots,\mathbf{x}_{i-1},\mathbf{x}'_i,\mathbf{x}_{i+1},\dots\mathbf{x}_n}}}$\\
$\choi{\cdot}$&the Choi representation\\
$\L\br{\sS,\sA}$&the set of all linear maps from $\M_\sS$ to $\M_\sA$\\
$\L\br{\sS}$&$\L\br{\sS,\sS}$\\
$\M_\sS$&the set of all linear operators in the system $\sS$\\
$\M_m$&the set of all linear operators of dimension $m$ \\
$\M_m^{\otimes n}$&$\underbrace{\M_m\otimes\cdots\otimes\M_m}_{n\text{ times}}$\\
$M\geq0$&the matrix $M$ is positive semi-definite \\
$M^{\dagger}$&the transposed conjugate of $M$\\
$M_{i,j}$&the $(i,j)$-entry of $M$\\
$\nnorm{M}_p$&$\br{\frac{1}{m}\Tr~\abs{H}^p}^{1/p}$\\
$\norm{M}_p$&$\br{\Tr~\abs{H}^p}^{1/p}$\\
$\widehat{M}\br{\sigma}$&$\innerproduct{\B_{\sigma}}{M}$, Fourier coefficient of $M$ with respect to $\B$\\
$[n]$&\set{1,\dots,n}\\
$[n]_{\geq0}$&\set{0,\dots,n-1}\\
$p=\br{p_{\sigma}}_{\sigma\in[m^2]_{\geq 0}^h}$&the associated vector-valued function of $\mathbf{P}$\\
POVM& $M_1,\ldots, M_t\geq0$ satisfying $\sum_{i=1}^tM_i=\id$\\
$\innerproduct{P}{Q}$&$\frac{1}{m}\Tr~P^{\dagger}Q$\\
$P^{\leq t}$&$\sum_{\sigma\in[m^2]_{\geq 0}^n:\abs{\sigma}\leq t}\widehat{P}\br{\sigma}\B_{\sigma}$(similar for $P^{< t}$, $P^{\geq t}$, $P^{>t}$, $P^{= t}$)\\
$\mathbf{P}\in L^p\br{\H_m^{\otimes h},\gamma_n}$ &$p_{\sigma}\in L^p\br{\reals,\gamma_n}$ for all $\sigma\in[m^2]_{\geq 0}^h$\\
$\R\br{x}$&$\arg\min\set{\norm{x-y}_2^2:y\in\Delta}$\\
$\Tr_\sS\psi^{\sS\sT}$&partial trace, $\sum_i\br{\id_\sS\otimes\bra{i}}\psi^{\sS\sT}\br{\id_\sS\otimes\ket{i}}$\\
$U_{\nu}f\br{z}$&$\expec{\mathbf{x}\sim \gamma_n}{f\br{\nu z+\sqrt{1-\nu^2}\mathbf{x}}}$\\
&$\br{U_{\nu}f_1,\ldots, U_{\nu}f_k}$ for $f=\br{f_1,\ldots,f_k}$\\
$\var{}{f}$&$\expec{\mathbf{x}\sim \gamma_n}{\abs{f\br{\mathbf{x}}-\expec{}{f}}^2}$\\
$\pos{X}$& $U\pos{\Lambda} U^{\dagger}$ where $X=U\Lambda U^{\dagger}$ is a spectral decomposition \\
&of $X$ and $\pos{\Lambda}_{i,i}=\Lambda_{i,i}$ if $\Lambda_{i,i}\geq 0$ and $\pos{\Lambda}_{i,i}=0$ otherwise.\\
$X^{+}$& Moore-Penrose inverse of $X$.\\
$\posint$& the set of all positive integers.
\end{longtable}
\bibliographystyle{plain}
	\bibliography{Fully_quantum_v2}

\begin{thebibliography}{10}

\bibitem{Beigi:2013}
Salman Beigi.
\newblock A new quantum data processing inequality.
\newblock {\em Journal of Mathematical Physics}, 54(8):082202, 2013.

\bibitem{4690981}
A.~{Ben-Aroya}, O.~{Regev}, and R.~d.~{Wolf}.
\newblock A hypercontractive inequality for matrix-valued functions with
  applications to quantum computing and ldcs.
\newblock In {\em 2008 49th Annual IEEE Symposium on Foundations of Computer
  Science}, pages 477--486, Oct 2008.

\bibitem{PhysRevA.53.2046}
Charles~H. Bennett, Herbert~J. Bernstein, Sandu Popescu, and Benjamin
  Schumacher.
\newblock Concentrating partial entanglement by local operations.
\newblock {\em Phys. Rev. A}, 53:2046--2052, Apr 1996.

\bibitem{bertsekas2015convex}
Dimitri~P Bertsekas.
\newblock {\em Convex optimization algorithms}.
\newblock Athena Scientific Belmont, 2015.

\bibitem{PhysRevLett.110.060405}
Cyril Branciard, Denis Rosset, Yeong-Cherng Liang, and Nicolas Gisin.
\newblock Measurement-device-independent entanglement witnesses for all
  entangled quantum states.
\newblock {\em Phys. Rev. Lett.}, 110:060405, Feb 2013.

\bibitem{BuGJKL:2022}
Kaifeng Bu, Roy~J. Garcia, Arthur Jaffe, Dax~Enshan Koh, and Lu~Li.
\newblock Complexity of quantum circuits via sensitivity, magic, and coherence.
\newblock {\em arXiv preprint arXiv:2204.12051}, 2022.

\bibitem{PhysRevLett.108.200401}
Francesco Buscemi.
\newblock All entangled quantum states are nonlocal.
\newblock {\em Phys. Rev. Lett.}, 108:200401, May 2012.

\bibitem{PhysRevA.87.032306}
Eric~G. Cavalcanti, Michael J.~W. Hall, and Howard~M. Wiseman.
\newblock Entanglement verification and steering when alice and bob cannot be
  trusted.
\newblock {\em Phys. Rev. A}, 87:032306, Mar 2013.

\bibitem{ChenNY:2022}
Thomas Chen, Shivam Nadimpalli, and Henry Yuen.
\newblock Testing and learning quantum juntas nearly optimally.
\newblock {\em arXiv preprint arXiv:2207.05898}, 2207.05898.

\bibitem{chung_et_al:LIPIcs:2015:5072}
Kai-Min Chung, Xiaodi Wu, and Henry Yuen.
\newblock {Parallel Repetition for Entangled k-player Games via Fast Quantum
  Search}.
\newblock In David Zuckerman, editor, {\em 30th Conference on Computational
  Complexity (CCC 2015)}, volume~33 of {\em Leibniz International Proceedings
  in Informatics (LIPIcs)}, pages 512--536, Dagstuhl, Germany, 2015. Schloss
  Dagstuhl--Leibniz-Zentrum fuer Informatik.

\bibitem{Cleve:2004:CLN:1009378.1009560}
Richard Cleve, Peter Hoyer, Benjamin Toner, and John Watrous.
\newblock Consequences and limits of nonlocal strategies.
\newblock In {\em Proceedings of the 19th IEEE Annual Conference on
  Computational Complexity}, CCC '04, pages 236--249, Washington, DC, USA,
  2004. IEEE Computer Society.

\bibitem{doi:10.1137/1.9781611975031.174}
Anindya De, Elchanan Mossel, and Joe Neeman.
\newblock Non interactive simulation of correlated distributions is decidable.
\newblock In {\em Proceedings of the Twenty-Ninth Annual ACM-SIAM Symposium on
  Discrete Algorithms}, SODA '18, pages 2728--2746, Philadelphia, PA, USA,
  2018. Society for Industrial and Applied Mathematics.

\bibitem{Delgosha2014}
Payam Delgosha and Salman Beigi.
\newblock Impossibility of local state transformation via hypercontractivity.
\newblock {\em Communications in Mathematical Physics}, 332(1):449--476, Nov
  2014.

\bibitem{FJVYuen:2019}
Joseph Fitzsimons, Zhengfeng Ji, Thomas Vidick, and Henry Yuen.
\newblock Quantum proof systems for iterated exponential time, and beyond.
\newblock In {\em Proceedings of the 51st Annual ACM SIGACT Symposium on Theory
  of Computing}, STOC 2019, New York, NY, USA, 2019. ACM.

\bibitem{10.1145/2688073.2688094}
Joseph Fitzsimons and Thomas Vidick.
\newblock A multiprover interactive proof system for the local hamiltonian
  problem.
\newblock In {\em Proceedings of the 2015 Conference on Innovations in
  Theoretical Computer Science}, ITCS '15, page 103–112, New York, NY, USA,
  2015. Association for Computing Machinery.

\bibitem{Ghazi:2018:DRP:3235586.3235614}
Badih Ghazi, Pritish Kamath, and Prasad Raghavendra.
\newblock Dimension reduction for polynomials over gaussian space and
  applications.
\newblock In {\em Proceedings of the 33rd Computational Complexity Conference},
  CCC '18, pages 28:1--28:37, Germany, 2018. Schloss Dagstuhl--Leibniz-Zentrum
  fuer Informatik.

\bibitem{7782969}
Badih Ghazi, Pritish Kamath, and Madhu Sudan.
\newblock Decidability of non-interactive simulation of joint distributions.
\newblock In {\em 2016 IEEE 57th Annual Symposium on Foundations of Computer
  Science (FOCS)}, pages 545--554, Los Alamitos, CA, USA, Oct 2016. IEEE
  Computer Society.

\bibitem{8242350}
L.~{Gyongyosi}, S.~{Imre}, and H.~V. {Nguyen}.
\newblock A survey on quantum channel capacities.
\newblock {\em IEEE Communications Surveys Tutorials}, 20(2):1149--1205,
  Secondquarter 2018.

\bibitem{10.1007/978-3-642-22006-7_8}
Aram~W. Harrow, Ashley Montanaro, and Anthony~J. Short.
\newblock Limitations on quantum dimensionality reduction.
\newblock In Luca Aceto, Monika Henzinger, and Ji{\v{r}}{\'i} Sgall, editors,
  {\em Automata, Languages and Programming}, pages 86--97, Berlin, Heidelberg,
  2011. Springer Berlin Heidelberg.

\bibitem{RevModPhys.81.865}
Ryszard Horodecki, Pawe\l{} Horodecki, Micha\l{} Horodecki, and Karol
  Horodecki.
\newblock Quantum entanglement.
\newblock {\em Rev. Mod. Phys.}, 81:865--942, Jun 2009.

\bibitem{Ito:2012:MIP:2417500.2417883}
Tsuyoshi Ito and Thomas Vidick.
\newblock A multi-prover interactive proof for {NEXP} sound against entangled
  provers.
\newblock In {\em Proceedings of the 2012 IEEE 53rd Annual Symposium on
  Foundations of Computer Science}, FOCS '12, pages 243--252, Washington, DC,
  USA, 2012. IEEE Computer Society.

\bibitem{Ji:2016:CVQ:2897518.2897634}
Zhengfeng Ji.
\newblock Classical verification of quantum proofs.
\newblock In {\em Proceedings of the Forty-eighth Annual ACM Symposium on
  Theory of Computing}, STOC '16, pages 885--898, New York, NY, USA, 2016. ACM.

\bibitem{JNVWY'20}
Zhengfeng Ji, Anand Natarajan, Thomas Vidick, John Wright, and Henry Yuen.
\newblock $\mathrm{MIP}^*=\mathrm{RE}$.
\newblock {\em arXiv preprint arXiv:2001.04383}, 2020.

\bibitem{JNVWYuen'20}
Zhengfeng Ji, Anand Natarajan, Thomas Vidick, John Wright, and Henry Yuen.
\newblock Quantum soundness of the classical low individual degree test.
\newblock {\em arXiv preprint arXiv:2009.12982}, 2020.

\bibitem{Johnston_2016}
Nathaniel Johnston, Rajat Mittal, Vincent Russo, and John Watrous.
\newblock Extended non-local games and monogamy-of-entanglement games.
\newblock {\em Proceedings of the Royal Society A: Mathematical, Physical and
  Engineering Sciences}, 472(2189):20160003, may 2016.

\bibitem{7452414}
S.~{Kamath} and V.~{Anantharam}.
\newblock On non-interactive simulation of joint distributions.
\newblock {\em IEEE Transactions on Information Theory}, 62(6):3419--3435, June
  2016.

\bibitem{doi:10.1137/090772885}
J.~Kempe, O.~Regev, and B.~Toner.
\newblock Unique games with entangled provers are easy.
\newblock {\em SIAM Journal on Computing}, 39(7):3207--3229, 2010.

\bibitem{Kempe:2008:EGH:1470582.1470594}
Julia Kempe, Hirotada Kobayashi, Keiji Matsumoto, Ben Toner, and Thomas Vidick.
\newblock Entangled games are hard to approximate.
\newblock In {\em Proceedings of the 2008 49th Annual IEEE Symposium on
  Foundations of Computer Science}, FOCS '08, pages 447--456, Washington, DC,
  USA, 2008. IEEE Computer Society.

\bibitem{8119865}
Felix Leditzky, Nilanjana Datta, and Graeme Smith.
\newblock Useful states and entanglement distillation.
\newblock {\em IEEE Transactions on Information Theory}, 64(7):4689--4708,
  2018.

\bibitem{cj13-11}
Debbie Leung, Ben Toner, and John Watrous.
\newblock Coherent state exchange in multi-prover quantum interactive proof
  systems.
\newblock {\em Chicago Journal of Theoretical Computer Science}, 2013(11),
  August 2013.

\bibitem{cj10-01}
Ashley Montanaro and Tobias~J. Osborne.
\newblock Quantum boolean functions.
\newblock {\em Chicago Journal of Theoretical Computer Science}, 2010(1),
  January 2010.

\bibitem{MosselOdonnell:2010}
Elchanan Mossel, Ryan O'Donnell, and Krzysztof Oleszkiewicz.
\newblock Noise stability of functions with low influences: Invariance and
  optimality.
\newblock {\em Annals of Mathematics}, 171:295--341, Mar 2010.

\bibitem{8948641}
A.~Natarajan and J.~Wright.
\newblock {NEEXP} is contained in {MIP}.
\newblock In {\em 2019 IEEE 60th Annual Symposium on Foundations of Computer
  Science (FOCS)}, pages 510--518, Los Alamitos, CA, USA, nov 2019. IEEE
  Computer Society.

\bibitem{NC00}
Michael~A Nielsen and Isaac~L Chuang.
\newblock {\em Quantum computation and quantum information}.
\newblock Cambridge University Press, Cambridge, UK, 2000.

\bibitem{Odonnell08}
Ryan O'Donnell.
\newblock {\em Analysis of Boolean Functions}.
\newblock Cambridge University Press, Cambridge, UK, 2013.

\bibitem{qin2021nonlocal}
Minglong Qin and Penghui Yao.
\newblock Nonlocal games with noisy maximally entangled states are decidable.
\newblock {\em SIAM Journal on Computing}, 50(6):1800--1891, 2021.

\bibitem{10.1145/2799560}
Oded Regev and Thomas Vidick.
\newblock Quantum {XOR} games.
\newblock {\em ACM Trans. Comput. Theory}, 7(4), aug 2015.

\bibitem{RoueWZ:2022}
Cambyse Rou\'e, Melchior Wirth, and Haonan Zhang.
\newblock Quantum talagrand, kkl and friedgut's theorems and the learnability
  of quantum boolean functions.
\newblock {\em arXiv preprint arXiv:2209.07279}, 2209.07279.

\bibitem{7426395}
C.~J. {Stark} and A.~W. {Harrow}.
\newblock Compressibility of positive semidefinite factorizations and quantum
  models.
\newblock {\em IEEE Transactions on Information Theory}, 62(5):2867--2880,
  2016.

\bibitem{Touchette:2015:QIC:2746539.2746613}
Dave Touchette.
\newblock Quantum information complexity.
\newblock In {\em Proceedings of the Forty-seventh Annual ACM Symposium on
  Theory of Computing}, STOC '15, pages 317--326, New York, NY, USA, 2015. ACM.

\bibitem{PhysRevA.84.052328}
Guoming Wang.
\newblock Property testing of unitary operators.
\newblock {\em Phys. Rev. A}, 84:052328, Nov 2011.

\bibitem{Wat08}
John Watrous.
\newblock {\em Theory of Quantum Information}.
\newblock Cambridge University Press, Cambridge, UK, 2018.

\bibitem{zhang2019some}
Pingping Zhang.
\newblock On some inequalities related to positive block matrices.
\newblock {\em Linear Algebra and its Applications}, 576:258--267, 2019.

\end{thebibliography}
\end{document}